\documentclass[english]{article}
\RequirePackage[OT1]{fontenc}
\RequirePackage{graphicx}
\RequirePackage{amsthm}
\RequirePackage{amsmath}
\RequirePackage[authoryear]{natbib}
\RequirePackage[colorlinks,citecolor=blue,urlcolor=blue]{hyperref}

\usepackage{amssymb}
\usepackage{footnote}
\usepackage{subfig}

\theoremstyle{plain}
\newtheorem{thm}{Theorem}[section]
\numberwithin{equation}{section}
\numberwithin{figure}{section}

\makeatletter
\newcommand{\indep}{\perp \!\!\! \perp}

\newcommand{\lyxdot}{.}
\@ifundefined{showcaptionsetup}{}{%
 \PassOptionsToPackage{caption=false}{subfig}}
\makeatother

\begin{document}
	
\title{Nonparametric Instrumental Variables Estimation Under Misspecification}
\author{Ben Deaner, University College London\thanks{bendeaner@gmail.com.
We thank Whitney Newey, Anna Mikusheva and Jerry Hausman
for indispensable advice. We thank attendees of this paper's sessions in the 2019 North American and European Summer
Meetings of the Econometric Society. We thank anonymous referees for helpful comments. We gratefully acknowledge the support of the MIT economics department, the Cowles Foundation, and the UCL department of Economics.}}

\maketitle
\begin{abstract}
	
Nonparametric Instrumental Variables (NPIV) analysis is based on a conditional moment restriction. We show that if this moment condition is even slightly misspecified, say because instruments are not quite valid, then NPIV estimates can be subject to substantial asymptotic error and the identified set under a relaxed moment condition may be large. Imposing strong a priori smoothness restrictions mitigates the problem but induces bias if the restrictions are too strong. In order to manage this trade-off we develop a methods for empirical sensitivity analysis and apply them to the consumer demand data previously analyzed in \citet{Blundell2007} and \citet{Horowitz2011}.
\end{abstract}


Instrumental validity can be difficult to defend. The assumption rules out not only confounding between outcomes and instruments, but also any direct causal effect of the instruments on outcomes. Thus instruments may not be valid even if they are assigned by an ideal randomized experiment. To justify the use of instrumental variables (IV) methods, applied researchers usually argue not that instrumental validity holds exactly, but that any deviation from validity is small.\footnote{See \citet{Conley2008} for further discussion.} `Small' does not imply `inconsequential', and so it is important to assess the sensitivity of IV methods to a modest deviation from full instrumental validity.

In this work we consider the sensitivity of Nonparametric Instrumental Variables (NPIV) analysis  (\citet{Newey2003},
\citet{Ai2003}) to misspecification, particularly that which results from invalid instruments. NPIV generalizes the linear IV model to allow for a flexible, non-linear `structural function'. The structural function is the key object of interest and typically has a causal interpretation. Identification and estimation is based on a conditional moment restriction. If instruments are invalid then the moment condition generally does not hold.


We show that identification in NPIV models is fragile. If only weak a priori conditions are placed on the structural function, then the identified set can be very large under only a slight relaxation of the NPIV moment condition. Estimation in NPIV models can be highly non-robust to misspecification. For estimators that impose only weak restrictions on the structural function,
an arbitrarily small deviation from instrumental validity can impart
a large asymptotic bias. This sensitivity is a consequence of the `ill-posedness' of the NPIV  moment restriction.

If researchers place sufficiently strong restrictions (typically smoothness conditions) on the structural function, then robust identification is possible. Estimation that is robust to misspecification can be achieved by constraining the estimates so that they satisfy such smoothness conditions. In both cases, the sensitivity to misspecification is still greater in a certain sense than for standard nonparametric regression or parametric IV. Imposing strong smoothness restrictions in estimation carries a cost. If the true structural function does not obey these restrictions then an estimator that imposes them is necessarily inconsistent.

Two NPIV estimation methods which impose sufficient smoothness conditions are the procedures of \citet{Newey2003} and
\citet{Blundell2007}. A number of prominent NPIV procedures do not constrain the estimates to be smooth. For example, the methods described in \citet{Chen2018}, \citet{Darolles},
\citet{Hall2005}, and \citet{Horowitz2011}.

Robust estimation is also possible if the researcher is interested in a continuous
functional of the structural function rather than the function itself. We provide necessary and sufficient conditions under
which such a functional can be estimated robustly without the need to impose smoothness.

Our results suggest that researchers face a trade-off: impose too little smoothness and NPIV estimates are non-robust to invalid instruments, impose too much smoothness and the estimates are inconsistent even if instruments are valid. In response, we supplement NPIV estimation with a method for empirical sensitivity analysis. Methods for sensitivity analysis are increasingly popular in applied economics and include local methods (e.g., \citet{Andrews2017a}, \citet{Armstrong}, \citet{Weidner}), and global methods (e.g., \citet{Altonji2005a}, \citet{Masten2018}, \citet{Oster2019})). We develop a technique for global sensitivity analysis in which the researcher estimates the identified set under relaxations of the NPIV moment condition and various smoothness assumptions. Thus researchers can assess which of their findings are robust to some misspecification of the NPIV moment condition and examine the trade-off between robustness and smoothness. Estimation of the identified set under weakened assumptions is used elsewhere for sensitivity analysis, for example in \citet{Masten2018}. The set estimation problem is well-posed, and we derive convergence rates under primitive conditions.

We apply our methods to the empirical setting shared by \citet{Blundell2007}
and \citet{Horowitz2011} who use NPIV methods to
estimate shape-invariant Engel curves using data from The British
Family Expenditure Survey. We argue that in this setting instruments are unlikely to be fully exogenous, in which case the NPIV moment restriction will not hold exactly. We use our methodology to assess which
features of the structural Engel curve for food can be inferred robustly.

\subsection*{A Note on Notation}

For a random variable $W$, the corresponding calligraphic letter $\mathcal{W}$ denotes its support and $\mu_W$ its probability distribution. If $W$ and $V$ are random variables then statements of the form $W=V$ or $W\leq V$ should be understood to hold with probability $1$. If $W$ is a random variable and $h$ a function on $\mathcal{W}$ then $|h|_\infty$ is the essential supremum of $|h(W)|$, we refer to this as the `sup norm'. If $\mathcal{S}$ is a subset of a topological space then $\bar{\mathcal{S}}$ is its closure and $int(\mathcal{S})$ its interior. For an operator $\mathbb{T}:\,\mathcal{B}_1\to\mathcal{B}_2$, if $\mathcal{S}\subseteq\mathcal{B}_1$ then $\mathbb{T}[\mathcal{S}]$ is the image of $\mathcal{S}$ under $\mathbb{T}$ and if $\mathcal{S}\subseteq\mathcal{B}_2$ then $\mathbb{T}^{-1}[\mathcal{S}]$ is the pre-image of $\mathcal{S}$. If $\mathcal{B}$ is a normed space then we denote its norm by $\|\cdot\|_\mathcal{B}$. If $\mathcal{B}$ is a set of functions with domain $\mathcal{W}$ then $\mathcal{B}^{d}$ contains all length-$d$ vectors of functions $r$ such that $r(w)=(h_1(w),h_2(w),...,h_d(w))$ for all $w\in\mathcal{W}$, where $h_1,h_2,...,h_d\in\mathcal{B}$. For a real vector  $v$ we let $\|v\|_p$ denote the $\ell_p$ norm of $v$.

\section{Analytical Sensitivity Results}

\subsection{Background and Assumptions}

NPIV estimation is a flexible, nonparametric alternative to linear IV. It has been studied extensively, for example in \citet{Newey2003}, \citet{Ai2003}, \citet{Cheng},
\citet{Darolles}, \citet{Hall2005}, \citet{Horowitz2011}, and \citet{Chen2018}. NPIV methods have been applied in empirical work by \cite{Blundell2007}, \cite{Chen2018}, and \cite{Horowitz2011} among others.

NPIV models relax the linearity assumption of linear IV but retain additive separability. Let $Y$ be an outcome of interest, $X$ an endogenous variable, and $Z$ an instrument. NPIV analysis assumes $Y=h_0(X)+\varepsilon$ where $h_0$ is a non-random `structural function' and the structural residual $\varepsilon$ satisfies $E[\varepsilon|Z]=0$. NPIV is `nonparametric' because the structural function is not assumed to take any particular parametric form.

The structural function $h_0$ is the key object of interest and typically has a causal interpretation. The use of instrumental variables is motivated by the policy-relevance of $h_0$. If  $\varepsilon$ is correlated with $X$ then standard nonparametric regression will not recover $h_0$ but rather a reduced-form parameter which may not be useful for analyzing the effects of a policy.

We can obtain an NPIV model by imposing an additive structure for potential outcomes and assuming instruments are valid in the sense of \cite{Angrist}. In particular, instruments are valid if they have no direct causal effect on outcomes and they are randomly assigned. Let $Y_{x,z}$ be the potential outcome from a counterfactual treatment level $x$ and counterfactual value of the instrument $z$. An additive model for  $Y_{x,z}$ is given below.
\begin{equation}Y_{x,z}=h_0(x)+\varepsilon \label{potent}
	\end{equation} 

The random variable $\varepsilon$ captures all heterogeneity in potential outcomes and without loss of generality $E[\varepsilon]=0$. The instrument has no direct effect on outcomes and so $Y_{x,z}$ does not depend on $z$, that is, $Y_{x,z}=Y_{x}$. If the instrument is randomly assigned then $E[\varepsilon|Z]=0$. In this case the NPIV model holds with $h_0$ the average potential outcome from treatment level $x$, that is $h_0(x)=E[Y_{x}]$.\footnote{In fact, we can allow for non-additive heterogeneity in potential outcomes so long as this heterogeneity is fully exogenous. Formally, suppose $Y_{x,z}=l_0(x,V)+\varepsilon$ where $l_0$ is a non-random function, $V\indep(X,Z)$, and as before $E[\varepsilon|Z]=0$. Then the NPIV model holds with $h_0(x)=E[Y_{x}]$.}

Assumption 1.1 below formally states the key NPIV identifying moment restriction. It is equivalent to the model $Y=h_0(X)+\varepsilon$ with $E[\varepsilon|Z]=0$.

\theoremstyle{definition}
\newtheorem*{A1.1}{Assumption 1.1 (Correct Specification)}
\begin{A1.1}
$E[Y-h_{0}(X)|Z]=0$
\end{A1.1}

We sometimes wish to place a priori restrictions on the structural function. To achieve this we can assume $h_0$ belongs to a restrictive subset $\mathcal{H}$ of the underlying space of functions $\mathcal{B}_X$.

\theoremstyle{definition}
\newtheorem*{A1.2}{Assumption 1.2 (Parameter Space)}
\begin{A1.2}
$h_0\in \mathcal{H} \subseteq \mathcal{B}_X$.
\end{A1.2}

Assumptions 1.1 and 1.2 are structural assumptions in the sense that one cannot confirm that they hold by looking at the joint distribution of the observables. The non-structural Assumptions 1.3 and 1.4 below are imposed throughout the NPIV literature. They restrict the joint distribution of $X$ and $Z$.

\theoremstyle{definition}
\newtheorem*{A1.3}{Assumption 1.3  (Completeness)}
\begin{A1.3}
For any $h\in\mathcal{B}_X$, $E[h(X)|Z]=0$ if and only if $h(X)=0$.
\end{A1.3}
\newtheorem*{A1.4}{Assumption 1.4  (Compactness)}
\begin{A1.4}
Define the linear operator $\mathbb{A}$ by the equation $\mathbb{A}[h](Z)=E[h(X)|Z]$. $\mathbb{A}$ is a compact infinite-dimensional linear operator from the Banach space $\mathcal{B}_{X}$
into a Banach space $\mathcal{B}_{Z}$.
\end{A1.4}

Assumption 1.3 imposes a type of statistical completeness. Completeness plays a role analogous to the rank condition for identification in linear IV (\citet{Newey2003}).\footnote{For some
work on statistical completeness see \citet{Andrews2017b}, \citet{Canay2013}, \citet{Chena},
\citet{DHaultfoeuille2011}, \citet{Freyberger2017a}, and \citet{Hu2018}.} If Assumption 1.3 holds then Assumptions 1.1 point identifies $h_0$.

Assumption 1.4 is a technical condition that allows us to apply powerful results from functional analysis to the NPIV estimation problem. For common choices of ${\mathcal{B}_X}$ and ${\mathcal{B}_Z}$, Assumption 1.4 holds under weak primitive conditions on the joint distribution of $X$ and $Z$ (\citet{Florens2011}, \citet{Horowitz2011}). However, the assumption rules out the case of $X=Z$ which corresponds to standard non-parametric regression.\footnote{When $X=Z$ and $\mathcal{B}_X=\mathcal{B}_Z$, $\mathbb{A}$ is the identity and cannot be compact unless $\mathcal{B}_X$ is finite-dimensional. See Theorem 2.20 in \cite{Kress2014}.}

As a leading example we consider $\mathcal{B}_X$ and $\mathcal{B}_Z$ to be spaces of continuous functions equipped with the sup norm $|\cdot|_\infty$. In this case a sufficient condition for Assumption 1.4 is that $X$ and $Z$ have compact supports and admit a strictly positive continuous joint probability density.\footnote{This implies the conditional density of $X$ given $Z$ is continuous and so we can apply Theorem 2.21 in \cite{Kress2014}.}

While our results apply for quite general choices of the underlying function spaces $\mathcal{B}_X$ and $\mathcal{B}_Z$, Assumption 1.5 below places some mild restrictions on these spaces. The condition ensures that means of functions in these spaces are finite. We also assume that the outcome has a finite mean.

\newtheorem*{A1.5}{Assumption 1.5  (First Moments)}
\begin{A1.5}
	$E[|Y|]<\infty$, $E[|h(X)|]<\infty,\forall h\in\mathcal{B}_X$, and $E[|g(Z)|]<\infty,\forall g\in\mathcal{B}_Z$. The mapping from an element $h\in\mathcal{B}_X$ to its mean $E[h(X)]$ is continuous.
\end{A1.5}

Let us define a function $g_0\in\mathcal{B}_Z$ by $g_0(Z)=E[Y|Z]$. Using $\mathbb{A}$ defined in Assumption 1.4, we can rewrite the moment condition in Assumption 1.1 as an equation $\mathbb{A}[h_{0}]=g_{0}$. If Assumption 1.3 holds then this equation has a unique solution $h_{0}=\mathbb{A}^{-1}[g_{0}]$. The quantity on the right-hand side depends only on the distribution of the observables, and so $h_0$ is identified. If the NPIV moment condition is misspecified then  $h_0$ need not equal $\mathbb{A}^{-1}[g_{0}]$, and so we refer to the latter as the `pseudo-solution'.

NPIV estimation is often said to be `ill-posed'. This refers to the fact that under Assumptions 1.3 and 1.4, the pseudo-solution varies discontinuously with $g_0$. In fact, if we perturb $g_{0}$ by some arbitrarily small amount, this can induce an arbitrarily large change in the pseudo-solution. This is problematic for estimation because $g_0$ and $\mathbb{A}$ must be empirically estimated, and thus are subject to error. If we substitute these estimates into the pseudo-solution then ill-posedness suggests the resulting estimate of $h_0$ can be highly inaccurate, even if $g_0$ and $\mathbb{A}$ are precisely estimated.

To tackle this problem researchers have proposed numerous `regularization' methods. Regularization involves replacing the discontinuous operator $\mathbb{A}^{-1}$ in the pseudo-solution with a continuous approximation. This generally imparts bias and so the degree of regularization is reduced as the sample size grows.\footnote{Available regularization methods including Tikhonov regularization, projection onto a finite-dimensional sieve space, and many others. See \citet{Darolles} for discussion.} Reducing the strength of regularization increases sensitivity to error in the estimates of $g_0$ and $\mathbb{A}$, but this error typically decreases with the sample size. If the regularization is relaxed sufficiently slowly then the reduction in error more than offsets the increased sensitivity.

A key insight in this work is that misspecification, due to say invalid instruments, perturbs $g_0$ away from the value it would take under correct specification. The pseudo-solution is highly sensitive to perturbations due to misspecification just as it is to estimation error. NPIV estimators are typically designed to converge in probability to the pseudo-solution (\cite{Ai2007}),  estimates that converge to the pseudo-solution under weak conditions are thus highly sensitive to misspecification, at least asymptotically. A similar sensitivity result applies for the identified set under a slight relaxation of the NPIV moment condition. To the best of our knowledge this point is entirely absent from the rest of the NPIV literature. The precise consequences for identification and estimation depend crucially on whether or not strong a priori conditions are imposed on the structural function. Strong restrictions can shrink the identified set and limit the conditions under which an NPIV estimator converges to the pseudo-solution.

\subsection{Introducing Misspecification}

The NPIV model is misspecified if Assumption 1.1 does not hold. That is, ${u}_0(Z)$ defined below is non-zero with positive probability:
\begin{equation}
E[Y-h_{0}(X)|Z]={u}_0(Z) \label{udef}
\end{equation}

\theoremstyle{definition}
\newtheorem*{E1}{Example 1: Direct Effect of the Instrument}
\begin{E1}

Suppose we relax the model for potential outcomes (\ref{potent}) so that the instrument
may have a direct and additively separable effect on outcomes:

\[
Y_{x,z}=h_{0}(x)+{u}_{0}(z)+\varepsilon
\]

Without loss of generality we suppose $E[\varepsilon]=0$ and $E[u_{0}(Z)]=0$ so that $h_{0}(x)=E[Y_{x}]$. We assume as before that the instrument
is randomly assigned so that $E[\varepsilon|Z]=0$. This model implies (\ref{udef}).
\end{E1}

\theoremstyle{definition}
\newtheorem*{E2}{Example 2: Confounded Instrument}
\begin{E2}
Suppose that (\ref{potent}) holds with $E[\varepsilon]=0$, but the instrument is not randomly assigned. Instead, $Z$ and $\epsilon$ are statistically associated due to their shared dependence on a latent factor $V$ and assume $Z\indep \epsilon|V$.  Let $f_{\varepsilon|V}$ and $f_{V|Z}$ be the conditional probability densities of $\varepsilon$ given $V$, and of $V$ given $Z$ respectively. In this case 
(\ref{udef}) holds with $h_{0}(x)=E[Y_{x}]$ and ${u}_{0}$ given below:
\[
{u}_{0}(z)=\int\int e f_{\varepsilon|V}(e|v)def_{V|Z}(v|z)dv
\]
\end{E2}

\theoremstyle{definition}
\newtheorem*{E3}{Example 3: Nonseparable IV Model}
\begin{E3}
Consider the fully nonseparable potential outcomes model $Y_{xz}=l_{0}(x,\varepsilon)$. As before, instruments
have no direct effect on outcomes. Let $X_{z}=r_{0}(z,V)$ be the potential treatment
from a counterfactual level of the instrument $z$. In this nonseparable
model random assignment of the instrument can be formalized as $Z\indep(\varepsilon,V)$.
 Suppose $\varepsilon$ and $V$ admit a joint density $f_{\varepsilon V}$ and
marginal densities $f_{\varepsilon}$ and $f_{V}$. Let $h_{0}(x)=E[Y_{x}]$ then (\ref{udef}) holds with ${u}_{0}$ given below:
\[
{u}_{0}(z)=\int\int l_{0}(r_{0}(z,v),e)\big(f_{\varepsilon V}(e,v)-f_{\varepsilon}(e)f_{V}(v)\big)dedv
\]
\end{E3}

\theoremstyle{definition}
\newtheorem*{A1.1alt}{Assumption $1.1^*$ (Nearly Correct Specification)}
\begin{A1.1alt}
$E[Y-h_{0}(X)|Z]={u}_0(Z)$ where ${u}_0\in\mathcal{U}\subseteq \mathcal{B}_Z$ and $\|{u}_0\|_{\mathcal{B}_Z}\leq b$.
\end{A1.1alt}

The scalar $b$ in Assumption $1.1^*$ controls the degree of misspecification. If $b=0$ then  Assumption 1.1 holds and so the NPIV moment condition is correctly specified. If $b$ is small but non-zero then the two sides of the moment condition in Assumption 1.1 need not be equal, but they must be close to equal, with the distance between them measured by $\|\cdot\|_{\mathcal{B}_Z}$. If  $\mathcal{B}_Z$ is equipped with the sup norm then $\|{u}_0\|_{\mathcal{B}_Z}\leq b$ is equivalent to $|E[Y-h_{0}(X)|Z]|\leq b$.

Other choices of norms correspond to stronger or weaker notions of misspecification. For example if we take $\|\cdot\|_{\mathcal{B}_Z}$ to be the $L_2(\mu_Z)$ norm then we restrict the second moment of ${u}_0(Z)$ to be less than $b$, which is weaker than the corresponding restriction on the sup norm.

Bounds on the norm of $u_0$ can be derived from more primitive conditions. Consider Example 1 above and let $\mathcal{B}_Z$ be a set of continuous functions equipped the sup norm. The direct causal effect of changing the instrument from $z_1$ to $z_2$ is $u_0(z_2)-u_0(z_1)$. Suppose that for all $z_1,z_2$ the magnitude of this effect is at most $b$, given $E[u_0(Z)]=0$ it follows that $\|u_0\|_{\mathcal{B}_Z}\leq b$. Conversely, $u_0(z_2)-u_0(z_1)$ must be bounded above by $2\|u_0\|_{\mathcal{B}_Z}$. Thus the norm of $u_0$ is small if and only if the instrument has at most a small direct impact on outcomes.

In Example 2 one can show the sup norm of $u_0$ must be less than the essential supremum of   $2|E[\varepsilon|V]-E[\varepsilon]|$. This quantity measures the strength of the dependence between the structural residual $\varepsilon$ and the unobserved confounder $V$. In Example 3 one can derive the following bound:
\[\sup_{z\in\mathcal{Z}}|u_0(z)|\leq 2\inf_{h,q\in\mathcal{B}_{X}}\sup_{x\in\mathcal{X},e\in\mathcal{E}}|l_{0}(x,e)-h(x)-q(e)|\]
The expression on the right-hand side measures how closely $l_0$ can be approximated by an additively separable function. Thus a  bound the approximation error implies a bound on the norm of $u_0$.

The restriction that $u_0\in \mathcal{U}$ allows us to incorporate additional conditions on $u_0$. In Examples 1 and 2, but not necessarily Example 3,  $E[u_0(Z)]=0$ by construction, and so we may take $\mathcal{U}$ to contain only functions with zero-mean. $\mathcal{U}$ can also incorporate smoothness restrictions. In Example 1 $u_0$ is smooth if, loosely speaking, small changes in $Z$ have small direct causal effects on the outcome.

In order to accommodate cases in which $\mathcal{U}$ contains only zero-mean functions it is helpful to define the subspaces $\tilde{\mathcal{B}}_X$ and $\tilde{\mathcal{B}}_Z$ so that $\tilde{\mathcal{B}}_X$ contains all $h\in{\mathcal{B}}_X$ with $E[h(X)]=0$ and $\tilde{\mathcal{B}}_Z$ contains all $g\in{\mathcal{B}}_Z$ with $E[g(Z)]=0$.

Our first result concerns the size of the identified set under the a priori assumptions $1.1^*$ and $1.2$. The identified set $\Theta_b$ is the set of functions that are consistent with these two assumptions:
\begin{equation}
\Theta_b=\big[h\in\mathcal{H}:\, g_0-\mathbb{A}[h]\in\mathcal{U} \text{ and } \|g_0-\mathbb{A}[h]\|_{\mathcal{B}_Z}\leq b\big] \label{idset}
\end{equation}
Theorem 1.1 concerns the diameter of the identified set. The diameter, denoted $diam(\Theta_b)$ is defined as follows:
\[
diam(\Theta_b)=\sup_{h_1, h_2 \in \Theta_b}\|h_1-h_2\|_{\mathcal{B}_X}
\]

\theoremstyle{plain}
\begin{thm}
	Suppose Assumptions 1.3-1.5 hold and ${\mathbb{A}^{-1}[\mathcal{U}]}$ contains an open $\tilde{\mathcal{B}}_X$-ball centered at zero of radius $r_u$.\footnote{An open $\tilde{\mathcal{B}}_X$-ball of radius $r$ is the intersection of $\tilde{\mathcal{B}}_X$ and an open ball in ${\mathcal{B}}_X$ of radius $r$.} 
	
	\textbf{a.}
	If $\mathbb{A}^{-1}[g_0]\in int(\bar{\mathcal{H}})$ then $\lim_{b\to0} diam(\Theta_b)> 0$. 	More precisely, for all $b>0$, $diam(\Theta_b)\geq 2\min\{r_h,r_u\}$, where $r_h$ is the radius of the largest open ball in $\bar{\mathcal{H}}$ centered at $\mathbb{A}^{-1}[g_0]$.
	
	\textbf{b.}
	If $\mathcal{H}$ is compact and $\mathbb{A}^{-1}[g_0]\in \mathcal{H}$, then  $\lim_{b\to0} diam(\Theta_b)= 0$. If in addition $\mathcal{H}$ is absolutely convex and infinite-dimensional, and  $\frac{1}{\alpha}\mathbb{A}^{-1}[g_0]\in\mathcal{H}$ for some $\alpha\in(0,1)$, then $\lim_{b\to0} \frac{diam(\Theta_b)}{b}=\infty$.\footnote{A subset $\mathcal{S}$ of a real vector space is `absolutely convex' if for any $h_1,h_2\in\mathcal{S}$ and $a,b\in\mathbb{R}$,  $|a|+|b|\leq 1$ implies $ah_1+bh_2\in\mathcal{S}$.}
	
\end{thm}

Theorem 1.1 examines the diameter of the identified under different degrees of misspecification. The behavior of the set for small values of $b>0$ depends crucially on the strength of the a priori restrictions on the structural function, as captured in $\mathcal{H}$. Part a. of the theorem requires that the pseudo-solution lies in the interior of the closure of $\mathcal{H}$. Some spaces of smooth functions are sufficiently restrictive that they are compact, in which case the interior is empty and part a. cannot apply.\footnote{A compact subset of an infinite-dimensional Banach space must be closed and have an empty interior.} Part b. provides results for these more restrictive parameter spaces.

Under the conditions of part a., the diameter of the identified set does not shrink to zero with $b$. Instead it remains bounded below by a constant that depends on $\mathcal{H}$ and $\mathcal{U}$. Consider the extreme case in which $\mathcal{H}$ is dense in $\mathcal{B}_X$ and $\mathcal{U}$ contains an open ball centered at zero. Then diameter of the identified set is infinite, regardless of how small we make the bound $b>0$.\footnote{If $\mathcal{U}$ contains an open $\tilde{\mathcal{B}}_Z$-ball at zero then for sufficiently small $b>0$ we can replace $\mathcal{U}$ with $\tilde{\mathcal{B}}_Z$ (so $r_u=\infty$) and the resulting identified set contains the original as a subset.} Under the conditions of part b., the identified set shrinks to a point as the degree of misspecification is reduced. However, the diameter of the identified set shrinks to zero strictly more slowly than the bound $b$.

Let us compare to standard non-parametric regression. This case corresponds to the NPIV moment condition with $X=Z$. Let $\mathcal{B}_X=\mathcal{B}_Z$. Recall that Assumption 1.4 cannot hold in this setting and so Theorem 1.1 does not apply. In this case the identified set has diameter at most $2b$, regardless of the choice of $\mathcal{H}$ and $\mathcal{U}$. The identified set shrinks to zero at exactly the same rate as $b$. This is strictly faster than the rate at which the identified set shrinks in the NPIV case, even under the strong restrictions on $\mathcal{H}$ imposed by part b. of Theorem 1.1.

We now specify spaces $\mathcal{H}$ and $\mathcal{U}$ compatible with parts a. and b. of the theorem. For concreteness let $\mathcal{B}_X$ be the set of continuous functions with the sup norm and assume $X$ has compact support.

First take $\mathcal{H}$ to be the set of Lipschitz continuous functions whose magnitude is bounded by a known constant $\bar{c}<\infty$:
\begin{equation}
\mathcal{H}=\big[h\in\mathcal{B}_X:\, \sup_{ x_1\neq x_2\in \mathcal{X}}\frac{|h(x_1)-h(x_2)|}{\|x_1-x_2\|_2}<\infty,\, \sup_{x\in\mathcal{X}}|h(x)|<\bar{c} \big]\label{lipz2}
\end{equation}
By the Stone-Weirstrauss theorem the set of Lipschitz functions is dense $\mathcal{B}_X$. Thus this set satisfies the conditions in part a. of Theorem 1.1 with $r_h=\bar{c}$.

Now suppose we further restrict $\mathcal{H}$ so that it contains only Lipschitz continuous functions with Lipschitz constant less than $C<\infty$:
\begin{equation}
\mathcal{H}=\big[h\in\mathcal{B}_X:\, \sup_{ x_1\neq x_2\in \mathcal{X}}\frac{|h(x_1)-h(x_2)|}{\|x_1-x_2\|_2}\leq  C,\, \sup_{x\in\mathcal{X}}|h(x)|\leq \bar{c}\big]\label{lipz1}
\end{equation}
The closure of this set does not contain an open ball so we cannot apply part a. of Theorem 1.1. In fact, the set is compact (see \cite{Freyberger2017}) and it is absolutely convex and infinite-dimensional. Thus this choice of $\mathcal{H}$ satisfies the restriction in part b. of Theorem 1.1. Part b. also requires that the pseudo-solution $\mathbb{A}^{-1}[g_0]$ be in $\mathcal{H}$ and that $\frac{1}{\alpha}\mathbb{A}^{-1}[g_0]\in\mathcal{H}$ for some $\alpha\in(0,1)$. Given the choice of $\mathcal{H}$ above, this holds if and only if $\mathbb{A}^{-1}[g_0]$ is Lipschitz with constant strictly less than $C$ and bounded in magnitude by a constant strictly less than $\bar{c}$.

Compact parameter spaces like (\ref{lipz1}) are employed in many NPIV papers. For example they appear in \citet{Newey2003}, \citet{Ai2003}, \citet{Blundell2007},
\citet{Freyberger2017a}, and \citet{Santos2012}.

The distinction between the sets (\ref{lipz2}) and (\ref{lipz1}) may seem subtle. However, (\ref{lipz1}) encodes much stronger knowledge about the structural function $h_0$. Not only do we know $h_0$ is Lipschitz continuous, we know it is Lipschitz continuous with a Lipschitz constant smaller than a known scalar $C$. A similar distinction holds for other smoothness conditions. For example, in Section 2 we consider a space of functions whose second derivatives are bounded. In order to achieve an identified set that shrinks to zero with $b$ we must be willing to assume a \textbf{specific} bound on the second derivatives of $h_0$. It is not enough to simply assume that there exists some unknown finite bound on the second derivatives.\footnote{For more examples of infinite-dimensional sets of compact functions see \cite{Freyberger2017}. All examples in that paper also satisfy the absolute convexity requirement.}

The condition on $\mathcal{U}$ is the same for both parts of the Theorem, it states that  ${\mathbb{A}^{-1}[\mathcal{U}]}$ contains an open $\tilde{\mathcal{B}}_X$-ball centered at zero of radius $r_u$. This condition is relatively weak in that it can hold even if $\mathcal{U}$ contains only very smooth functions.  Suppose we impose the same strong Lipschitz condition on $u_0$ as in (\ref{lipz1}):
\begin{equation}
	\mathcal{U}=\big[u\in\tilde{\mathcal{B}}_Z:\, \sup_{ z_1\neq z_2\in \mathcal{Z}}\frac{|u(z_1)-u(z_2)|}{\|z_1-z_2\|_2}\leq C,\,\sup_{z\in\mathcal{Z}}|u(z)|\leq\bar{c} \big] \label{Udef}
\end{equation}
Although this set is restrictive, ${\mathbb{A}^{-1}[\mathcal{U}]}$ can contain an open ball at zero. The following is a sufficient (but not necessary) condition for such a ball to exist given the choice of $\mathcal{U}$ above. Let $f_{X|Z}$ be the conditional probability density of $X$ given $Z$. Suppose $f_{X|Z}$ is well-defined and Lipschitz continuous in its $Z$ argument:
\[\sup_{x\in\mathcal{X}}\sup_{ z_1\neq z_2\in \mathcal{Z}}\frac{|f(x|z_1)-f(x|z_2)|}{\|z_1-z_2\|_2}<\infty \]
Denote the quantity on the right-hand side above by $c$. Then $\mathbb{A}^{-1}[\mathcal{U}]$ contains an open $\tilde{\mathcal{B}}_X$-ball of radius $\min\{\bar{c},C/c\}$.

\subsubsection*{The Sensitivity of Estimators of the Structural Function}

The sensitivity of a given NPIV estimator to misspecification may be substantially worse than is suggested by Theorem 1.1. The reason for this is that under misspecification the probability limit of a particular NPIV estimate may be outside of the identified set. We show below that the sensitivity of an NPIV estimator to misspecification depends crucially on the conditions under which it is consistent under correct specification. In particular, an NPIV estimator that achieves consistency under only weak a priori conditions on the structural function is necessarily very sensitive to misspecification. Thus in order to obtain estimates that are robust to misspecification, one must forgo consistent estimation whenever the true structural function lies outside a restrictive parameter space.

In order to capture this formally we introduce a notion of asymptotic bias. In particular, 
we consider the largest possible (un-scaled) asymptotic bias under Assumptions $1.1^*$, 1.2, 1.3, and 1.4. We fix parameters
that are not directly related to misspecification. That is, we fix the structural function $h_{0}$ and $\mu_{XZ\eta}$ the joint probability
distribution of $X$, $Z$, and $\eta\equiv Y-E[Y|Z]$. The worst-case asymptotic bias of an estimator $\hat{h}_{n}$ for a given $b$ is then:
\[
bias_{\hat{h}_{n}}(b)=\sup_{u_{0}\in \mathcal{U}:\,\|u_{0}\|_{\mathcal{B}_{Z}}\leq b}\inf \big[\epsilon : P(\|\hat{h}_{n}-h_{0}\|_{\mathcal{B}_{X}}\leq \epsilon)\to 1 \big]
\] 
The above is a version of the `maximum
bias' discussed in \citet{Huber2011}. The definition above does not require that $\hat{h}_n$ has a probability limit.\footnote{The existence of a probability limit for an NPIV estimator is difficult to establish under misspecification, sufficient conditions for the existence of a probability limit for misspecified sieve minimum distance estimators can be found in \citet{Ai2007}.} If $\hat{h}_n$ converges in probability to a limit $\tilde{h}$ then the infimum in the definition above is equal to  $\|\tilde{h}-h_{0}\|_{\mathcal{B}_{X}}$.

\theoremstyle{plain}
\begin{thm}
	Fix $\mu_{XZ\eta}$ so that Assumptions 1.3-1.5 hold. Let $\mathcal{S}\subseteq\mathcal{B}_X$ be the set of all functions $h$ so that if $h_0=h$ and Assumption 1.1 holds then  $\text{plim}_{n\to\infty}\|\hat{h}_{n}-h_{0}\|_{\mathcal{B}_{X}}=0$. Suppose $\bar{\mathcal{S}}$ contains an open ball centered at $h_0$ of radius $r_s$ and ${\mathbb{A}^{-1}[\mathcal{U}]}$ contains an open $\tilde{\mathcal{B}}_X$-ball centered at zero of radius $r_u$.
	
	 Then $\lim_{b\to 0} bias_{\hat{h}_{n}}(b)>0$. More precisely, for all $b>0$, $bias_{\hat{h}_{n}}(b)\geq \min\{r_s,r_u\}$.
\end{thm}

Theorem 1.2 shows that there is a trade-off between consistency under correct specification and robustness to misspecification. Suppose that under correct specification and some restrictions on the distribution of observables, an NPIV estimator is consistent whenever $h_0\in\mathcal{S}$. Ideally $\mathcal{S}$ would be large so that consistency does not require strong a priori assumptions on the structural function. However, if $\mathcal{S}$ is large (in the sense given in the theorem), then the estimator must be very sensitive to certain kinds of misspecification. This is true even if the structural function $h_0$ in fact lies in a set $\mathcal{H}$ that is much more restrictive than $\mathcal{S}$.

Theorem 1.2 can be used to derive the sensitivity properties of particular NPIV estimators. To apply the theorem we can use existing asymptotic results to establish consistency over a sufficiently large space $\mathcal{S}$. In Appendix A1 we verify the key conditions of the theorem for some highly-cited NPIV estimators.

Some NPIV estimators are constructed in such a way that Theorem 1.2 cannot apply. These estimators are constrained so that $\hat{h}_n$ is always an element of a restrictive parameter space like (\ref{lipz1}) whose closure has an empty interior. For example, let $\hat{Q}_n$ be an empirical objective function that may change with the sample size (for example, to incorporate a penalty function with decreasing strength) and let $\mathcal{S}_n$ be a finite-dimensional subset of $\mathcal{S}$ that can grow with the sample size. Many NPIV estimators take the form below, including those in \cite{Newey2003} and \cite{Blundell2007}.
\[
\hat{h}_n=\arg\min_{h\in\mathcal{S}_n}\hat{Q}_n(h) 
\]

Because $\mathcal{S}_n$ is a subset of  $\mathcal{S}$, the estimate $\hat{h}_n$ above must be an element of $\mathcal{S}$ for all $n$. Such an estimator is necessarily inconsistent if $h_0\notin \mathcal{S}$. If $\mathcal{S}$ is sufficiently restrictive, for example if it takes the form (\ref{lipz2}), then the sensitivity result in Theorem 1.2 cannot apply. Thus such estimators forgo consistency outside of a restrictive parameter space but have the advantage that they potentially avoid the non-robustness to misspecification suggested by the theorem. If the researcher is certain that $h_0$ really does satisfy the restrictions imposed by inclusion in $\mathcal{S}$, then there is no cost to inconsistency outside of $\mathcal{S}$.

Our next result establishes conditions on an estimator that ensure it is robust in the sense that  $\lim_{b\to 0} bias_{\hat{h}_{n}}(b)=0$. However, analogous to Theorem 1.1, the rate at which the bias goes to zero must be strictly slower than the rate at which the degree of misspecification shrinks to zero

The key condition is that $\hat{h}_n$ converges in probability to a function $\mathbb{A}^{-1}P_{Z}[g_{0}]$. $P_Z$ is a projection onto $\mathbb{A}[\mathcal{S}]$. That is, a possibly non-linear operator that maps functions into $\mathbb{A}[\mathcal{S}]$ and leaves those already in this set unchanged. This condition implies the estimator has a probability limit in $\mathcal{S}$. Thus if $h_0$ is not in $\mathcal{S}$ then the estimator is necessarily inconsistent, even under correct specification. Conversely, if Assumption 1.1 holds and $h_0\in\mathcal{S}$, then $h_0=\mathbb{A}^{-1}P_{Z}[g_{0}]$ and so $\hat{h}_n$ is consistent.

\theoremstyle{plain}
\begin{thm}
	Fix $\mu_{XZ\eta}$ so Assumptions 1.3-1.5 hold and suppose ${\mathbb{A}^{-1}[\mathcal{U}]}$ contains an open $\tilde{\mathcal{B}}_X$-ball centered at zero of radius $r_u$.
	
	Let $\mathcal{S}$ be a compact subset of $\mathcal{B}_X$ and let $P_{Z}$ be a continuous projection onto $\mathbb{A}[\mathcal{S}]$. Suppose $\hat{h}_{n}$ satisfies $
	\|\hat{h}_{n}-\mathbb{A}^{-1}P_{Z}[g_{0}]\|_{\mathcal{B}_{X}}\to^{p}0$.
	 
	If $h_{0}\in\mathcal{S}$ then  $\lim_{b\to0}bias_{\hat{h}_{n}}(b)=0$. If in addition $\mathcal{S}$ is absolutely convex and infinite-dimensional, and  $\frac{1}{\alpha}h_0\in\mathcal{S}$ for some $\alpha\in(0,1)$, then  $\lim_{b\to0}bias_{\hat{h}_{n}}(b)/b=\infty$.
\end{thm}

Let us compare again with the case of standard nonparametric regression (where $Z=X$). Suppose that under some conditions on $\mu_{XZ\eta}$ a nonparametric regression estimator converges to the condition mean $g_0(X)=E[Y|X]$. If Assumption 1.1 holds then $g_0(X)=h_0(X)$ and so the estimator is consistent under correct specification. Under misspecification, the difference between the conditional mean and the structural function is simply $u_0(X)$. If we use the same norm to measure the degree of misspecification (i.e, the magnitude of $u_0$) and the worst-case bias, then this estimator has worst-case bias $b$. Thus the estimator is robust and has worst-case bias that shrinks at the same rate as the degree of misspecification. This does not require projection onto some restricted parameter space $\mathcal{S}$.

\subsection{Continuous Linear Functionals}

The sensitivity results in Theorems 1.1-1.3 apply to the identified set for the structural function itself and estimators of the structural function. If the object of interest is instead a smooth linear
functional of the structural function then the identified set may shrink to zero with the degree of misspecification, even without any strong a priori smoothness restrictions. Moreover, it may be possible to construct
estimates that a) are consistent under correct specification of the NPIV moment condition without
any a priori restrictions on the true structural function and also b) are robust to the failure of instrumental validity.

The estimation of functionals of the structural function in NPIV models
is analyzed extensively in the literature (\citet{Ai2003}, \citet{Severini2012}, \citet{Ichimuraa}, and others). Following
\citet{Severini2012}, we let $\mathcal{B}_{X}=L_{2}(\mu_{X})$ and similarly  $\mathcal{B}_{Z}=L_{2}(\mu_{Z})$. These choices of function spaces have the advantage that any continuous linear functional $\mathbb{L}$ takes the form $\mathbb{L}[h]=E[w(X)h(X)]$ where $w\in L_{2}(\mu_{X})$. This characterization follows by the Reisz representation theorem.

The identified set for the linear functional $\mathbb{L}[h_0]$ is simply  $\mathbb{L}[\Theta_b]$.  Let $\hat{l}_{n}$ be an estimator of $\mathbb{L}[h_{0}]$ and fix $h_{0}$ and $\mu_{XZ\eta}$. The worst-case asymptotic bias of $\hat{l}_{n}$ is defined below.
\[
bias_{\hat{l}_{n}}(b)=\sup_{u_{0}\in R(\mathbb{A}):\,\|u_{0}\|_{L_{2}(\mu_{Z})}\leq b}\inf \big[\epsilon : P(|\hat{l}_{n}-\mathbb{L}[h_{0}]|)\leq \epsilon)\to 1 \big] \]

\citet{Severini2012} show that $E[w(X)h(X)]$
is estimable at rate $\sqrt{n}$ only if there exists a function $\alpha\in L_{2}(\mu_{Z})$ so that $
w(X)=E[\alpha(Z)|X]$. Theorem 1.4 shows that this condition is necessary and sufficient for robust estimation of $\mathbb{L}[h_{0}]$ to be possible without restricting the parameter space.
\begin{thm}
Fix $\mu_{XZ\eta}$ so Assumptions 1.3-1.5 hold. Take $\mathcal{H}=L_{2}(\mu_{X})$ and $\mathcal{U}=L_{2}(\mu_{Z})$. Suppose that for any $h_{0}\in L_{2}(\mu_{X})$ if Assumption 1.1 holds then $|\hat{l}_{n}-\mathbb{L}[h_{0}]|\to^{p}0$. 

\textbf{a.} If there exists $\alpha\in L_{2}(\mu_{Z})$ with $w(X)=E[\alpha(Z)|X]$, then $diam(\mathbb{L}[\Theta_b])/2$ and $bias_{\hat{l}_{n}}(b)$ are bounded above by $ b\|\alpha\|_{L_{2}(\mu_{Z})}$ for all $b>0$.

\textbf{b.} If no such $\alpha$ exists, $diam(\mathbb{L}[\Theta_b])=\infty$ and $bias_{\hat{l}_{n}}(b)=\infty$ for all $b>0$.
\end{thm}

Theorem 1.4.a suggests that robust estimation of a continuous linear functional is possible without imposing any strong conditions on the NPIV estimates. However, the condition that $w(X)=E[\alpha(Z)|X]$ for some $\alpha$ may be hard to verify. In fact, we conjecture that for a given function $w$, it is not possible to empirically verify this condition.\footnote{More precisely, we conjecture that given a function $w$ any test of the null hypothesis that there is no $\alpha$ with $w(X)= E[\alpha(Z)|X]$ has power no greater than size against any alternative.} 

\section{Empirical Sensitivity Analysis}

The results in Section 1 show that identification and estimation in NPIV models entails a trade-off. Strong a priori restrictions allow for identification and estimation that is somewhat robust to misspecification. However, if the true structural function violates these restrictions then it will not lie in the corresponding identified set and an estimator that imposes the conditions will be inconsistent, even if the NPIV moment condition holds.

To help researchers assess the robustness of their findings, and to better assess the trade-off between robustness and strong smoothness conditions, we present a method for empirical sensitivity analysis in NPIV models. Our procedure is based on estimation of the identified set under the weakened instrumental validity condition in Assumption $1.1^*$ and Assumption 1.2. We estimate the identified set for a range of bounds $b$ on the degree of misspecification and a range of smoothness restrictions.  The approach is similar in spirit to \citet{Masten2018} who analyze the sensitivity of treatment effect estimates to a small failure of unconfoundedness by weakening the assumptions used for point identification. More precisely, we estimate the identified set for linear functionals of $h_0$. An important special case is the identified set for $h_0(x)$ at some $x$ in the support of $X$.

Note that our approach assesses the robustness of empirical findings to deviations from instrumental validity and weakened smoothness conditions, it does not assess the sensitivity of any particular NPIV estimator. In Appendix A3 we also provide a method to directly assess the finite-sample sensitivity of specific NPIV point-estimates to invalid instruments.

Let $\Theta_b$ be the identified set as defined in (\ref{idset}). We only restrict the degree of misspecification and thus we take $\mathcal{U}$ to be the whole space $\mathcal{B}_Z$. Recall that the identified set for a linear functional $\mathbb{L}$ is $\mathbb{L}[\Theta_b]$. Unlike in the previous section we do not restrict ourselves to only continuous linear functionals. Of particular interest is the identified set for the value of $h_0$ at a point $x$, this is the set $[h(x): h\in\Theta_b]$ and corresponds to $\mathbb{L}:h\mapsto h(x)$. 

For the choices of $\mathcal{B}_Z$ and $\mathcal{H}$ we consider, $\mathbb{L}[\Theta_b]$ is an interval.\footnote{See Proposition
2.1 in the Appendix B for a proof.} Denote the lower and upper endpoints for the interval by $\underline{\theta}_{\mathbb{L}}$ and $\bar{\theta}_{\mathbb{L}}$, for simplicity we suppress the subscript $b$. We wish to estimate these end-points.


Assumption $1.1^*$ depends on the norm of the space $\mathcal{B}_Z$ and Assumption 1.2 on the choice of $\mathcal{H}$. We assume that these spaces are chosen so that the two assumptions can be written in terms of inequality constraints.
\newtheorem*{A2.1}{Assumption 2.1}
\begin{A2.1}
i. $\mathcal{U}=\mathcal{B}_Z$ and $\|\cdot\|_{\mathcal{B}_Z}$ is the sup norm. Thus Assumption $1.1^*$ states that if $h=h_0$, then with probability $1$:
\begin{equation}
|E[Y-h(X)|Z]|\leq b\label{eq:cond1}
\end{equation}
 ii. $\mathcal{H}$ is the set of continuous bounded functions $h$ so that:
\begin{equation}
 \mathbb{T}[h](x)\leq C(x),\forall x\in\mathcal{X}\label{eq:cond2}
\end{equation}
Where  $\mathbb{T}$ is a known linear functional from
$\mathcal{B}_{X}$ to $\mathcal{B}_{X}^{d}$, and $C\in \mathcal{B}_{X}^{d}$
is a known vector of functions.
\end{A2.1}
Under Assumption 2.1 $\underline{\theta}_\mathbb{L}$ and $\bar{\theta}_\mathbb{L}$ are defined as follows:
\begin{align*}
\underline{\theta}_\mathbb{L}= & \inf_{h\in\mathcal{B}_{X}}\mathbb{L}[h] \text{\ \ subject to (\ref{eq:cond1}) and (\ref{eq:cond2}).}\\
\bar{\theta}_\mathbb{L}= & \sup_{h\in\mathcal{B}_{X}}\mathbb{L}[h] \text{\ \ subject to (\ref{eq:cond1}) and (\ref{eq:cond2}).}
 \end{align*}
 
Our statistical analysis applies for more general constraints than (\ref{eq:cond1}) and (\ref{eq:cond2}). We detail the more general framework in Appendix A2.
 
\subsection{Calibrating $b$}
The bound $b$ in Assumption $1.1^*$ determines the size of the identified set. We suggest the researcher estimate
the identified set for a range of choices for $b$ corresponding
to mild, moderate, and severe misspecification. To determine
whether $b$ corresponds to say, `mild' misspecification, it is helpful
to compare $b$ to an estimable reference quantity. This approach is common in the empirical sensitivity literature. For
example, \citet{Masten2018} compare the degree of selection on unobservables to selection
on observables, which is estimable. See also \citet{Imbens2003}, \citet{Altonji2005a}, and \citet{Oster2019}.

We suggest calibrating of $b$ against the residual variation in the outcome. Consider the following decomposition:
\[Y=E[h_0(X)|Z]+u_0(Z)+\eta\]
In the context of Example 1 in Section 1, the first two terms on the right-hand side capture the additive indirect and direct effects of the instrument. The residual $\eta$ absorbs the effect of all other factors that influence $Y$. The variation in $\eta$ is a useful reference quantity against which to calibrate the bound $b$ on the magnitude of $u_0(Z)$. We measure the variation using a difference in quantiles.

Let $q_\eta(\alpha)$ denote the $\alpha$ quantile of $\eta$. We say $b$ is small if it is equal to $q_\eta(0.5+\tau)-q_\eta(0.5-\tau)$ for a small $\tau$. Suppose $b=q_\eta(0.6)-q_\eta(0.4)$ which corresponds to $\tau=0.1$. In Example 1, this restricts the direct causal effect of the instrument to be less than twice the effect of shifting $\eta$ from its $0.4$ to its $0.6$ quantile. In our application we take $\tau=0.01$, $0.05$, and $0.1$ to correspond to mild, moderate, and severe misspecification respectively. Note that $\eta=Y-E[Y|Z]$ is a reduced-form residual and thus we can estimate $\eta$ and its quantiles in a first stage.

\subsection{Set Estimation}
 
To estimate $\underline{\theta}_\mathbb{L}$ and $\bar{\theta}_\mathbb{L}$ we must replace the optimization problems that define these quantities with feasible ones. Instead of optimizing over $\mathcal{B}_X$, we optimize over a finite-dimensional sieve space. We replace the conditional moment in (\ref{eq:cond1}) with a non-parametric estimate. Finally, we enforce the inequalities (\ref{eq:cond1}) and (\ref{eq:cond2}) only on finite grids.
 
Let $\Phi_{n}$ be a length-$K_{n}$ column vector of basis functions
on $\mathcal{X}$. In a first stage, we nonparametrically regress $Y$ and $\Phi_n(X)$ on $Z$ to obtain estimates $\hat{g}_{n}$ and $\hat{\Pi}_{n}$ of $g_{0}(Z)=E[Y|Z]$ and $\text{\ensuremath{\Pi}}_{n}(Z)=E[\Phi_{n}(X)|Z]$. Let $\mathcal{X}_{n}$ and $\mathcal{Z}_{n}$ be finite grids in the supports of $X$ and $Z$. We replace the constraints (\ref{eq:cond1}) and (\ref{eq:cond2}) with:
\begin{align}
&|\hat{g}_{n}(z)-\hat{\Pi}_{n}'(z)\beta|\leq b,\,\forall z\in\mathcal{Z}_{n} \label{eq:cond3}\\
&\mathbb{T}[\Phi_n'](x)\beta\leq C(x),\,\forall x\in\mathcal{X}_{n} \label{eq:cond4}
\end{align}
Where  $\mathbb{T}[\Phi_n'](x)$ be the $d$-by-$K_n$ matrix whose $k^{th}$ column is $\mathbb{T}[\Phi_{n,k}](x)$ where $\Phi_{n,k}$ is the $k^{th}$ component of $\Phi_n$. Let $\mathbb{L}[\Phi_n]$ be the column-vector whose $k^{th}$ entry is $\mathbb{L}[\Phi_{n,k}]$. The estimates of $\bar{\theta}_\mathbb{L}$ and $\underline{\theta}_\mathbb{L}$ are respectively:
\begin{align*}
\hat{\bar{\theta}}_{\mathbb{L},n} & =\max_{\beta\in\mathbb{R}^{K_{n}}}\mathbb{L}[\Phi_{n}]'\beta\text{\ \ subject to (\ref{eq:cond3}) and (\ref{eq:cond4}).}\\
\hat{\underline{\theta}}_{\mathbb{L},n} & =\min_{\beta\in\mathbb{R}^{K_{n}}}\mathbb{L}[\Phi_{n}]'\beta\text{\ \ subject to (\ref{eq:cond3}) and (\ref{eq:cond4}).}
\end{align*}
These are linear programming problems with $K_{n}$ scalar parameters
and $2|\mathcal{Z}_{n}|+d|\mathcal{X}_{n}|$ constraints.

\subsection{Consistency and Convergence Rates}

A vector-valued function $f$ on a domain $\mathcal{V}$ is of H{\"o}lder smoothness class $s\in (0,1]$ with constant $\xi$ if and only if:
\[
\sup_{v_{1},v_{2}\in\mathcal{V}: v_1 \neq v_2}\frac{\|f(v_{1})-f(v_{2})\|_2}{\|v_{1}-v_{2}\|_{2}^s}=\xi
\]
Let $\lfloor s \rfloor$ denote the largest integer less than $s$. $f$ is of H{\"o}lder smoothness class $s>1$ with constant $\xi$ if and only if all the derivatives of each entry of $f$ of order weakly less than $\lfloor s \rfloor$ are uniformly bounded by $\xi$ and all the derivatives of order exactly $\lfloor s \rfloor$ are of H{\"o}lder smoothness class $ s - \lfloor s \rfloor$ with constant $\xi$. A function is Lipschitz continuous with constant $\xi$ if it is of H{\"o}lder smoothness class $1$ with constant $\xi$.

Let $D_{1,n}=\underset{{x_{1}\in\mathcal{X}}}{\sup}\underset{x_{2}\in\mathcal{X}_{n}}{\min}\|x_{1}-x_{2}\|_{2}$ and $D_{2,n}=\underset{{z_{1}\in\mathcal{Z}}}{\sup}\underset{{z_{2}\in\mathcal{Z}_{n}}}{\min}\|z_{1}-z_{2}\|_{2}$. Let $C_{n}=1+\sup_{\beta\in\mathbb{R}^{K_{n}}}\frac{\|\beta\|_{2}}{|\Phi_{n}'\beta|_{\infty}}
$.
\theoremstyle{definition}
\newtheorem*{A2.2}{Assumption 2.2}
\newtheorem*{A2.3}{Assumption 2.3}
\newtheorem*{A2.4}{Assumption 2.4}
\newtheorem*{A2.5}{Assumption 2.5}
\begin{A2.2}
$\mathbb{T}[h](x)\leq C(x)$ implies $|h(x)|\leq\bar{c}$ for some
$0<\bar{c}<\infty$, and for some $\underline{c}>0$, $C(x)\geq \underline{c}$
for all $x\in\mathcal{X}$.
\end{A2.2}
\begin{A2.3}
There is a sequence of positive scalars $a_{n}\to0$ so that:
\[
|\hat{g}_{n}-g_{0}|_{\infty}+\sup_{\beta\in\mathbb{R}^{K_{n}}:\,\Phi_{n}'\beta\in\mathcal{H}}|(\hat{\Pi}_{n}-\Pi_{n})'\beta|_{\infty}=O_{p}(a_{n})
\]
\end{A2.3}
\begin{A2.4}
There is a sequence of positive scalars $\kappa_{n}\to0$ so that for any $h\in\mathcal{H}$
there exists $\beta_{n}\in\mathbb{R}^{K_{n}}$ with $|\Phi_{n}'\beta_{n}-h|_{\infty}\leq\kappa_{n}$ and:
\begin{align*}
\mathbb{T}[\Phi_{n}'](x)\beta_{n} & \leq\mathbb{T}[h](x),\,\forall x\in\mathcal{X}
\end{align*}
\end{A2.4}
\begin{A2.5}
i. $\Phi_{n}$, $C$, and each row of $\mathbb{T}[\Phi_{n}']$ are Lipschitz
continuous with constant at most $\xi_{n}$. ii. With probability approaching
$1$ both $\hat{g}_{n}$ and $\hat{\Pi}_{n}$ are Lipschitz continuous
with constant at most $G_{n}$. iii. $D_{1,n},D_{2,n}\to0$, $C_{n}\xi_{n}D_{1,n}\to0$ and $C_{n}G_{n}D_{2,n}\to0$.
\end{A2.5}

Assumption 2.2 implies that $\mathcal{H}$ is
bounded and convex. Assumption 2.3 quantifies the estimation error in $\hat{g}_n$ and $\hat{\Pi}_n$. 2.4  quantifies the error from the replacing $\mathcal{B}_{X}$ with a sieve space. Theorems 2.2 and 2.3 establish primitive conditions for Assumptions 2.3 and 2.4. Assumption 2.5 controls the error from the use
of grids $\mathcal{X}_{n}$ and $\mathcal{Z}_{n}$.
\theoremstyle{plain}
\begin{thm}
Suppose Assumptions 2.1-2.5
hold and there exists $h\in\mathcal{B}_X$ so that (\ref{eq:cond1}) and (\ref{eq:cond2}) are slack, then uniformly over all linear functionals $\mathbb{L}$ with operator norms less than some fixed constant:
\begin{align*}
|\hat{\underline{\theta}}_\mathbb{L}-\underline{\theta}_\mathbb{L}| & =O_{p}(a_{n}+\kappa_{n}+C_{n}\xi_{n}D_{1,n}+C_{n}G_{n}D_{2,n})=o_{p}(1)\\
|\hat{\bar{\theta}}_{\mathbb{L}}-\bar{\theta}_\mathbb{L}| & =O_{p}(a_{n}+\kappa_{n}+C_{n}\xi_{n}D_{1,n}+C_{n}G_{n}D_{2,n}) =o_{p}(1)
\end{align*}
In particular, the above holds uniformly over the functionals of the form $\mathbb{L}[h]=h(x)$ for each $x$ in the support of $X$.
\end{thm}
Theorem 2.1 demonstrates the well-posedness of the set estimation
problem. The first-stage rate $a_{n}$ is not premultiplied by some
growing factor like a `sieve measure of ill-posedness' (\citet{Blundell2007}).

$D_{1,n}$ and $D_{2,n}$ are small when the grids $\mathcal{X}_{n}$
and $\mathcal{Z}_{n}$ are dense. $a_{n}$ and $\kappa_{n}$ are independent of the grids, so if the grids grow dense quickly enough, the rates in the Theorem reduce to $a_n+\kappa_n$.
This suggests the grids should be made as dense as computationally feasible.

If the dimension of the sieve space $K_{n}$ grows quickly then 
$\kappa_{n}$ converges rapidly to zero. Theorem 2.2 establishes a rate $a_{n}$ that is independent
of $K_{n}$. This suggests that under the conditions of Theorem 2.2, $K_{n}$ should be made as large as is computational constraints allow.

We now provide primitive conditions for Assumption 2.3. Define series estimators $\hat{g}_n$ and $\hat{\Pi}_n$. Let $\Psi_{n}$ be a length-$L_{n}$ column vector of basis
functions on $\mathcal{Z}$ and define $\hat{Q}_{n}=\frac{1}{n}\sum_{i=1}^{n}\Psi_{n}(Z_{i})\Psi_{n}(Z_{i})'$:
\begin{align}
\hat{g}_{n}(z)&=\Psi_{n}(z)'\hat{Q}_{n}^{-1}\frac{1}{n}\sum\Psi_{n}(Z_{i})Y_{i}\label{eq:regg}\\
\hat{\Pi}_{n}(z)&=\Psi_{n}(z)'\hat{Q}_{n}^{-1}\frac{1}{n}\sum\Psi_{n}(Z_{i})\Phi_{n}(X_{i})'\label{eq:regpi}
\end{align}
In the Assumptions below $\mathcal{N}(\mathcal{H},|\cdot|_{\infty},\delta)$ is the smallest
number of $|\cdot|_{\infty}$-balls of radius $\delta$ that can cover
$\mathcal{H}$ and $\dim (Z)$ is the dimension of $Z$.

\theoremstyle{definition}
\newtheorem*{A2.6}{Assumption 2.6}
\newtheorem*{A2.7}{Assumption 2.7}
\newtheorem*{A2.8}{Assumption 2.8}
\begin{A2.6}
i. The eigenvalues of $Q_{n}=E[\Psi_{n}(Z_{i})\Psi_{n}(Z_{i})']$
are bounded uniformly above and away from zero. ii. $\mathcal{Z}$
is bounded and the distribution of $X$ given $Z$ admits a conditional
density $f_{X|Z}$ so that $\forall x\in\mathcal{X}$,  $f_{X|Z}(x,\cdot)$ is of H{\"o}lder smoothness class $s>0$ with constant at most $\bar{\ell}$.
\end{A2.6}
\begin{A2.7}
For any $s>0$ there is a sequence $R_n(s)\to 0$ so that for any $g\in\mathcal{B}_{Z}$ of H{\"o}lder smoothness class $s$ with constant $\xi$:
\[
\sup_{z\in\mathcal{Z}}|g(z)-\Psi_{n}(z)'Q_{n}^{-1}E[\Psi_{n}(Z)g(z)]|\leq \xi R_n(s)
\]
ii. For all $z\in\mathcal{Z}$, $\|\Psi_{n}(z)\|_{2}\leq\bar{\xi}_{n}$.
$\alpha_{n}(z)=\frac{\Psi_{n}(z)}{\|\Psi_{n}(z)\|_2}$
is Lipschitz continuous with constant $\ell_{n}$. iii. $\int_{0}^{1}\sqrt{log\mathcal{N}(\mathcal{H},|\cdot|_{\infty},u)}du<\infty$.
iv. $\frac{\bar{\xi}_{n}^{2}log(L_{n})}{n}\to0$
\end{A2.7}
\begin{A2.8}
i. The function $g_{0}(Z)=E[Y|Z]$ is of H{\"o}lder smoothness class $s>0$. For $m>2$,  $E\big[|Y-E[Y]|^{m}\big|Z\big]<\infty$,
$\bar{\xi}_{n}^{2m/(m-2)}log(L_{n})/n=O(1)$, $L_{n}log(L_{n})/(n^{1-2/m})=O(1)$
and $L_{n}^{2-2/\dim(Z)}/n=O(1)$. ii. $log(\ell_{n})=O\big(log(L_{n})\big)$, $\bar{\xi}_{n}=O(\sqrt{L_{n}})$ and $R_{n}(s)=O(L_{n}^{-s_0 (s)/\dim(Z)})$ for some function $s_0:\mathbb{R}_{++}\to \mathbb{R}_{++}$ and $R_n(s)=O(L_n^{-1/2})$.
\end{A2.8}

Assumption 2.6.i is standard. Smoothness
of the conditional density in 2.6.ii ensures that for any $h\in\mathcal{B}_{X}$,
$\mathbb{A}[h]$ is smooth. Assumptions 2.7.i and 2.7.ii. can be verified for common basis functions. 2.7.iii is a condition
on the metric entropy of $\mathcal{H}$, such conditions are commonplace in the sieve estimation literature (see \citet{Chen2007a}). Spaces of smooth functions like those used in our empirical application typically obey the condition (see \citet{Wainwright2019}
Chapter 5). Assumption 2.7.iv allows us to apply Rudelson's matrix law of large numbers (\citet{Rudelson1999}).

Assumption 2.8.i allows us to apply results from \citet{Belloni2015}
to derive a convergence rate for $\hat{g}_{n}$. 2.8.ii can be
verified for a given choice of basis functions, for example if $s\geq 1/2$, then the Assumption holds
for the one-dimensional B-spline case in Section 3.

\theoremstyle{plain}
\begin{thm}
Suppose Assumptions 2.1.ii, 2.2, and 2.6-2.8 hold.
Define $\hat{g}_{n}$ and $\hat{\Pi}_{n}$ as
in (\ref{eq:regg}) and (\ref{eq:regpi}). Then uniformly over sequences
$\{K_{n}\}_{n=1}^{\infty}$:
\begin{align*}
&|\hat{g}_{n}-g_{0}|_{\infty}+\sup_{\beta\in\mathbb{R}^{K_{n}}:\,\Phi_{n}'\beta\in\mathcal{H}}|(\hat{\Pi}_{n}-\Pi_{n})'\beta|_{\infty} \\
=& O_{p}\bigg(\sqrt{\frac{L_{n}log(L_{n})}{n}}+L_{n}^{-s_0 (s)/\dim(Z)}\bigg) = o_{p}(1)
\end{align*}
\end{thm}

If the conditions of the theorem hold, then setting $L_{n}$ optimally
Assumption 2.3 holds with $a_{n}=\big( \frac{log(n)}{n} \big) ^{s_0(s)/(2 s_0(s)+\dim (Z))}$.

Finally, we show Assumption 2.4 holds with $\kappa_{n}=O(K_{n}^{-\frac{1}{2}})$ for the setting in our empirical application.

\begin{thm}
Let $\mathcal{H}$ contain functions that map from a closed interval
$[a,b]$ to $[0,1]$ so that any $h\in\mathcal{H}$ is twice-differentiable
with $|\frac{\partial^{2}}{\partial x^{2}}h|_{\infty} \leq c$.
Let $\Phi_{n}$ be a vector of $s_0$-order B-spline basis functions
with $K_{n}$ knot points evenly spaced between $[a,b]$. If $s_0 \geq3$, Assumption 2.4 holds with $\kappa_{n}=O(K_{n}^{-\frac{1}{2}})$.
\end{thm}

\section{An Empirical Application}

To demonstrate the usefulness of our methods we apply them to setting in Section 5.1 of \citet{Horowitz2011}. \citet{Horowitz2011}
estimates a shape-invariant Engel curve for food using data
from the British Family Expenditure Survey.\footnote{We made use of the data file that accompanies \citet{Horowitz2011}
and adapted the accompanying code in order to evaluate Horowitz's
estimator and B-spline bases for our own methods.} Horowitz's application is in turn based on \citet{Blundell2007}
who also carry out NPIV estimation of shape-invariant Engel curves
and use the same data.

From \citet{Horowitz2011}: ``The data are 1655 household-level observations
from the British Family Expenditure Survey. The households consist
of married couples with an employed head-of-household between the
ages of 25 and 55 years.''

\citet{Blundell2007} and  \citet{Horowitz2011} aim to estimate `structural' Engel curves. Suppose a researcher were to exogenously allocate to a household a particular budget for non-durable goods. A
structural Engel curve $h_0$ measures the share of that budget the household would choose to allocate to a class of goods as a function of the budget size.

In observational settings, the share of household wealth allocated
to nondurables is decided by the
household. Therefore, the budget for nondurables depends on household
preferences. These preferences also determine the allocation of the budget to classes of non-durable goods. Thus expenditure on non-durable goods is endogenous.

To tackle this problem \citet{Blundell2007} and \citet{Horowitz2011} use household income as an instrument
for nondurables expenditure. There are many reasons to think that household income and preferences are correlated. Both tastes and income may depend on household size and socio-economic status, and those with expensive tastes may seek high-paying jobs. However, some shocks to household income may result from exogenous external factors like shocks to production costs in an employed householder's firm. If one controls for a sufficiently rich set of household covariates this may absorb the endogenous variation in income leaving only the exogenous variation. If the data do not contain rich enough household observables, then some endogeneity likely remains and the income instrument is unlikely to be fully valid.

\citet{Blundell2007}
and \citet{Horowitz2011} can only control for some coarse demographic variables.\footnote{Both papers control for demographics by analyzing a
homogeneous sub-sample. \citet{Blundell2007}
incorporate additional dummy variable controls.}
Therefore, in this setting instrumental validity may not hold exactly and it is of interest to assess what empirical findings are robust to some failure of instrumental validity.

In this setting $Y$ is the share of total expenditure on non-durables
that a household spends on food. $X$ is the
logarithm of the household's total expenditure and $Z$ is the logarithm
of household income. We take $\mathcal{H}$ to be the set of functions that map from $\mathcal{X}$
to the unit interval and have second derivative bounded
in magnitude by a constant $c$. We present results for $c=2$ and $c=5$, for comparison, the magnitude of the second derivatives of Horowitz's estimated structural
function never exceed $0.5$.

Following Subsection 2.1 we set $b=\hat{q}_\eta(0.5+\tau)-\hat{q}_\eta(0.5-\tau)$ for different values of $\tau$. $\hat{q}_\eta$ is an estimate of $q_\eta$ and is calculated by taking empirical quantiles of $Y_i -\hat{g}_n(Z_i)$. In particular we consider $\tau=0.01,0.05,0.1$ to correspond to mild, moderate, and severe misspecification. These values of $\tau$ correspond to $b=0.0043,0.024,0.046$.
 
Following \citet{Horowitz2011} we let $\Phi_{n}$ be fourth-order (cubic)
B-spline basis functions with evenly-spaced knot points.\footnote{See \citet{Boor2002} for a practical introduction to B-splines.} We carry out our first-stage estimates using series regression onto
cubic B-splines.  Motivated by the results in Theorem
2.2 we set $K_{n}$ to be large, specifically we let $K_{n}=10$.
The grid $\mathcal{X}_{n}$ consists of 100 evenly spaced points between
the smallest and largest observed values of $X$. The grid $\mathcal{Z}_{n}$
consists of 100 evenly spaced points between the $0.005$ and $0.995$
quantiles of $Z$.

Figure 4.1 contains the results of our procedure.
The figure contains six sub-figures each corresponding to a different
set of values for $\tau$ and $c$. The
lower and upper dotted lines represent $\underline{\hat{\theta}}_{h\mapsto h(x)}$
and $\hat{\bar{\theta}}_{{h\mapsto h(x)}}$, the upper and lower end points of the identified set for $h_0(x)$. The thick black line represents the half-way point between the end points, which is a point estimator with the smallest possible worst-case asymptotic bias under Assumptions $1.1^*$ and 1.2. The thin black
line is the estimator from \citet{Horowitz2011}.

\begin{figure}[ht]
\subfloat[]{
\includegraphics[scale=0.24]{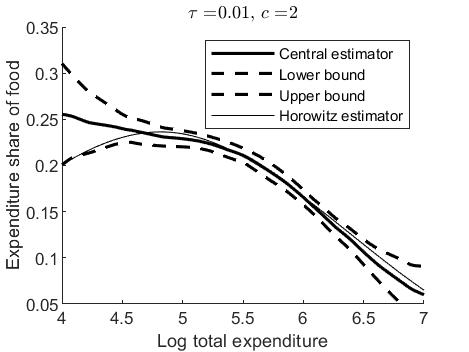}}\subfloat[]{
\includegraphics[scale=0.24]{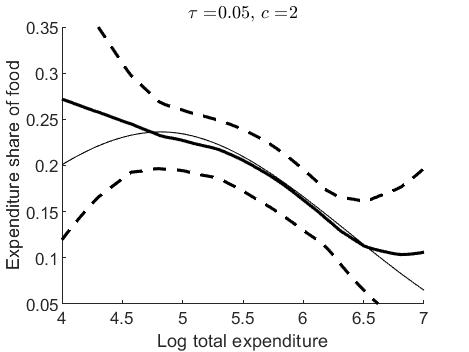}}\subfloat[]{
\includegraphics[scale=0.24]{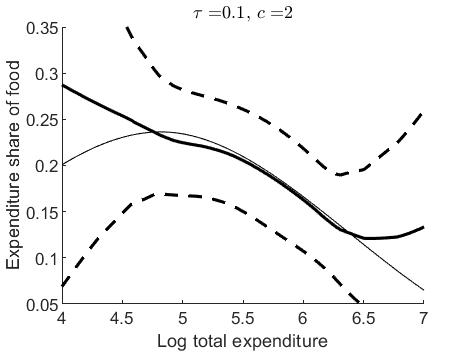}}\\
\subfloat[]{
\includegraphics[scale=0.24]{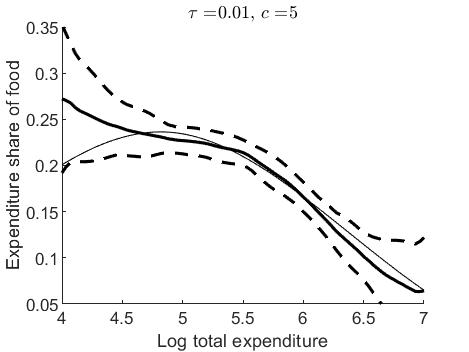}}\subfloat[]{
\includegraphics[scale=0.24]{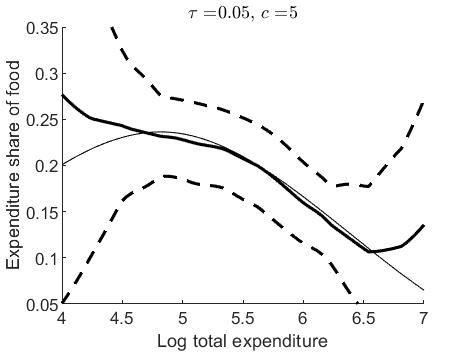}}\subfloat[]{
\includegraphics[scale=0.24]{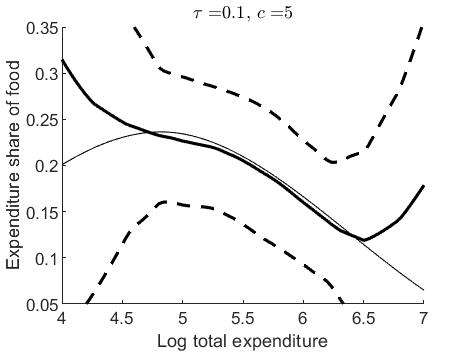}}\\
\caption{Set-Estimated Engel Curves}
{\scriptsize{}}%
\noindent\begin{minipage}[t]{1\columnwidth}%
{\scriptsize{}Results for various $\tau$ and $c$. Lower and upper dotted lines represent the end points of the identified set for $h_0(x)$ for each $x$. The thick black line is half-way between the dotted lines.
The thin black line is the estimate in \citet{Horowitz2011}.}%
\end{minipage}{\scriptsize\par}
\end{figure}

The results in Section 1 suggest that if the bound $c$ on the second
derivatives is loose then the identified set for $h_0(x)$ will be large, even
if $\tau$ (and hence $b$) is small. This is clear in Figure 4.1 which shows that for a given $\tau$, the intervals are wider when $c=5$ than when $c=2$. 

Apart from in the severely misspecified case of $\tau=0.1$ (Sub-figures (c) and (f)), the figures suggest a general downward slope in the Engel curve for medium values of total expenditure. More precisely, the lower envelope at log-total expenditure of $5$ exceeds the upper envelope at $6$, which implies the Engel curve has decreased between these two values.  None of the results in Figure 4.1 provide evidence in favor
of an upward sloping Engel curve for low values of total expenditure
as found by \citet{Horowitz2011} and so we conclude then that this finding is not robust to even a mild failure of instrumental validity.

\section*{Conclusions}

Our results show that identification and estimation in NPIV can be highly sensitive to misspecification, and that the sensitivity depends crucially upon the strength of a priori restrictions on the structural function. We develop a method for empirical sensitivity analysis in NPIV that allows researchers to better assess the relationship between the smoothness restrictions they impose and the robustness of their findings to misspecification.

We conjecture that our sensitivity results extend to a broader
class of conditional moment restriction models. The non-robustness
of NPIV estimators is tied to the ill-posedness of NPIV and a range of other nonparametric conditional moment restriction models are likewise ill-posed. It may be possible to adapt our empirical methods to these settings, although non-linearity of the conditional moment restriction would likely complicate estimation of the identified set. We leave these extensions for future work.

\bibliographystyle{authordate1}
\bibliography{npivcitesnew}

\begin{thebibliography}{}

\bibitem[\protect\citename{Ai \& Chen, }2003]{Ai2003}
Ai, Chunrong, \& Chen, Xiaohong. 2003.
\newblock Efficient estimation of models with conditional moment restrictions
  containing unknown functions.
\newblock {\em Econometrica}, {\bf 71}(6), 1795--1843.

\bibitem[\protect\citename{Ai \& Chen, }2007]{Ai2007}
Ai, Chunrong, \& Chen, Xiaohong. 2007.
\newblock Estimation of possibly misspecified semiparametric conditional moment
  restriction models with different conditioning variables.
\newblock {\em Journal of Econometrics}, {\bf 141}, 5--43.

\bibitem[\protect\citename{Altonji {\em et~al.}, }2005]{Altonji2005a}
Altonji, Joseph G., Elder, Todd E., \& Taber, Christopher R. 2005.
\newblock Selection on Observed and Unobserved Variables: Assessing the
  Effectiveness of Catholic Schools.
\newblock {\em Journal of Political Economy}, {\bf 113}, 151--184.

\bibitem[\protect\citename{Andrews, }2017]{Andrews2017b}
Andrews, Donald~WK. 2017.
\newblock Examples of L2-complete and boundedly-complete distributions.
\newblock {\em Journal of Econometrics}, {\bf 199}(2), 213--220.

\bibitem[\protect\citename{Andrews {\em et~al.}, }2017]{Andrews2017a}
Andrews, Isaiah, Gentzkow, Matthew, \& Shapiro, Jesse~M. 2017.
\newblock Measuring the Sensitivity of Parameter Estimates to Estimation
  Moments.
\newblock {\em The Quarterly Journal of Economics}, {\bf 132}, 1553--1592.

\bibitem[\protect\citename{Angrist {\em et~al.}, }1996]{Angrist}
Angrist, Joshua, Imbens, Guido, \& Rubin, D.~B. 1996.
\newblock Identification of Causal Effects Using Instrumental Variables.
\newblock {\em Journal of the American Statistical Association}.

\bibitem[\protect\citename{Armstrong \& Kolesar, }2018]{Armstrong}
Armstrong, Timothy~B., \& Kolesar, Michal. 2018.
\newblock Sensitivity Analysis using Approximate Moment Condition Models.
\newblock {\em Cowles Foundation Discussion Papers}.

\bibitem[\protect\citename{Belloni {\em et~al.}, }2015]{Belloni2015}
Belloni, Alexandre, Chernozhukov, Victor, Chetverikov, Denis, \& Kato, Kengo.
  2015.
\newblock Some new asymptotic theory for least squares series: Pointwise and
  uniform results.
\newblock {\em Journal of Econometrics}, {\bf 186}, 345--366.

\bibitem[\protect\citename{Blundell {\em et~al.}, }2007]{Blundell2007}
Blundell, Richard, Chen, Xiaohong, \& Kristensen, Dennis. 2007.
\newblock Semi-nonparametric {IV} estimation of shape-invariant {E}ngel curves.
\newblock {\em Econometrica. Journal of the Econometric Society}, {\bf 75}(6),
  1613--1669.

\bibitem[\protect\citename{Bonhomme \& Weidner, }2020]{Weidner}
Bonhomme, Stéphane, \& Weidner, Martin. 2020.
\newblock Minimizing Sensitivity to Model Misspecification.

\bibitem[\protect\citename{Canay {\em et~al.}, }2013]{Canay2013}
Canay, Ivan~A., Santos, Andres, \& Shaikh, Azeem~M. 2013.
\newblock On the testability of identification in some nonparametric models
  with endogeneity.
\newblock {\em Econometrica. Journal of the Econometric Society}, {\bf 81}(6),
  2535--2559.

\bibitem[\protect\citename{Chen, }2007]{Chen2007a}
Chen, Xiaohong. 2007.
\newblock {\em Large Sample Sieve Estimation of Semi-nonparametric Models. The
  Handbook of Econometrics, JJ Heckman and EE Leamer (eds.), 6B}.

\bibitem[\protect\citename{Chen \& Christensen, }2018]{Chen2018}
Chen, Xiaohong, \& Christensen, Timothy~M. 2018.
\newblock Optimal sup-norm rates and uniform inference on nonlinear functionals
  of nonparametric IV regression.
\newblock {\em Quantitative Economics}, {\bf 9}, 39--84.

\bibitem[\protect\citename{Chen \& Pouzo, }2012]{Cheng}
Chen, Xiaohong, \& Pouzo, Demian. 2012.
\newblock Estimation of Nonparametric Conditional Moment Models with Possibly
  Nonsmooth Generalized Residuals.
\newblock {\em Econometrica}.

\bibitem[\protect\citename{Chen {\em et~al.}, }2014]{Chena}
Chen, Xiaohong, Chernozhukov, Victor, Lee, Sokbae, \& Newey, Whitney~K. 2014.
\newblock Local Identification of Nonparametric and Semiparametric Models.
\newblock {\em Econometrica}.

\bibitem[\protect\citename{Conley {\em et~al.}, }2008]{Conley2008}
Conley, Timothy~G., Hansen, Christian~B., McCulloch, Robert~E., \& Rossi,
  Peter~E. 2008.
\newblock A semi-parametric {B}ayesian approach to the instrumental variable
  problem.
\newblock {\em Journal of Econometrics}, {\bf 144}(1), 276--305.

\bibitem[\protect\citename{Darolles {\em et~al.}, }2011]{Darolles}
Darolles, Serge, Fan, Yanqin, Florens, Jean-Pierre, \& Renault, Eric. 2011.
\newblock Nonparametric Instrumental Regression.
\newblock {\em Econometrica}.

\bibitem[\protect\citename{de~Boor, }2014]{Boor2002}
de~Boor, Carl. 2014.
\newblock {\em (B)asic-Spline Basics}.

\bibitem[\protect\citename{DeVore \& Lorentz, }1993]{DeVore1993}
DeVore, Ronald~A., \& Lorentz, George~G. 1993.
\newblock {\em Constructive Approximation}.
\newblock Springer-Verlag.

\bibitem[\protect\citename{D'Haultfoeuille, }2011]{DHaultfoeuille2011}
D'Haultfoeuille, Xavier. 2011.
\newblock On the completeness condition in nonparametric instrumental problems.
\newblock {\em Econometric Theory}, {\bf 27}(3), 460--471.

\bibitem[\protect\citename{Dudley, }1967]{Dudley1967}
Dudley, R.~M. 1967.
\newblock The sizes of compact subsets of Hilbert space and continuity of
  Gaussian processes.
\newblock {\em Journal of Functional Analysis}, {\bf 1}, 290--330.

\bibitem[\protect\citename{Florens, }2011]{Florens2011}
Florens, Jean-Pierre. 2011.
\newblock {\em Non-parametric Models with Instrumental Variables}.

\bibitem[\protect\citename{Freyberger, }2017]{Freyberger2017a}
Freyberger, Joachim. 2017.
\newblock On completeness and consistency in nonparametric instrumental
  variable models.
\newblock {\em Econometrica. Journal of the Econometric Society}, {\bf 85}(5),
  1629--1644.

\bibitem[\protect\citename{Freyberger \& Masten, }2019]{Freyberger2017}
Freyberger, Joachim, \& Masten, Matthew~A. 2019.
\newblock A Practical Guide to Compact Infinite Dimensional Parameter Spaces.
\newblock {\em Econometric Reviews}.

\bibitem[\protect\citename{Hall \& Horowitz, }2005]{Hall2005}
Hall, Peter, \& Horowitz, Joel~L. 2005.
\newblock Nonparametric methods for inference in the presence of instrumental
  variables.
\newblock {\em The Annals of Statistics}, {\bf 33}, 2904--2929.

\bibitem[\protect\citename{Horowitz, }2011]{Horowitz2011}
Horowitz, Joel~L. 2011.
\newblock Applied nonparametric instrumental variables estimation.
\newblock {\em Econometrica. Journal of the Econometric Society}, {\bf 79}(2),
  347--394.

\bibitem[\protect\citename{Horowitz, }2012]{Horowitz2012}
Horowitz, Joel~L. 2012.
\newblock Specification testing in nonparametric instrumental variable
  estimation.
\newblock {\em Journal of Econometrics}, {\bf 167}, 383--396.

\bibitem[\protect\citename{Hu \& Shiu, }2018]{Hu2018}
Hu, Yingyao, \& Shiu, Ji-Liang. 2018.
\newblock Nonparametric identification using instrumental variables: sufficient
  conditions for completeness.
\newblock {\em Econometric Theory}, {\bf 34}(3), 659--693.

\bibitem[\protect\citename{Huber, }2011]{Huber2011}
Huber, Peter~J. 2011.
\newblock {\em Robust Statistics}.

\bibitem[\protect\citename{Ichimura \& Newey, }2017]{Ichimuraa}
Ichimura, Hidehiko, \& Newey, Whitney~K. 2017.
\newblock {\em The influence function of semiparametric estimators}.

\bibitem[\protect\citename{Imbens, }2003]{Imbens2003}
Imbens, Guido~W. 2003.
\newblock Sensitivity to Exogeneity Assumptions in Program Evaluation.
\newblock {\em American Economics Review, P\&P}, {\bf 93}, 126--132.

\bibitem[\protect\citename{Kress, }2014]{Kress2014}
Kress, Rainer. 2014.
\newblock {\em Linear Integral Equations}.

\bibitem[\protect\citename{Masten \& Poirier, }2018]{Masten2018}
Masten, Matthew~A., \& Poirier, Alexandre. 2018.
\newblock Identification of Treatment Effects under Conditional Partial
  Independence.
\newblock {\em Econometrica}.

\bibitem[\protect\citename{Newey \& Powell, }2003]{Newey2003}
Newey, Whitney~K., \& Powell, James~L. 2003.
\newblock Instrumental Variable Estimation of Nonparametric Models.
\newblock {\em Econometrica}, {\bf 71}, 1565--1578.

\bibitem[\protect\citename{Oster, }2019]{Oster2019}
Oster, Emily. 2019.
\newblock Unobservable Selection and Coefficient Stability: Theory and
  Evidence.
\newblock {\em Journal of Business \& Economic Statistics}, {\bf 37}, 187--204.

\bibitem[\protect\citename{Rudelson, }1999]{Rudelson1999}
Rudelson, M. 1999.
\newblock Random Vectors in the Isotropic Position.
\newblock {\em Journal of Functional Analysis}, {\bf 164}, 60--72.

\bibitem[\protect\citename{Santos, }2012]{Santos2012}
Santos, Andres. 2012.
\newblock Inference in nonparametric instrumental variables with partial
  identification.
\newblock {\em Econometrica. Journal of the Econometric Society}, {\bf 80}(1),
  213--275.

\bibitem[\protect\citename{Severini \& Tripathi, }2012]{Severini2012}
Severini, Thomas~A., \& Tripathi, Gautam. 2012.
\newblock Efficiency bounds for estimating linear functionals of nonparametric
  regression models with endogenous regressors.
\newblock {\em Journal of Econometrics}, {\bf 170}, 491--498.

\bibitem[\protect\citename{Wainwright, }2019]{Wainwright2019}
Wainwright, Martin~J. 2019.
\newblock {\em High-Dimensional Statistics: A Non-Asymptotic Viewpoint}.
\newblock Cambridge University Press.

\end{thebibliography}

\appendix

\section{Additional Materials}

\subsection{Applying Theorem 1.2 to Specific Estimators}

We can apply Theorem 1.2 to assess the sensitivity of a particular NPIV estimator. Key is to find a set $\mathcal{S}$ that satisfies the conditions of the theorem. We can leverage existing asymptotic results for an NPIV estimator in order to find such a set. As an example, consider the estimator of \cite{Horowitz2011} (whose results we replicate in our empirical example). The estimator in \cite{Horowitz2011} is based on an estimator in \cite{Horowitz2012}.\footnote{While this appears to contradict the chronology, the 2011 paper cites an earlier pre-print of the 2012 paper.} The estimators in \cite{Horowitz2011} and \cite{Horowitz2012} differ in an important respect. While the latter constrains the estimate to be inside a Sobolev ball of fixed radius, the former does not. Theorem 1.2 applies for the unconstrained \cite{Horowitz2011} estimator but not the estimator in \cite{Horowitz2012}.

Consistency of the estimator in \cite{Horowitz2011} follows from results in \cite{Horowitz2012}. \cite{Horowitz2012} provides conditions on $\mu_{XZ\eta}$ so that the estimator is consistent whenever $h_0$ has Sobolev norm below some bound $C_0$. Consistency is defined using the $L_2[0,1]$-norm ($X$ is transformed to lie in the unit interval). The unconstrained estimator in \cite{Horowitz2011} does not depend on the specific level of $C_0$ and it appears in no other conditions, thus consistency only requires that $h_0$ has a finite Sobolev norm. Thus we can take $\mathcal{S}$ to be the space of functions with a finite Sobolev norm, and this is dense in $L_2[0,1]$. So in this case $\bar{\mathcal{S}}$ is the entire space $L_2[0,1]$. Thus if $\mathcal{U}$ contains an open ball at zero and assumptions 1.3-1.5 hold for $\mathcal{B}_X=L_2[0,1]$, the worst-case bias of Horowitz's estimator is infinite for all $b>0$.

As a second example, \cite{Darolles} prove consistency of their estimator when $h_{0}$ is in a space of functions $\Phi_{F}^{\beta}$ that obey a `source condition' (their Assumption A.2). Under statistical completeness, $\Phi_{F}^{\beta}$ is dense in $L_{2}(\mu_X)$ and so $\bar{\mathcal{S}}$ is the whole space $L_{2}(\mu_X)$.

\citet{Chen2018} study the supremum-norm consistency of a more general series two-stage least-squares estimator than \citet{Horowitz2011}. They suggest a Sobolev ball for the parameter space $\mathcal{H}$. The radius of the ball is left unspecified and the formula for the series two-stage least-squares estimator does not depend on this radius. Thus if the estimator is supremum-norm consistent when $h_0$ lies in one particular ball, it is consistent when $h_0$ lies in any ball. The space of all functions with finite Sobolev norm is dense in the set of continuous bounded functions on a compact support. So if $\mathcal{X}$ is compact we have $\bar{\mathcal{S}}$ equal to the whole space of continuous bounded functions.

\subsection{Extensions to the Set Estimator}
The results in Section 2 extend to a more general set estimation problem. For $k=1,..,K$ let $Z_k$ be a vector of instruments and let $b_{1,k}$ and $b_{2,k}$ be known functions.  We can replace the supremum norm constraint with a set of constraints:
\begin{equation}
	b_{1,k}(Z)\leq E[Y-h(X)|Z_k]\leq b_{2,k}(Z),k=1,...,K
\end{equation}
The supremum norm constraint in Section 2 is a special case of the above in which $K=1$ and $b_{1,1}(Z)=-b_{2,1}(Z)=b$. To estimate of $\bar{\theta}_\mathbb{L}$ and $\underline{\theta}_\mathbb{L}$ with the constraints above, we replace the corresponding constraint in the feasible problem with:
\[
b_{1,k}(z)\leq\hat{g}_{n,k}(z)-\hat{\Pi}_{n,k}'(z)\beta\leq b_{2,k}(z),\forall z\in\mathcal{Z}_{n,k},k=1,...,K
\]
Where $\hat{g}_{n,k}$ and $\hat{\Pi}_{n,k}$ are first-stage nonparametric estimates of $g_{0,k}(Z_k)=E[Y|Z_k]$ and $\Pi_{n,k}(Z_k)=E[\Phi_n|Z_k]$, and $\mathcal{Z}_{n,k}$ is a finite subset of the support of $Z_k$.
Theorem 2.0 in the appendix provides a rate of convergence for $\hat{\bar{\theta}}_\mathbb{L}$ and $\hat{\underline{\theta}}_\mathbb{L}$ in this more general set estimation problem. Theorem 2.1 is a special case of the more general theorem.

\subsection{Direct Finite-Sample Sensitivity}

Our set estimation method is designed to analyze the sensitivity of one's empirical findings to misspecification. It is not based on any particular estimator. In this subsection we suggest a method to assess the finite-sample sensitivity of a particular NPIV estimator and apply the method to the empirical setting in Section 3.

Consider a linear NPIV
estimator of the form $\hat{h}_{n}(x)=\sum_{i=1}^{n}\hat{q}_{i}(x)Y_{i}$
where $\hat{q}_{i}(x)$ may depend on the instruments $Z_{1},..,Z_{n}$
and endogenous regressors $X_{1},...,X_{n}$ but not the outcomes
$Y_{1},...,Y_{n}$. Many NPIV estimators take this form, for example those of \cite{Horowitz2011} and \cite{Darolles}. We can decompose the estimator as follows:

\[
\hat{h}_{n}(x)=\sum_{i=1}^{n}\hat{q}_{i}(x)E[h_{0}(X_{i})|Z_{i}]+\sum_{i=1}^{n}\hat{q}_{i}(x)u_{0}(Z_{i})+\sum_{i=1}^{n}\hat{q}_{i}(x)\eta_{i}
\]

Under correct specification $u_{0}(Z_{i})=0$ for all $i$ and keeping
all else fixed we would instead attain the estimate $\tilde{h}_{n}(x)$
given by: 
\[
\tilde{h}_{n}(x)=\sum_{i=1}^{n}\hat{q}_{i}(x)E[h_{0}(X_{i})|Z_{i}]+\sum_{i=1}^{n}\hat{q}_{i}(x)\eta_{i}
\]

Suppose we assume
that $|u_{0}(Z)|\leq b$. Then with probability $1$ the interval below contains the value the estimator would take under correct specification. This interval is not conservative, for
each point in the interval there is a function $u_{0}$ that satisfies
the constraint so that $\tilde{h}(x)$ is equal to that point.
\[
\tilde{h}(x)\in[\hat{h}_{n}(x)-b\sum_{i=1}^{n}|\hat{q}_{i}(x)|,\hat{h}_{n}(x)+b\sum_{i=1}^{n}|\hat{q}_{i}(x)|]
\]

A wide interval suggests that the presence of misspecification (of norm below $b$) can have a large effect on the the estimate. Note the width scales linearly with $b$ and at each point $x$ the width is determined by the factor $\sum_{i=1}^{n}|\hat{q}_{i}(x)|$ which we can take as a measure of the sensitivity to misspecification.

Below we apply the method above to analyze the sensitivity of the specific estimator in \cite{Horowitz2011}. The results are given in Figure A1 below. For  $\tau\geq 0.05$ these bounds are wider than those based on the identified set, even though the bounds do not account for any bias due to excessive Tikhonov regularization or the finite number of series terms used.\footnote{The estimates in \cite{Horowitz2011} employ a small degree of Tikhonov regularization although the description of the estimator in the paper omits this.} The maximum width of each interval is approximately $26$ times the corresponding value of $b$ (in each case $\tau$ is around twice the corresponding value of $b$). This demonstrates the substantial sensitivity of this estimator to misspecification.

\begin{figure}[ht]
	\subfloat[]{
		\includegraphics[scale=0.24]{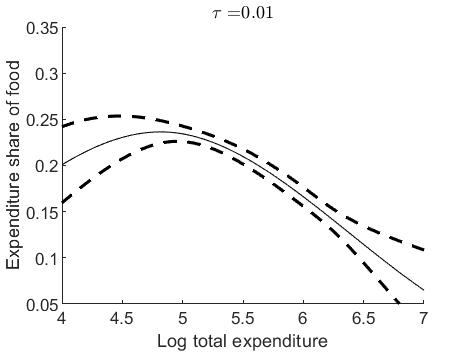}}\subfloat[]{
		\includegraphics[scale=0.24]{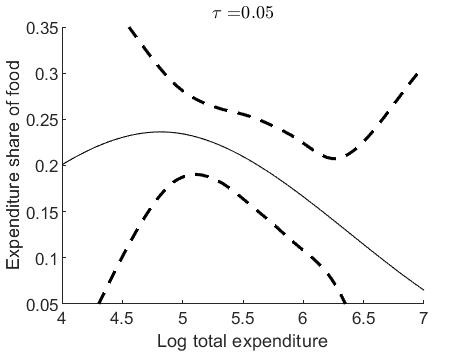}}\subfloat[]{
		\includegraphics[scale=0.24]{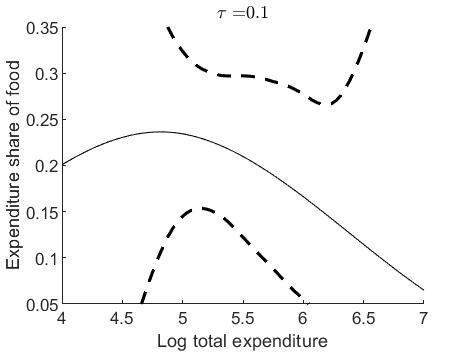}}\\
	
	\caption{Direct Finite-Sample Sensitivity}
	{\scriptsize{}}%
	\noindent\begin{minipage}[t]{1\columnwidth}%
		{\scriptsize{}Results for various $\tau$. Lower and upper dotted lines represent the end points of the interval for $\tilde{h}_n(x)$ for each $x$ as detailed in Subsection 2.3.	The thin black line is the estimate $\hat{h}_n$ in \citet{Horowitz2011}.}%
	\end{minipage}{\scriptsize\par}
\end{figure}

\section{Proofs}
\subsection{Proofs of Results in Section 1}

\newtheorem*{1.1}{Lemma 1.1}

\begin{1.1} 
	
	Suppose Assumption 1.4 holds. Then $\tilde{\mathcal{B}}_{X}$ is a
	Banach space with the norm $\|\cdot\|_{\mathcal{B}_{X}}$and $\mathbb{A}$
	is a compact infinite-dimensional linear operator from $\tilde{\mathcal{B}}_{X}$
	to $\mathcal{B}_{Z}$.
	
\end{1.1} 

\begin{proof} 
	
	First we show that $\tilde{\mathcal{B}}_{X}$ is a Banach space with
	the norm $\|\cdot\|_{\mathcal{B}_{X}}$, that is, a linear space that
	is complete with respect to $\|\cdot\|_{\mathcal{B}_{X}}$. Recall
	that $\tilde{\mathcal{B}}_{X}$ contains all elements of $\mathcal{B}_{X}$
	that have mean zero. Let $h_{1},h_{2}\in\tilde{\mathcal{B}}_{X}$
	and $\alpha,\beta\in\mathbb{R}$. Note that $E[h_{1}(X)]=E[h_{2}(X)]=0$
	and $h_{1},h_{2}\in\mathcal{B}_{X}$. Since $\mathcal{B}_{X}$ is
	a Banach space and thus linear, $\alpha h_{1}+\beta h_{2}\in\mathcal{B}_{X}$,
	and by linearity of the mean $E[\alpha h_{1}(X)+\beta h_{2}(X)]=0$.
	Thus $\alpha h_{1}+\beta h_{2}\in\tilde{\mathcal{B}}_{X}$ and so
	$\tilde{\mathcal{B}}_{X}$ is a linear space. To show $\mathcal{\tilde{B}}_{X}$
	is complete, let $\{h_{n}\}_{n=1}^{\infty}$ be a Cauchy sequence
	in $\mathcal{\tilde{B}}_{X}$. Since $\mathcal{B}_{X}$ is Banach
	and thus complete, this sequence converges (in the norm $\|\cdot\|_{\mathcal{B}_{X}}$)
	to an element $h_{\infty}\in\mathcal{B}_{X}$. By Assumption 1.5 the
	mapping of a function in $\mathcal{B}_{X}$ to its mean is continuous,
	and so $h_{n}\to h_{\infty}$ implies $E[h_{n}(X)]\to E[h_{\infty}(X)]$
	and since $E[h_{n}(X)]=0$ for all $n$, we must have $E[h_{\infty}(X)]=0$.
	Thus $h_{\infty}\in\tilde{\mathcal{B}}_{X}$ and so $\tilde{\mathcal{B}}_{X}$
	is complete.
	
	Finally we show that $\mathbb{A}$ is a compact operator from $\tilde{\mathcal{B}}_{X}$
	to $\mathcal{B}_{Z}$. By definition a compact operator maps bounded
	sets into relatively compact sets. Assumption 1.4 states that $\mathbb{A}$
	is a compact operator between $\mathcal{B}_{X}$ and $\mathcal{B}_{Z}$.
	Since any bounded set in $\tilde{\mathcal{B}}_{X}$ is also a bounded
	set in $\mathcal{B}_{X}$ we get that $\mathbb{A}$ is compact between
	$\tilde{\mathcal{B}}_{X}$ to $\mathcal{B}_{Z}$.
\end{proof} 

\newtheorem*{1.2}{Lemma 1.2}

\begin{1.2} 
	
	Suppose Assumption 1.4 holds, there is an open $\tilde{\mathcal{B}}_{X}$-ball
	centered at zero in $\mathbb{A}^{-1}[\mathcal{U}]$ of radius $r_{u}$,
	and an open $\tilde{\mathcal{B}}_{X}$-ball centered at an element
	$h^{*}$ in a set $\bar{\mathcal{S}}$ of radius $r_{s}$. Then for
	any $\delta>0$ there exist functions $h_{1},h_{2}\in\mathcal{S}$,
	so that for $i=1,2$, $\mathbb{A}[h_{i}-h^{*}]\in\mathcal{U}$, $\|\mathbb{A}[h_{i}-h^{*}]\|_{\mathcal{B}_{Z}}\leq b$,
	$\|h_{i}-h^{*}\|\geq\min\{r_{h},r_{u}\}-\delta$, and $\|h_{1}-h_{2}\|\geq2\min\{r_{s},r_{u}\}-\delta$.
	
\end{1.2} 

\begin{proof}
	
	From Lemma 1.1 $\tilde{\mathcal{B}}_{X}$ is a Banach space and $\mathbb{A}$
	is a compact infinite-dimensional operator from $\tilde{\mathcal{B}}_{X}$
	to $\mathcal{B}_{Z}$. Thus we can apply Theorem 15.4 (or 2.20) in \cite{Kress2014} to get that $\mathbb{A}^{-1}$ is unbounded and so:
	\begin{equation}
		\sup_{h\in\tilde{\mathcal{B}}_{X},\|h\|_{\mathcal{B}_{X}}=1}\|\mathbb{A}[h]\|_{\mathcal{B}_{Z}}=\infty\label{eq:illp-1}
	\end{equation}
	
	Let $\|\mathbb{A}\|$ be the operator norm of $\mathbb{A}$, this
	must be finite because $\mathbb{A}$ is compact and therefore bounded
	(see Theorem 2.14 in \cite{Kress2014}). By (\ref{eq:illp-1}) for any $0<\delta<2\min\{r_{s},r_{u},\frac{b}{\|\mathbb{A}\|}\}$
	there exists an element $\tilde{h}$ of $\tilde{\mathcal{B}}_{X}$
	so that $\|\tilde{h}\|_{\mathcal{B}_{X}}=\min\{r_{s},r_{u}\}-\frac{\delta}{2}$
	and $\|\mathbb{A}[\tilde{h}]\|_{\mathcal{B}_{Z}}\leq b-\frac{\delta}{2}\|\mathbb{A}\|$.
	By linearity of $\mathbb{A}$ and the elementary properties of norms
	we also have $\|-\tilde{h}\|_{\mathcal{B}_{X}}=\min\{r_{s},r_{u}\}-\frac{\delta}{2}$
	and $\|\mathbb{A}[-\tilde{h}]\|_{\mathcal{B}_{Z}}\leq b-\frac{\delta}{2}\|\mathbb{A}\|$.
	
	Because $\|\tilde{h}\|_{\mathcal{B}_{X}}\leq r_{h}$ we have $h^{*}+\tilde{h}\in\bar{\mathcal{S}}$
	and similarly $h^{*}-\tilde{h}\in\bar{\mathcal{S}}$. Since $\mathcal{S}$
	is dense in $\bar{\mathcal{S}}$ there exists an $h_{1}\in\mathcal{H}$
	so that $\|h^{*}+\tilde{h}-h_{1}\|_{\mathcal{B}_{X}}\leq\frac{\delta}{2}$
	in which case, by the triangle inequality:
	\[
	\|h_{1}-h^{*}\|_{\mathcal{B}_{X}}\leq\min\{r_{s},r_{u}\}
	\]
	
	The inequality above implies that $h_{1}-h^{*}\in\mathbb{A}^{-1}[\mathcal{U}]$
	and therefore $\mathbb{A}[h_{1}-h^{*}]\in\mathcal{U}$. Also by the
	definition of the operator norm and the triangle equality: 
	\begin{align*}
		\|\mathbb{A}[h_{1}-h^{*}]\|_{\mathcal{B}_{Z}} & \leq\|\mathbb{A}[\tilde{h}]\|_{\mathcal{B}_{Z}}+\|\mathbb{A}[h^{*}+\tilde{h}-h_{1}]\|_{\mathcal{B}_{Z}}\\
		& \leq\|\mathbb{A}[\tilde{h}]\|_{\mathcal{B}_{Z}}+\|\mathbb{A}\|\|h^{*}+\tilde{h}-h_{1}\|_{\mathcal{B}_{X}}\\
		& \leq b
	\end{align*}
	
	So in all $h_{1}\in\mathcal{S}$, $\mathbb{A}[h_{1}-h^{*}]\in\mathcal{U}$,
	$\|\mathbb{A}[h^{*}-h_{1}]\|_{\mathcal{B}_{Z}}\leq b$, and $\|h^{*}+\tilde{h}-h_{1}\|_{\mathcal{B}_{X}}\leq\frac{\delta}{2}$.
	Applying the same reasoning with $\tilde{h}$ replaced by $-\tilde{h}$
	we see that there exists an $h_{2}\in\mathcal{S}$ with $\mathbb{A}[h_{2}-h^{*}]\in\mathcal{U}$,
	and $\|\mathbb{A}[h^{*}-h_{2}]\|_{\mathcal{B}_{Z}}\leq b$, and $\|h^{*}-\tilde{h}-h_{2}\|_{\mathcal{B}_{X}}\leq\frac{\delta}{2}$.
	Now that note that by the triangle inequality:
	
	\begin{align*}
		\|h_{1}-h^{*}\|_{\mathcal{B}_{X}} & \geq\|\tilde{h}\|_{\mathcal{B}_{X}}-\|h^{*}+\tilde{h}-h_{1}\|_{\mathcal{B}_{X}}\\
		& \geq\min\{r_{s},r_{u}\}-\delta
	\end{align*}
	
	And similarly $\|h_{2}-h^{*}\|_{\mathcal{B}_{X}}\geq\min\{r_{s},r_{u}\}-\delta$.
	Moreover, by the triangle inequality:
	\begin{align*}
		\|h_{1}-h_{2}\|_{\mathcal{B}_{X}} & \geq\|h^{*}+\tilde{h}-(h^{*}-\tilde{h})\|_{\mathcal{B}_{X}}-\delta\\
		& =2\|\tilde{h}\|_{\mathcal{B}_{X}}-2\delta\\
		& =2\min\{r_{s},r_{u}\}-2\delta
	\end{align*}
\end{proof} 

\newtheorem*{1.3}{Lemma 1.3}

\begin{1.3} 
	
	Suppose Assumptions 1.3 and 1.4 hold. If $\mathcal{S}$ is a compact
	subset of $\mathcal{B}_{X}$ then:
	
	\[
	\lim_{b\to0}\sup_{u\in\mathbb{A}[\mathcal{S}],\|u\|_{\mathcal{B}_{Z}}\leq b}\|\mathbb{A}^{-1}[u]\|_{\mathcal{B}_{X}}=0
	\]
	
\end{1.3} 

\begin{proof}
	
	Since Assumption 1.3 holds, $\mathbb{A}$ is injective. Denote the
	restriction of $\mathbb{A}$ to $\mathcal{S}$ by $\mathbb{A}_{\mathcal{S}}$
	and its inverse by $\mathbb{A}_{\mathcal{S}}^{-1}$. It is well-known
	that a continuous and injective function defined on a compact set
	has a continuous inverse. So by compactness of $\mathcal{S}$, $\mathbb{A}_{\mathcal{S}}^{-1}$
	is continuous. Continuity of $\mathbb{A}_{\mathcal{S}}^{-1}$ implies:
	\[
	\lim_{b\to0}\sup_{u\in\mathbb{A}[\mathcal{S}],\|u\|_{\mathcal{B}_{Z}}\leq b}\|\mathbb{A}_{\mathcal{S}}^{-1}[u]\|_{\mathcal{B}_{X}}=0
	\]
	
	The final result follows because $\mathbb{A}_{\mathcal{S}}^{-1}[u]$
	and $\mathbb{A}^{-1}[u]$ coincide for $u\in\mathbb{A}[\mathcal{S}]$.	
\end{proof}

\newtheorem*{1.4}{Lemma 1.4} \begin{1.4} Suppose Assumptions
	1.3-1.5 hold, there is an open $\tilde{\mathcal{B}}_{X}$-ball centered
	at zero in $\mathbb{A}^{-1}[\mathcal{U}]$, and the set $\text{\ensuremath{\mathcal{S}}\ensuremath{\subset\mathcal{B}_{X}}}$
	is such that $\mathcal{S}\cap\tilde{\mathcal{B}}_{X}$ is absolutely
	convex and infinite dimensional. Let $h^{*}\in\mathcal{S}$ and suppose
	there exists $\alpha\in(0,1)$ so that $\frac{1}{\alpha}h^{*}\in\mathcal{S}$.
	Then:
	\[
	\lim_{b\to0}\frac{\sup_{h\in\mathcal{S},\mathbb{A}[h-h^{*}]\in\mathcal{U}:\,\|\mathbb{A}[h-h^{*}]\|_{\mathcal{B}_{Z}}\leq b}\|h-h^{*}\|_{\mathcal{B}_{X}}}{b}=\infty
	\]
\end{1.4} \begin{proof} Assume the contrary, then for some $b^{*}>0$
	there exists a finite scalar $C$ so that for any $h\in\mathcal{S}$
	with $\mathbb{A}[h-h^{*}]\in\mathcal{U}$:
	\[
	\|\mathbb{A}[h-h^{*}]\|_{\mathcal{B}_{Z}}\leq b^{*}\implies\|h-h^{*}\|_{\mathcal{B}_{X}}\leq C\|\mathbb{A}[h-h^{*}]\|_{\mathcal{B}_{Z}}
	\]
	
	By definition of the operator norm of $\mathbb{A}$, $\|\mathbb{A}\|_{op}$,
	$\|\mathbb{A}[h-h^{*}]\|_{\mathcal{B}_{Z}}\leq\|\mathbb{A}\|_{op}\|h-h^{*}\|_{\mathcal{B}_{X}}$.
	Assumption 1.3 implies that $\|\mathbb{A}\|_{op}>0$. And so for any
	$h\in\mathcal{S}$ with $\mathbb{A}[h-h^{*}]\in\mathcal{U}$ we must
	have: 
	\[
	\|h-h^{*}\|_{\mathcal{B}_{X}}\leq\frac{1}{\|\mathbb{A}\|_{op}}b^{*}\implies\|h-h^{*}\|_{\mathcal{B}_{X}}\leq C\|\mathbb{A}[h-h^{*}]\|_{\mathcal{B}_{Z}}
	\]
	
	Recall that there is an open $\tilde{\mathcal{B}}_{X}$-ball centered
	at zero in $\mathbb{A}^{-1}[\mathcal{U}]$, call the radius of this
	ball $r_{u}$. If $\|h-h^{*}\|_{\mathcal{B}_{X}}\leq r_{u}$ and $h-h^{*}\in\tilde{\mathcal{B}}_{X}$
	it follows that $\mathbb{A}[h-h^{*}]\in\mathcal{U}$. Let $c^{*}=\min\{\frac{1}{\|\mathbb{A}\|_{op}}b^{*},r_{u}\}$,
	we get that for any $h\in\mathcal{S}$ with $h-h^{*}\in\tilde{\mathcal{B}}_{X}$:
	\[
	\|h-h^{*}\|_{\mathcal{B}_{X}}\leq c^{*}\implies\|h-h^{*}\|_{\mathcal{B}_{X}}\leq C\|\mathbb{A}[h-h^{*}]\|_{\mathcal{B}_{Z}}
	\]
	
	We have $\frac{1}{\alpha}h^{*}\in\mathcal{S}$ for some $\alpha\in(0,1)$
	and $\mathcal{S}$ is convex. So for any $h\in\mathcal{S}\cap\tilde{\mathcal{B}}_{X}$,
	$(1-\alpha)h+\alpha(\frac{1}{\alpha}h^{*})=h'\in\mathcal{S}$ in which
	case $(1-\alpha)h=h'-h^{*}\in\tilde{\mathcal{B}}_{X}$. Thus any $h\in\mathcal{S}\cap\tilde{\mathcal{B}}_{X}$
	can be written as $h'-h^{*}$ for some $h'\in\mathcal{S}$ and $h-h^{*}\in\tilde{\mathcal{B}}_{X}$
	and $\|h'-h^{*}\|_{\mathcal{B}_{X}}=(1-\alpha)\|h\|_{\mathcal{B}_{X}}$
	. Therefore, for any $h\in\mathcal{S}\cap\tilde{\mathcal{B}}_{X}$:
	\[
	\|h\|_{\mathcal{B}_{X}}\leq\frac{1}{1-\alpha}c^{*}\implies\|h\|_{\mathcal{B}_{X}}\leq C\|\mathbb{A}[h]\|_{\mathcal{B}_{Z}}
	\]
	Let $R$ be the closed ball in $\mathcal{\tilde{B}}_{X}$ of radius
	$\frac{1}{1-\alpha}\min\{\frac{1}{\|\mathbb{A}\|_{op}}b^{*},r_{u}\}$.
	Let $\mathcal{C}=[\gamma h:\,h\in R\cap\mathcal{S},\gamma\in\mathbb{R}]$.
	We have already shown that for any $h\in R\cap\mathcal{S}$, $\|h\|_{\mathcal{B}_{X}}\leq C\|\mathbb{A}[h]\|_{\mathcal{B}_{Z}}$.
	By linearity of $\mathbb{A}$ and properties of norms, for any $h\in\mathcal{C}$
	we have that $\|h\|_{\mathcal{B}_{X}}\leq C\|\mathbb{A}[h]\|_{\mathcal{B}_{Z}}$.
	
	Let $\bar{\mathcal{C}}$ be the closure of $\mathcal{C}$. For any
	$h\in\bar{\mathcal{C}}$ there is a sequence $h_{k}$ in $\mathcal{C}$
	so that $\|h-h_{k}\|_{\mathcal{B}_{X}}\to0$. For all $k$, $\|h_{k}\|_{\mathcal{B}_{X}}\leq C\|\mathbb{A}[h_{k}]\|_{\mathcal{B}_{Z}}$,
	so by the triangle inequality and the definition of the operator norm:
	\[
	\|h\|_{\mathcal{B}_{X}}\leq C\|\mathbb{A}[h]\|_{\mathcal{B}_{Z}}+(1+C\|\mathbb{A}\|_{op})\|h-h_{k}\|_{\mathcal{B}_{X}}
	\]
	$\mathbb{A}$ is continuous so $\|\mathbb{A}\|_{op}<\infty$, and
	so since $\|h-h_{k}\|_{\mathcal{B}_{X}}\to0$ we get $\|h\|_{\mathcal{B}_{X}}\leq C\|\mathbb{A}[h]\|_{\mathcal{B}_{Z}}$.
	Thus the inverse $\mathbb{A}^{-1}$ (which exists by Assumption 1.3)
	is bounded on $\bar{\mathcal{C}}$.
	
	Now, $R\cap\mathcal{S}$ is absolutely convex which implies $\mathcal{C}$
	is a linear space and therefore so is $\bar{\mathcal{C}}$. Because
	$\mathcal{S}$ is infinite-dimensional and absolutely convex, $\mathcal{C}$
	is infinite-dimensional and likewise $\bar{\mathcal{C}}$. It is well-known
	that a closed subset of a complete space is complete, $\bar{\mathcal{C}}$
	is a closed subset of $\tilde{\mathcal{B}}_{X}$ by construction and
	$\tilde{\mathcal{B}}_{X}$ is a Banach space (see Lemma 1.1) and thus
	complete in the norm $\|\cdot\|_{\mathcal{B}_{X}}$. Thus $\bar{\mathcal{C}}$
	is an infinite-dimensional, complete linear space, i.e., an infinite-dimensional
	Banach space. But the inverse of a compact injective operator on an
	infinite-dimensional Banach space cannot be bounded (see Theorem 15.4 in \cite{Kress2014}), and
	so we have a contradiction. \end{proof}

\begin{proof}[Proof of Theorem 1.1]
	We begin with part a. Apply Lemma 1.2 with $\mathcal{S}=\mathcal{H}$ and $h^{*}$ the element
	of $\mathbb{A}^{-1}[g_{0}]$ so that $\bar{\mathcal{H}}$ contains
	an open ball of radius $r_{h}$ centered at this element. We see that
	for any $\delta>0$ there exists $h_{1},h_{2}\in\mathcal{H}$, with
	$\|h_{1}-h_{2}\|_{\mathcal{B}_{X}}\geq2\min\{r_{h},r_{u}\}-2\delta$
	and for $i=1,2$, $\mathbb{A}[h_{i}]-g_{0}\in\mathcal{U}$, and $\|g_{0}-\mathbb{A}[h_{i}]\|_{\mathcal{B}_{Z}}\leq b$.
	Thus for $i=1,2$, $h_{i}\in\Theta_{b}$. Since we make $\delta$
	arbitrarily small, it follows that the diameter of $\Theta_{b}$ satisfies:
	\[
	\sup_{h_{1},h_{2}\in\Theta_{b}}\|h_{1}-h_{2}\|_{\mathcal{B}_{X}}\geq2\min\{r_{h},r_{u}\}
	\]
	
	Now consider part b. First we show that $\lim_{b\to0}diam(\Theta_{b})=0$. It is well known that a continuous function maps compact sets into
	compact sets. By Assumption 1.4 $\mathbb{A}$ is compact and therefore
	continuous. Thus $\mathbb{A}[\mathcal{H}]$ is compact. It is also
	well-known that a compact set in a Hausdorff space is closed, and
	so $\mathbb{A}[\mathcal{H}]$ is closed. Suppose $g_{0}\notin\mathbb{A}[\mathcal{H}]$,
	then by closedness of $\mathbb{A}[\mathcal{H}]$ there exists an open
	ball centered at $g_{0}$ that does not intersect $\mathbb{A}[\mathcal{H}]$,
	in which case for $b$ sufficiently small, there is no $h\in\mathcal{H}$
	so that $\|g_{0}-\mathbb{A}[h]\|_{\mathcal{B}_{Z}}<b$ and so the
	identified set is empty for $b$ sufficiently small and $\lim_{b\to\infty}diam(\Theta_{b})=0$
	holds trivially. So it remains to consider the case of $g_{0}\in\mathbb{A}[\mathcal{H}]$
	and thus we assume this for the remainder of the theorem.
	
	Now, given $g_{0}\in\mathbb{A}[\mathcal{H}]$ and $\mathbb{A}$ is
	injective by Assumption 1.3, we have:
	\begin{align*}
		\sup_{h_{1},h_{2}\in\Theta_{b}}\|h_{1}-h_{2}\|_{\mathcal{B}_{X}}		& \leq2\sup_{h\in\Theta_{b}}\|h-\mathbb{A}^{-1}[g_{0}]\|_{\mathcal{B}_{X}}\\
		& \leq2\sup_{h\in\mathcal{H}:\,\|\mathbb{A}[h]-g_{0}\|_{\mathcal{B}_{Z}}\leq b}\|h-\mathbb{A}^{-1}[g_{0}]\|_{\mathcal{B}_{X}}\\
		& \leq\sup_{u\in\mathbb{A}[\mathcal{H}]-g_{0}:\,\|u\|_{\mathcal{B}_{Z}}\leq2b}\|\mathbb{A}^{-1}[u]\|_{\mathcal{B}_{X}}
	\end{align*}
	
	Where the first inequality follows by the triangle inequality, the
	second because the set $[h\in\mathcal{H}:\,\|\mathbb{A}[h]-g_{0}\|\leq b]$
	is a subset of $\Theta_{b}$, and the final inequality by a reparameterization.
	Note that $\mathbb{A}[\mathcal{H}]-g_{0}$ is defined so that $u\in\mathbb{A}[\mathcal{H}]-g_{0}$
	if and only if $u-g_{0}\in\mathbb{A}[\mathcal{H}]$. 
	
	Finally, since $\mathcal{H}$ is compact, it follows that $\mathcal{H}-\mathbb{A}^{-1}[g_{0}]$
	is compact. Thus we can apply Lemma 1.4 with $\mathcal{S}=\mathcal{H}-\mathbb{A}^{-1}[g_{0}]$
	and we get the result.
	
	Now we show that $\lim_{b\to0}diam(\Theta_{b})/b=\infty$. Since $\mathbb{A}^{-1}[g_{0}]\in\Theta_b$:
	\begin{align*}
		\sup_{h_{1},h_{2}\in\Theta_{b}}\|h_{1}-h_{2}\|_{\mathcal{B}_{X}}
		\geq & \sup_{h\in\Theta_{b}}\|h-\mathbb{A}^{-1}[g_{0}]\|_{\mathcal{B}_{X}}\\
		= & \sup_{h\in\mathcal{S},\mathbb{A}[h]-g_{0}\in\mathcal{U}:\,\|\mathbb{A}[h]-g_{0}\|_{\mathcal{B}_{Z}}\leq b}\|h-\mathbb{A}^{-1}[g_{0}]\|_{\mathcal{B}_{X}}
	\end{align*}
	
	Applying Lemma 1.4 with $\mathcal{S}=\mathcal{H}$ and $h^{*}=\mathbb{A}^{-1}[g_{0}]$
	then gives the result. \end{proof} 

\begin{proof}[Proof of Theorem 1.2]
	
	Applying Lemma 1.2 with $h^{*}=h_{0}$, there exists an $h\in\mathcal{S}$
	so that $\mathbb{A}[h-h_{0}]\in\mathcal{U}$, $\|\mathbb{A}[h-h_{0}]\|_{\mathcal{B}_{Z}}\leq b$,
	and $\|h-h_{0}\|\geq\min\{r_{s},r_{u}\}-\delta$. Fix $u_{0}=\mathbb{A}[h-h_{0}]$,
	then $u_{0}\in\mathcal{U}$ and $\|u_{0}\|_{\mathcal{B}_{Z}}\leq b$.
	By Assumption 1.3 $\mathbb{A}^{-1}[u_{0}]$ is singleton which implies
	$h-h_{0}=\mathbb{A}^{-1}[u_{0}]$, therefore we have $\|\mathbb{A}^{-1}[u_{0}]\|\geq\min\{r_{s},r_{u}\}-\delta$.
	Since $\delta$ can be made arbitrarily small we get:
	
	\[
	\sup_{u_{0}\in\mathcal{U}:\,\|u_{0}\|_{\mathcal{B}_{Z}}\leq b}\|\mathbb{A}^{-1}[u_{0}]\|_{\mathcal{B}_{X}}\geq\min\{r_{s},r_{u}\}
	\]
	
	By supposition $\hat{h}_{n}$ is consistent for $h_{0}$ under Assumptions
	1.1 and 1.2 whenever $h_{0}\in\mathcal{S}$. Under Assumption 1.1,
	$g_{0}=\mathbb{A}[h_{0}]$ and under Assumption 1.2, $\mathbb{A}^{-1}[g_{0}]$
	is singleton and thus $h_{0}=\mathbb{A}^{-1}[g_{0}]$. Therefore whenever
	$\mathbb{A}^{-1}[g_{0}]\in\mathcal{S}$, we have $\underset{n\to\infty}{\text{plim}}\|\hat{h}_{n}-\mathbb{A}^{-1}[g_{0}]\|_{\mathcal{B}_{X}}=0$.
	Now, by definition of $u_{0}$, $g_{0}=\mathbb{A}[h_{0}]+u_{0}$ and
	so:
	\[
	\hat{h}_{n}-h_{0}=\hat{h}_{n}-\mathbb{A}^{-1}[g_{0}]+\mathbb{A}^{-1}[u_{0}]
	\]
	
	By the triangle inequality we get:
	\[
	\|\hat{h}_{n}-h_{0}\|_{\mathcal{B}_{X}}\geq\|\mathbb{A}^{-1}[u_{0}]\|_{\mathcal{B}_{X}}-\|\hat{h}_{n}-\mathbb{A}^{-1}[g_{0}]\|_{\mathcal{B}_{X}}
	\]
	
	Recall that $u_{0}=\mathbb{A}[h-h_{0}]$ for an $h\in\mathcal{S}$.
	Thus we have $\mathbb{A}^{-1}[g_{0}]=h\in\mathcal{S}$ and so $\underset{n\to\infty}{\text{plim}}\|\hat{h}_{n}-\mathbb{A}^{-1}[g_{0}]\|_{\mathcal{B}_{X}}=0$.
	From the above we then get:
	\[
	\underset{n\to\infty}{\text{plim}}\|\hat{h}_{n}-h_{0}\|_{\mathcal{B}_{X}}\geq\|\mathbb{A}^{-1}[u_{0}]\|_{\mathcal{B}_{X}}
	\]
	
	Recall that $\|\mathbb{A}^{-1}[u_{0}]\|\geq\min\{r_{s},r_{u}\}-\delta$
	and so:
	\[
	\underset{n\to\infty}{\text{plim}}\|\hat{h}_{n}-h_{0}\|_{\mathcal{B}_{X}}\geq\min\{r_{s},r_{u}\}-\delta
	\]
	
	Since our choice of $u_{0}$ satisfies $u_{0}\in\mathcal{U}$ and
	$\|u_{0}\|_{\mathcal{B}_{Z}}\leq b$ we have:
	\[
	\sup_{u_{0}\in\mathcal{U}:\,\|u_{0}\|_{\mathcal{B}_{Z}}\leq b}\underset{n\to\infty}{\text{plim}}\|\hat{h}_{n}-h_{0}\|_{\mathcal{B}_{X}}\geq\min\{r_{s},r_{u}\}-\delta
	\]
	
	Since we can set $\delta$ arbitrarily small we get:
	
	\[
	\sup_{u_{0}\in\mathcal{U}:\,\|u_{0}\|_{\mathcal{B}_{Z}}\leq b}\underset{n\to\infty}{\text{plim}}\|\hat{h}_{n}-h_{0}\|_{\mathcal{B}_{X}}\geq\min\{r_{s},r_{u}\}
	\]
	
	Using the definition of the worst-case asymptotic bias then gives the result. \end{proof}

\begin{proof}[Proof of Theorem 1.3]
	
	It is well-known that a continuous and injective function defined
	on a compact set has a continuous inverse. So by compactness of $\mathcal{S}$,
	the restriction of $\mathbb{A}^{-1}$ to $\mathbb{A}[\mathcal{S}]$
	is continuous. $P_{Z}$ is continuous by assumption, and maps into
	$\mathbb{A}[\mathcal{S}]$. Hence $\mathbb{A}^{-1}P_{Z}$ is continuous.
	By continuity, for any $g\in\mathcal{B}_{Z}$: 
	\[
	\lim_{b\to0}\sup_{g'\in\mathcal{B}_{Z}:\,\|g'-g\|\leq b}\|\mathbb{A}^{-1}P_{Z}[g']-\mathbb{A}^{-1}P_{Z}[g]\|_{\mathcal{B}_{X}}=0
	\]
	
	Set $g=\mathbb{A}[h_{0}]$ in the above. Since $h_{0}\in\mathcal{S}$
	we have $\mathbb{A}^{-1}P_{Z}\mathbb{A}[h_{0}]=h_{0}$. Reparameterizing
	in terms of $u_{0}$ and using $g_{0}=u_{0}+\mathbb{A}[h_{0}]$ we
	get: 
	\[
	\lim_{b\to0}\sup_{u_{0}\in\mathcal{B}_{Z}:\,\|u_{0}\|\leq b}\|\mathbb{A}^{-1}P_{Z}[g_{0}]-h_{0}\|_{\mathcal{B}_{X}}=0
	\]
	Further restrictions on $u_{0}$ can only decrease the supremum and
	so:
	\begin{equation}
		\lim_{b\to0}\sup_{u_{0}\in\mathcal{U}:\,\|u_{0}\|\leq b}\|\mathbb{A}^{-1}P_{Z}[g_{0}]-h_{0}\|_{\mathcal{B}_{X}}=0\label{eq:limusm}
	\end{equation}
	
	Now, by the triangle inequality:
	\[
	\|\hat{h}_{n}-h_{0}\|_{\mathcal{B}_{X}}\leq\|\mathbb{A}^{-1}P_{Z}[g_{0}]-h_{0}\|_{\mathcal{B}_{X}}+\|\hat{h}_{n}-\mathbb{A}^{-1}P_{Z}[g_{0}]\|_{\mathcal{B}_{X}}
	\]
	Since $\|\hat{h}_{n}-\mathbb{A}_{\mathcal{H}}^{-1}P_{Z}[g_{0}]\|_{\mathcal{B}_{X}}\to^{p}0$
	it follows that: 
	\begin{equation}
		\underset{n\to\infty}{\text{plim}}\|\hat{h}_{n}-h_{0}\|_{\mathcal{B}_{X}}=\|\mathbb{A}^{-1}P_{Z}[g_{0}]-h_{0}\|_{\mathcal{B}_{X}}\label{eq:cove}
	\end{equation}
	
	Substituting the above into (\ref{eq:limusm}) and using the definition
	of the worst case asymptotic bias gives the result.
	
	For the second part of the theorem we apply Lemma 1.4 with $h^{*}=h_{0}$
	and reparameterize $h=h_{0}+\mathbb{A}^{-1}[u_{0}]$ and use $g_{0}=u_{0}+\mathbb{A}[h_{0}]$
	to get:
	\[
	\lim_{b\to0}\frac{\sup_{u_{0}\in\mathbb{A}[\mathcal{S}]-h_{0},u_{0}\in\mathcal{U}:\,\|u_{0}\|_{\mathcal{B}_{Z}}\leq b}\|\mathbb{A}^{-1}[g_{0}]-h_{0}\|_{\mathcal{B}_{X}}}{b}=\infty
	\]
	
	Note that $u_{0}\in\mathbb{A}[\mathcal{S}-h_{0}]$ and $h\in\mathcal{S}$
	implies $g_{0}\in\mathbb{A}[\mathcal{S}]$ and so $\mathbb{A}^{-1}[g_{0}]=\mathbb{A}^{-1}P_{Z}[g_{0}]$.
	Substituting we get:
	\[
	\lim_{b\to0}\frac{\sup_{u_{0}\in\mathbb{A}[\mathcal{S}]-h_{0},u_{0}\in\mathcal{U}:\,\|u_{0}\|_{\mathcal{B}_{Z}}\leq b}\|\mathbb{A}^{-1}P_{Z}[g_{0}]-h_{0}\|_{\mathcal{B}_{X}}}{b}=\infty
	\]
	
	The supremum in the above must be smaller than the supremum without
	the constraint $u_{0}\in\mathbb{A}[\mathcal{S}]-h_{0}$, and so we
	have:
	\[
	\lim_{b\to0}\frac{\sup_{u_{0}\in\mathcal{U}:\,\|u_{0}\|_{\mathcal{B}_{Z}}\leq b}\|\mathbb{A}^{-1}P_{Z}[g_{0}]-h_{0}\|_{\mathcal{B}_{X}}}{b}=\infty
	\]
	
	Substituting (\ref{eq:cove}) and using the definition of the worst-case
	bias then gives the result.
\end{proof}

\begin{proof}[Proof of Theorem 1.4] Let us introduce notation.
	Let the $L_{2}(\mu_{X})$ and $L_{2}(\mu_{Z})$ inner products be
	$\langle\cdot,\cdot\rangle_{L_{2}(\mu_{X})}$ and $\langle\cdot,\cdot\rangle_{L_{2}(\mu_{Z})}$
	respectively. The linear functional of interest $\mathbb{L}[h_{0}]$
	can then be written as $\mathbb{L}[h_{0}]=\langle w,h_{0}\rangle_{L_{2}(\mu_{X})}$.
	The adjoint of the operator $\mathbb{A}$, denoted $\mathbb{A}^{*}$
	is given by $\mathbb{A}^{*}[g](X)=E[g(Z)|X]$.
	
	First we prove the following: 
	\begin{align}
		bias_{\hat{l}_{n}}(b) & =\sup_{u_{0}\in\mathbb{A}[L_{2}(\mu_{X})]:\,\|u_{0}\|_{L_{2}(\mu_{Z})}\leq b}|E[w(X)\mathbb{A}^{-1}[u_{0}](X)]|\nonumber \\
		& =\sup_{u_{0}\in\mathbb{A}[L_{2}(\mu_{X})]:\,\|u_{0}\|_{L_{2}(\mu_{Z})}\leq b}|\langle w,\mathbb{A}^{-1}[u_{0}]\rangle_{L_{2}(\mu_{X})}|\label{eq:ass1}
	\end{align}
	If instruments are valid then $h_{0}=\mathbb{A}^{-1}[g_{0}]$ and
	$g_{0}\in\mathbb{A}[L_{2}(\mu_{X})]$. So consistency under instrumental
	validity, implies that if $g_{0}\in\mathbb{A}[L_{2}(\mu_{X})]$ then:
	\[
	|\hat{l}_{n}-E[w(X)\mathbb{A}^{-1}[g_{0}](X)]|\to^{p}0
	\]
	Let $u_{0}\in\mathbb{A}[L_{2}(\mu_{X})]$. Using $g_{0}=\mathbb{A}[h_{0}]+u_{0}$:
	\[
	|E[w(X)\mathbb{A}^{-1}[g_{0}](X)]-\mathbb{L}[h_{0}]|=|E[w(X)\mathbb{A}^{-1}[u_{0}](X)]|
	\]
	And so: 
	\[
	\underset{n\to\infty}{\text{plim}}|\hat{l}_{n}-\mathbb{L}[h_{0}]|=|E[w(X)\mathbb{A}^{-1}[u_{0}](X)]|
	\]
	Applying the definition of the worst-case asymptotic bias then gives
	(\ref{eq:ass1}).
	
	Next we show that :
	\begin{align}
		diam(\mathbb{L}(\Theta_{b})) & =2\sup_{u_{0}\in\mathbb{A}[L_{2}(\mu_{X})]:\,\|u_{0}\|_{L_{2}(\mu_{Z})}\leq b}|E[w(X)\mathbb{A}^{-1}[u_{0}](X)]|\nonumber \\
		& =2\sup_{u_{0}\in\mathbb{A}[L_{2}(\mu_{X})]:\,\|u_{0}\|_{L_{2}(\mu_{Z})}\leq b}|\langle w,\mathbb{A}^{-1}[u_{0}]\rangle_{L_{2}(\mu_{X})}|\label{eq:ass2l}
	\end{align}
	
	First we note that by the triangle inequality:
	\begin{align}
		& \sup_{\mathbb{A}^{-1}[g_{0}]+h,\mathbb{A}^{-1}[g_{0}]-h\in\Theta_{b}}|\mathbb{L}[\mathbb{A}^{-1}[g_{0}]+h]-\mathbb{L}[\mathbb{A}^{-1}[g_{0}]-h]|\nonumber \\
		\leq & diam(\mathbb{L}(\Theta_{b}))\nonumber \\
		\leq & 2\sup_{h\in\Theta_{b}}|\mathbb{L}[\mathbb{A}^{-1}[g_{0}]-h]|\label{eq:boundsin}
	\end{align}
	
	Given our choice of $\mathcal{U}$ and $\mathcal{H}$ the identified
	set is symmetric around $\mathbb{A}^{-1}[g_{0}]$, thus $\mathbb{A}^{-1}[g_{0}]-h\in\Theta_{b}$
	implies $\mathbb{A}^{-1}[g_{0}]+h\in\Theta_{b}$. So (using linearity
	of $\mathbb{L}$) the lower bound is equal to 
	\[
	\sup_{\mathbb{A}^{-1}[g_{0}]-h\in\Theta_{b}}|\mathbb{L}[\mathbb{A}^{-1}[g_{0}]+h]-\mathbb{L}[\mathbb{A}^{-1}[g_{0}]-h]|=2\sup_{\mathbb{A}^{-1}[g_{0}]-h\in\Theta_{b}}|\mathbb{L}[h]|
	\]
	
	Using the respresentation of $\mathbb{L}$ and our choice of $\mathcal{U}$
	and $\mathcal{H}$ we have:
	\[
	\sup_{\mathbb{A}^{-1}[g_{0}]-h\in\Theta_{b}}|\mathbb{L}[h]|=\sup_{h\in L_{2}(\mu_{X}):\,\|\mathbb{A}[\mathbb{A}^{-1}[g_{0}]-h]-g_{0}\|_{L_{2}(\mu_{Z})}\leq b}|E[w(X)h(X)]|
	\]
	
	Reparameterizing $u_{0}=\mathbb{A}[h]$ we get:
	\[
	\sup_{\mathbb{A}^{-1}[g_{0}]-h\in\Theta_{b}}|\mathbb{L}[h]|=\sup_{u_{0}\in\mathbb{A}[L_{2}(\mu_{X})]:\,\|u_{0}\|_{L_{2}(\mu_{Z})}\leq b}|E[w(X)\mathbb{A}^{-1}[u_{0}]]|
	\]
	
	And thus:
	\[
	diam(\mathbb{L}(\Theta_{b}))\geq2\sup_{u_{0}\in\mathbb{A}[L_{2}(\mu_{X})]:\,\|u_{0}\|_{L_{2}(\mu_{Z})}\leq b}|E[w(X)\mathbb{A}^{-1}[u_{0}](X)]|
	\]
	
	Now, given our choice of $\mathcal{U}$ and $\mathcal{H}$, the upper
	bound in (\ref{eq:boundsin}) can be written as:
	\[
	2\sup_{h\in L_{2}(\mu_{X}):\,\|\mathbb{A}[h]-g_{0}\|_{L_{2}(\mu_{Z})}\leq b}|E\big[w(X)\big(\mathbb{A}^{-1}[g_{0}]-h(X)\big)\big]|
	\]
	
	Reparameterizing $u_{0}=\mathbb{A}^{-1}[g_{0}]-h(X)$ and using (\ref{eq:boundsin})
	we get:
	\[
	diam(\mathbb{L}(\Theta_{b}))\leq2\sup_{u_{0}\in\mathbb{A}[L_{2}(\mu_{X})]:\,\|u_{0}\|_{L_{2}(\mu_{Z})}\leq b}|E[w(X)\mathbb{A}^{-1}[u_{0}](X)]|
	\]
	
	Thus we have (\ref{eq:ass2l}).
	
	Now let us now prove claim a. Suppose that for some $\alpha$, $w(X)=E[\alpha(Z)|X]$,
	equivalently $w=\mathbb{A}^{*}[\alpha]$. Then for any $g\in\mathbb{A}[L_{2}(\mu_{X})]$:
	\begin{align*}
		\langle w,\mathbb{A}^{-1}[g]\rangle_{L_{2}(\mu_{X})}=\langle\mathbb{A}^{*}[\alpha],\mathbb{A}^{-1}[g]\rangle_{L_{2}(\mu_{X})} & =\langle\alpha,\mathbb{A}\mathbb{A}^{-1}[g]\rangle_{L_{2}(\mu_{Z})}
	\end{align*}
	Which simply equals $\langle\alpha,g\rangle_{L_{2}(\mu_{Z})}$. Suppose
	$\|g\|_{L_{2}(\mu_{Z})}\leq b$, by Cauchy-Schwartz: 
	\begin{align*}
		|\langle w,\mathbb{A}^{-1}[g]\rangle_{L_{2}(\mu_{X})}|\leq\|g\|_{L_{2}(\mu_{Z})}\|\alpha\|_{L_{2}(\mu_{Z})}\leq b\|\alpha\|_{L_{2}(\mu_{Z})}
	\end{align*}
	Combining the above with (\ref{eq:ass1}) and (\ref{eq:ass2l}) gives
	the result.
	
	We now prove claim b. Given (\ref{eq:ass1}) and (\ref{eq:ass2l})
	it is enough to show that for any $b>0$:
	\[
	\sup_{u_{0}\in\mathbb{A}[L_{2}(\mu_{X})]:\,\|u_{0}\|_{L_{2}(\mu_{Z})}\leq b}|\langle w,\mathbb{A}^{-1}[u_{0}]\rangle_{L_{2}(\mu_{X})}|=\infty
	\]
	Suppose that for some $\bar{b}>0$ the worst-case asymptotic bias
	is finite. That is, there exists a scalar $c<\infty$ so that: 
	\[
	\sup_{u_{0}\in\mathbb{A}[L_{2}(\mu_{X})]:\,\|u_{0}\|_{L_{2}(\mu_{Z})}\leq\bar{b}}|\langle w,\mathbb{A}^{-1}[u_{0}]\rangle_{L_{2}(\mu_{X})}|\leq c
	\]
	Linearity of $\mathbb{A}^{-1}$ and the inner-product then implies
	that for all $b>0$: 
	\[
	\sup_{u_{0}\in\mathbb{A}[L_{2}(\mu_{X})]:\,\|u_{0}\|_{L_{2}(\mu_{Z})}\leq b}|\langle w,\mathbb{A}^{-1}[u_{0}]\rangle_{L_{2}(\mu_{X})}|\leq\frac{c}{\bar{b}}b
	\]
	Also by linearity, the LHS above equals: 
	\[
	\sup_{u_{0}\in\mathbb{A}[L_{2}(\mu_{X})]:\,\|u_{0}\|_{L_{2}(\mu_{Z})}\leq b}\langle w,\mathbb{A}^{-1}[u_{0}]\rangle_{L_{2}(\mu_{X})}
	\]
	Define the function $D:\mathbb{A}[L_{2}(\mu_{X})]\to\mathbb{R}$ by
	$D[g]=\langle w,\mathbb{A}^{-1}[g]\rangle_{L_{2}(\mu_{X})}$. Then:
	\[
	\sup_{u_{0}\in\mathbb{A}[L_{2}(\mu_{X})]:\,\|u_{0}\|_{L_{2}(\mu_{Z})}\leq b}|D[u_{0}]|\leq\frac{c}{\bar{b}}b
	\]
	By the Hahn-Banach theorem we can extend $D$ to a bounded linear
	function $\bar{D}$ defined on the whole space $L_{2}(\mu_{Z})$ which
	then satisfies: 
	\[
	\sup_{u_{0}\in L_{2}(\mu_{Z}):\,\|u_{0}\|_{L_{2}(\mu_{Z})}\leq b}|\bar{D}[u_{0}]|\leq\frac{c}{\bar{b}}b
	\]
	Since $\bar{D}$ is a bounded linear functional defined on a Hilbert
	space, by the Reisz representation theorem there exists an element
	$\alpha\in L_{2}(\mu_{Z})$ so that for all $g\in L_{2}(\mu_{Z})$,
	$\bar{D}[g]=\langle\alpha,g\rangle_{L_{2}(\mu_{Z})}$. And so for
	any $u_{0}\in\mathbb{A}[L_{2}(\mu_{X})]$: 
	\begin{align*}
		D[u_{0}]=\langle\alpha,u_{0}\rangle_{L_{2}(\mu_{Z})}=\langle\mathbb{A}^{*}[\alpha],\mathbb{A}^{-1}[u_{0}]\rangle_{L_{2}(\mu_{X})}=\langle w,\mathbb{A}^{-1}[u_{0}]\rangle_{L_{2}(\mu_{X})}
	\end{align*}
	Where the final equality follows by the definition of $D$.
	
	Since the equality above holds for all $u_{0}\in\mathbb{A}[L_{2}(\mu_{X})]$,
	for all $h\in L_{2}(\mu_{X})$ we have (using bi-linearity of the
	inner product) $\langle\mathbb{A}^{*}[\alpha]-w,h\rangle_{L_{2}(\mu_{X})}=0$.
	But we can set $h=\mathbb{A}^{*}[\alpha]-w$ and the above implies
	that the norm of $\mathbb{A}^{*}[\alpha]-w$ equals zero and so we
	have that $\mathbb{A}^{*}[\alpha]-w=0$. Or equivalently $w(X)=E[\alpha(Z)|X]$.
\end{proof}

\subsection{Proofs of Results in Section 2}

\newtheorem*{P2.1}{Proposition 2.1}

\begin{P2.1}
	Under Assumption 2.1, for any linear operator $\mathbb{L}:\mathcal{B}_X\to\mathbb{R}$, 
	$\Theta_{\mathbb{L}}=[\mathbb{L}[h]\in\mathbb{R}:h\in\Theta]$ is an interval.
\end{P2.1}
\begin{proof}
	The constraints (2.1) and (2.2) are clearly convex and therefore so is $\Theta$.
	Suppose $\theta',\theta''\in\Theta_{\mathbb{L}}$, then $\exists h',h''\in\mathcal{B}_{X}$ with $\mathbb{L}[h']=\theta'$, $\mathbb{L}[h'']=\theta''$ and $h',h''\in\Theta$. Since $\Theta$ is convex $h'''=\alpha h'+(1-\alpha)h''\in\Theta$ for any
	$\alpha\in[0,1]$. $\mathbb{L}$ is linear, so we then have $\alpha\theta'+(1-\alpha)\theta''=\mathbb{L}[h''']\in\Theta_{\mathbb{L}}$. So $\Theta_{\mathbb{L}}$ is a convex subset of $\mathbb{R}$, i.e., an interval. 
\end{proof}

\newtheorem*{2.1}{Lemma 2.1}
\begin{2.1}
	For each $j\in[J]$ let $\mathbb{M}_{j}$ be a linear operator from a vector space $\mathcal{S}$
	to the space of functions from $\mathcal{V}_j$ to $\mathbb{R}$ and let $b_j$ be a function from $\mathcal{V}_j$ to $\mathbb{R}$.\footnote{$[J]$ denote the set of natural numbers from $1$ to $J$.} Let $\mathbb{L}$ be a linear operator that maps from $\mathcal{S}$ to $\mathbb{R}$. Define:
	\begin{align*}
		\bar{\theta}_{\mathbb{L}} & =\sup_{s\in\mathcal{S}}\mathbb{L}[s] \text{ s.t. }  \mathbb{M}_{j}[s](v)\leq b_j(v),\forall v\in\mathcal{V}_j,\forall j\in[J]
	\end{align*}
	And $\underline{\theta}_{\mathbb{L}}$ is the infimum subject the same constraints.
	
	Consider some $\hat{s}\in\mathcal{S}$ so that for some $r>0$ and all $j\in[J]$, $\mathbb{M}_{j}[\hat{s}]\leq b_j(v)+r$,$\forall v \in \mathcal{V}_j$. Suppose there exists $\epsilon>0$ and $\tilde{s}\in\mathcal{S}$ so that for each $j\in[J]$, $\mathbb{M}_{j}[\tilde{s}](v)\leq b_j(v)-\epsilon$, $\forall v \in \mathcal{V}_j$. Then $\mathbb{L}[\hat{s}]-\bar{\theta}_{\mathbb{L}}\leq \frac{r}{\epsilon}(\bar{\theta}_{\mathbb{L}}-\mathbb{L}[\tilde{s}])\leq\frac{r}{\epsilon}(\bar{\theta}_{\mathbb{L}}-\underline{\theta}_{\mathbb{L}})$.
\end{2.1}
\begin{proof}
	Define $s^*\in\mathcal{S}$ by $s^*=\frac{\epsilon}{r+\epsilon}\hat{s}+\frac{r}{r+\epsilon}\tilde{s}$.
	By the linearity of $\mathbb{M}_{j}$, for  $j\in[J]$:
	\begin{align*}
		\mathbb{M}_{j}[s^*](v) & =\frac{\epsilon}{r+\epsilon}\mathbb{M}_{j}[\hat{s}](v)+\frac{r}{r+\epsilon}\mathbb{M}_{j}[\tilde{s}](v)\\
		& \leq\frac{\epsilon}{r+\epsilon}(r+b_j(v))-\frac{r}{r+\epsilon}(\epsilon-b_j(v)) =b_j(v)
	\end{align*}
	So $s^*$ satisfies the constraints in the problems for $\bar{\theta}_{\mathbb{L}}$
	and $\underline{\theta}_{\mathbb{L}}$, so we must have 
	$ \mathbb{L}[s^*]\leq\bar{\theta}_{\mathbb{L}}$. Substituting the definition of $s^*$ and using linearity of $\mathbb{L}$ we get:
	\[
	\mathbb{L}[\hat{s}]-\overline{y}\leq \frac{r}{r+\epsilon}\mathbb{L}[\hat{s}-\tilde{s}]= \frac{r}{r+\epsilon}(\bar{\theta}_{\mathbb{L}}-\mathbb{L}[\tilde{s}])+\frac{r}{r+\epsilon}(\mathbb{L}[\hat{s}]-\bar{\theta}_{\mathbb{L}})
	\]
	Subtracting $\frac{r}{r+\epsilon}(\mathbb{L}[\hat{s}]-\bar{\theta}_{\mathbb{L}})$ from both sides and then dividing by $\frac{\epsilon}{r+\epsilon}$ we get $\mathbb{L}[\hat{s}]-\bar{\theta}_{\mathbb{L}}\leq\frac{r}{\epsilon}(\bar{\theta}_{\mathbb{L}}-\mathbb{L}[\tilde{s}])$. Noting $\underline{\theta}_{\mathbb{L}}\leq \mathbb{L}[\tilde{s}]$ gives the conclusion. \end{proof}

\newtheorem*{2.2}{Lemma 2.2}
\begin{2.2}
	For each $j\in[J]$ let $\mathbb{M}_{j}$ be a linear operator from a vector space $\mathcal{S}$ with norm $\|\cdot\|_{\mathcal{S}}$ 
	to the space of functions from $\mathcal{V}_j$ to $\mathbb{R}$, and let $b_j$ be a function from $\mathcal{V}_j$ to $\mathbb{R}$. Let $\mathbb{L}$ be a linear operator from $\mathcal{S}$ to $\mathbb{R}$. Define:
	\begin{align*}
		\bar{\theta}_{\mathbb{L}} & =\sup_{s\in\mathcal{S}}\mathbb{L}[s] \text{ s.t. }  \mathbb{M}_{j}[s](v)\leq b_j(v),\forall v\in\mathcal{V}_j,\forall j\in[J]
	\end{align*}
	And let $\underline{\theta}_{\mathbb{L}}$ be the infimum subject to the same constraints.
	
	Consider a subset $\tilde{\mathcal{S}}\subseteq \mathcal{S}$. Suppose: i. There exists $\mathcal{J}\subseteq[J]$ so that for any $s\in\mathcal{S}$, if $\mathbb{M}_{j}[s](v)\leq b_j(v),\forall v\in\mathcal{V}_j, j\in\mathcal{J}$, then there exists $s'\in\tilde{\mathcal{S}}$ so that $\|s-s'\|_{\mathcal{S}}\leq \omega$ and $\mathbb{M}_{j}[s'](v)\leq \mathbb{M}_{j}[s](v),\forall v\in\mathcal{V}_j, j\in\mathcal{J}$. ii. For each $j\in\mathcal{J}^c=[J]\setminus\mathcal{J}$ the linear operators of the form $s\mapsto\mathbb{M}_j[s](v)$ for each $v\in\mathcal{V}_j$ have operator norm less than $c'$. iii. There exists $\epsilon>0$ and $\tilde{s}\in\mathcal{S}$ so that for each $j\in[J]$,
	$\mathbb{M}_{j}[\tilde{s}](v)\leq b_j(v)-\epsilon$,
	$\forall v\in\mathcal{V}_j$. iv. $\mathbb{L}$ is a continuous linear operator with operator norm $c''$.
	
	Under conditions i., ii., and iv., there exists $\tilde{s}\in\tilde{\mathcal{S}}$ so that for each $j\in[J]$, $\mathbb{M}_{j}[\tilde{s}](v)\leq b_j(v)+c'\omega-\epsilon$, $\forall v\in\mathcal{V}_j$.
	
	Define $\tilde{\bar{\theta}}_{\mathbb{L}}$ by:
	\begin{align*}
		\tilde{\bar{\theta}}_{\mathbb{L}} & =\sup_{s\in\tilde{\mathcal{S}}}\mathbb{L}[s] \text{ s.t. }  \mathbb{M}_{j}[s](v)\leq b_j(v),\forall v\in\mathcal{V}_j,\forall j\in[J]\end{align*}
	Under conditions i. ii., iii., and iv., $|\bar{\theta}_{\mathbb{L}}-\tilde{\bar{\theta}}_{\mathbb{L}}|\leq  \omega\big(\frac{c'}{\epsilon}(\bar{\theta}_{\mathbb{L}}-\underline{\theta}_{\mathbb{L}})+c''\big)$.
\end{2.2}
\begin{proof}
	By conditions i. and iii. there exists an $\tilde{s}'\in\tilde{\mathcal{S}}$ so that $\|\tilde{s}'-\tilde{s}\|_{\mathcal{S}}\leq \omega$ and for $j\in\mathcal{J}$, $\mathbb{M}_{j}[\tilde{s}'](v)\leq b_j(v)-\epsilon$ for each $v\in\mathcal{V}_j$. By ii.,  for any $s_1,s_2\in\mathcal{S}$ and all $v\in\mathcal{V}_j$, $|\mathbb{M}_{j}[s_1-s_2](v)|\leq c'\|s_1-s_2\|_{\mathcal{S}}$. Therefore for $j\in\mathcal{J}^c$, $|\mathbb{M}_{j}[\tilde{s}-\tilde{s}'](v)|\leq c'\omega$. By the triangle inequality and $\mathbb{M}_{j}[\tilde{s}](v)\leq b_j(v)-\epsilon$:
	\[
	\mathbb{M}_{j}[\tilde{s}'](v)\leq b_j(v)-\epsilon +c'\omega
	\]
	So the first claim of the lemma holds, now we prove the second claim. By definition of the supremum, for any $r>0$ there must be some $s'\in\mathcal{S}$ that satisfies the constraints in the problem for $\bar{\theta}_{\mathbb{L}}$ and achieves $\bar{\theta}_{\mathbb{L}}-\mathbb{L}[s']\leq r$. By the same reasoning used to establish the first claim there must exist a $s''\in\tilde{\mathcal{S}}$ with $\|s''-s'\|_{\mathcal{S}}\leq \omega$ and for $j\in\mathcal{J}$, $\mathbb{M}_{j}[s''](v)\leq b_j(v)$ and for $j\in\mathcal{J}^c$, $\mathbb{M}_{j}[s''](v)\leq b_j(v)+c'\omega$. Let $s^*=\frac{\epsilon-c'\omega}{\epsilon}s''+\frac{c'\omega}{\epsilon}\tilde{s}'$. Since $\tilde{\mathcal{S}}$ is linear $s^*\in\tilde{S}$. Moreover, by linearity of $M_j$ we have get that for $j\in\mathcal{J}$ $\mathbb{M}_{j}[s^*](v)\leq b_j(v)$ and for $j\in\mathcal{J}^c$:
	\[
	\mathbb{M}_{j}[s^*](v)\leq \frac{\epsilon-c'\omega}{\epsilon}(b_j(v)+c'\omega)+\frac{c'\omega}{\epsilon}(b_j(v)-\epsilon+c'\omega)=b_j(v)
	\]
	So  $s^*$ satisfies the constraints of the problems for $\bar{\theta}_{\mathbb{L}}$ and $\tilde{\bar{\theta}}_{\mathbb{L}}$ so $\bar{\theta}_{\mathbb{L}}\geq\tilde{\bar{\theta}}_{\mathbb{L}}\geq \mathbb{L}[s^*]$. Now, by iv., for any $s_1,s_2\in\mathcal{S}$, $|\mathbb{L}[s_1-s_2]|\leq c''\|s_1-s_2\|_{\mathcal{S}}$. Therefore using the definition of $s^*$:
	\[
	\mathbb{L}[s'-s^*]=\mathbb{L}[s'-s'']+\mathbb{L}[s'-s^*]\leq c''\omega+\frac{c'\omega}{\epsilon}\mathbb{L}[s'-\tilde{s}']
	\]
	And so $\bar{\theta}_{\mathbb{L}}-\mathbb{L}[s^*]\leq r + c''\omega+\frac{c'\omega}{\epsilon}\mathbb{L}[s'-\tilde{s}']$. Since $s'$ and $\tilde{s}'$ satisfy the constraints of the problems for $\bar{\theta}_{\mathbb{L}}$ and $\underline{\theta}_{\mathbb{L}}$, $\mathbb{L}[s']\leq\bar{\theta}_{\mathbb{L}}$ and $\mathbb{L}[\tilde{s}']\geq\underline{\theta}_{\mathbb{L}}$, and recall $\tilde{\bar{\theta}}_{\mathbb{L}}\geq \mathbb{L}[s^*]$ so (using linearity of $\mathbb{L}$) we get $\bar{\theta}_{\mathbb{L}}-\tilde{\bar{\theta}}_{\mathbb{L}}\leq r + c''\omega+\frac{c'\omega}{\epsilon}(\bar{\theta}_{\mathbb{L}}-\underline{\theta}_{\mathbb{L}})$. Since this holds for any $r>0$ we get  $\bar{\theta}_{\mathbb{L}}-\tilde{\bar{\theta}}_{\mathbb{L}}\leq c''\omega+\frac{c'\omega}{\epsilon}(\bar{\theta}_{\mathbb{L}}-\underline{\theta}_{\mathbb{L}})$. Finally, note that the constraints of the problem for $\tilde{\bar{\theta}}_{\mathbb{L}}$ are stronger than those for $\bar{\theta}_{\mathbb{L}}$ (because $\tilde{\mathcal{S}}\subseteq\mathcal{S}$) so $\bar{\theta}_{\mathbb{L}}-\tilde{\bar{\theta}}_{\mathbb{L}}\geq 0$. \end{proof}

\newtheorem*{2.3}{Lemma 2.3}
\begin{2.3}
	For each $j\in[J]$ let $\mathbb{M}_{j}$ be a linear operator from a vector space $\mathcal{S}$
	to the space of functions from $\mathcal{V}_j$ to $\mathbb{R}$. Let $\mathbb{L}$ be a linear operator that maps from $\mathcal{S}$ to $\mathbb{R}$. Define:
	\begin{align*}
		\bar{\theta}_{\mathbb{L}} & =\sup_{s\in\mathcal{S}}\mathbb{L}[s] \text{ s.t. }  \mathbb{M}_{j}[s](v)\leq b_j(v),\forall v\in\mathcal{V}_j,\forall j\in[J]
	\end{align*}
	And let $\underline{\theta}_{\mathbb{L}}$ be the infimum subject to the same constraints. For each $j$ let $\tilde{\mathcal{V}}_j$ be a finite subset of $\mathcal{V}_j$ and define:
	\begin{align*}
		\tilde{\bar{\theta}}_{\mathbb{L}} & =\sup_{s\in\mathcal{S}}\mathbb{L}[s] \text{ s.t. }  \mathbb{M}_{j}[s](v)\leq b_j(v),\forall v\in\tilde{\mathcal{V}}_j,\forall j\in[J]
	\end{align*}
	
	Suppose there exists $\tilde{s}\in\mathcal{S}$ and $\epsilon>0$
	so that for each $j\in[J]$ we have $\mathbb{M}_{j}[\tilde{s}](v)\leq b_j(v)-\epsilon$, $\forall v \in\mathcal{V}_j$. Suppose for each $j$, $\mathcal{V}_j$ is a subset of Euclidean space, $b_j(\cdot)$ is Lipschitz continuous with constant at most $\xi_j$, and that if $s\in\mathcal{S}$ satisfies the constraints in the problem for $\tilde{\bar{\theta}}_{\mathbb{L}}$ then 
	$\mathbb{M}_{j}[s](\cdot)$ is Lipschitz continuous with constant at most $\xi_j$. 
	Then $|\tilde{\bar{\theta}}_{\mathbb{L}}- \bar{\theta}_{\mathbb{L}}|\leq \frac{2\sum^J_{j=1}\xi_j D_j}{\epsilon}(\bar{\theta}_{\mathbb{L}}-\underline{\theta}_{\mathbb{L}})$, where $D_j= \underset{v\in\mathcal{V}_j}{\sup}\underset{v'\in\tilde{\mathcal{V}_j}}{\min}\|v-v'\|_2$.
\end{2.3}
\begin{proof}
	By definition of the supremum, for any $r>0$ there must be some $s'\in\mathcal{S}$ that satisfies the constraints in the problem for $\tilde{\bar{\theta}}_{\mathbb{L}}$ and achieves $\tilde{\bar{\theta}}_{\mathbb{L}}-\mathbb{L}[s']\leq r$. By supposition $\mathbb{M}_{j}[s']$ is Lipschitz continuous with constant at most $\xi_j$ and likewise for $b_j$. Thus $\mathbb{M}_{j}[s']-b_j$ is Lipschitz with constant at most $2\xi_j$. This implies:
	\begin{align*}
		&\sup_{v\in\mathcal{V}_j}\big(\mathbb{M}_{j}[s'](v)-b_j(v)\big)-\max_{v\in\tilde{\mathcal{V}_j}}\big(\mathbb{M}_{j}[s'](v)-b_j(v)\big)\\
		\leq &2\xi_j \sup_{v\in\mathcal{V}_j}\min_{v'\in\tilde{\mathcal{V}_j}}\|v-v'\|_2=2\xi_j D_j\leq 2\sum^J_{j=1}\xi_j D_j
	\end{align*}
	We know that $\mathbb{M}_{j}[s'](v)\leq b_j(v)$ for all $v\in\mathcal{V}_j$, so from the above we get that $\mathbb{M}_{j}[s'](v)\leq b_j(v) +2\sum^J_{j=1}\xi_j D_j$, $\forall v\in\mathcal{V}_j$. Applying Lemma 2.1 we then get $\mathbb{L}[s']- \bar{\theta}_{\mathbb{L}}\leq \frac{2\sum^J_{j=1}\xi_j D_j}{\epsilon}(\bar{\theta}_{\mathbb{L}}-\mathbb{L}[\tilde{s}])$ which implies  $\tilde{\bar{\theta}}_{\mathbb{L}}- \bar{\theta}_{\mathbb{L}}\leq r+ \frac{2\sum^J_{j=1}\xi_j D_j}{\epsilon}(\bar{\theta}_{\mathbb{L}}-\underline{\theta}_{\mathbb{L}})$. Since this holds for any $r$ we get $\tilde{\bar{\theta}}_{\mathbb{L}}- \bar{\theta}_{\mathbb{L}}\leq \frac{2\sum^J_{j=1}\xi_j D_j}{\epsilon}(\bar{\theta}_{\mathbb{L}}-\underline{\theta}_{\mathbb{L}})$. Finally, the constraints in the problem for $\tilde{\bar{\theta}}_{\mathbb{L}}$ are weaker than for $\bar{\theta}_{\mathbb{L}}$, so $\tilde{\bar{\theta}}_{\mathbb{L}}- \bar{\theta}_{\mathbb{L}}\geq 0 $.
\end{proof}

\newtheorem*{2.0}{Theorem 2.0}
\begin{2.0}
	Let $\mathcal{B}_X$ be the space of real valued functions on the support of $X$ equipped with the supremum norm.  Define $\bar{\theta}_{\mathbb{L}}$ and $\underline{\theta}_{\mathbb{L}}$ as follows:
	\begin{align*}
		&\bar{\theta}_{\mathbb{L}}  =\sup_{h\in\mathcal{B}_X}\mathbb{L}[h],\,  \underline{\theta}_{\mathbb{L}}   =\inf_{h\in\mathcal{B}_X}\mathbb{L}[h] \text{ s.t. }\mathbb{T}[h](x)\leq C(x) ,\forall x\in\mathcal{X}\\
		&\text{and } b_{1,k}(Z_k) \leq E[Y-h(X)|Z_k]\leq b_{2,k}(Z_k),\forall k\in[K]
	\end{align*}
	Define estimates $\hat{\bar{\theta}}_{\mathbb{L}}$ and $\hat{\underline{\theta}}_{\mathbb{L}}$ by:
	\begin{align*}
		&\hat{\bar{\theta}}_{\mathbb{L}}  =\max_{\beta\in\mathbb{R}^{K_n}}\mathbb{L}[\Phi_n']\beta,\,  \hat{\underline{\theta}}_{\mathbb{L}}   =\min_{\beta\in\mathbb{R}^{K_n}}\mathbb{L}[\Phi_n']\beta\\
		&\text{s.t. }\mathbb{T}[\Phi_n'](x)\beta\leq C(x) ,\forall x\in\mathcal{X}_n\\
		&\text{and } b_{1,k}(z) \leq \hat{g}_{n,k}(z)-\hat{\Pi}_{n,k}(z)'\beta\leq b_{2,k}(z),\forall z\in\mathcal{Z}_{n,k}, k\in[K]
	\end{align*}
	Suppose Assumptions 2.1.i, 2.2, and 2.4 hold, and for each $k\in[K]$ Assumption 2.3 holds with $g_0(Z)$ replaced by $g_{0,k}(Z_k)=E[Y|Z_k]$, $\Pi_n(Z)$ by $\Pi_{n,k}(Z)=E[\Phi_n(X)|Z_k]$, $\hat{g}_n$ by $\hat{g}_{n,k}$, and $\hat{\Pi}_n$ by $\hat{\Pi}_{n,k}$. Suppose for each $k\in[K]$, $b_{1,k}$ and $b_{2,k}$ are Lipschitz continuous with constant at most $\xi_n$. Finally suppose there exists $\tilde{h}\in\mathcal{B}_X$ that satisfies the constraints of the problem for $\bar{\theta}_\mathbb{L}$ with slack $\epsilon$, that is $\mathbb{T}[h](x)\leq C(x)-\epsilon$ for all $x\in\mathcal{X}$ and $b_{1,k}(Z_k)+\epsilon \leq E[Y-h(X)|Z_k]\leq b_{2,k}(Z_k)-\epsilon$ for all $k\in[K]$. Then uniformly over all linear functionals $\mathbb{L}$ with operator norm less than $c_\mathbb{L}$:
	\[
	|\bar{\theta}_{\mathbb{L}}-\hat{\bar{\theta}}_{\mathbb{L}}|=O_p(\kappa_n+ a_n+C_n\xi_n D_{1,n}+C_nG_n D_{2,n})
	\]
	And likewise for $|\underline{\theta}_{\mathbb{L}}-\underline{\bar{\theta}_{\mathbb{L}}}|$.
\end{2.0}

\begin{proof}
	Define $\tilde{a}_n$ by:
	\begin{equation}
		\tilde{a}_n=\max_k\{|{g}_{k}-{g}_{n,k}|_\infty + \sup_{\beta:\Phi_n'\beta\in\mathcal{H}}|(\hat{\Pi}_{n,k}-{\Pi}_{n,k})'\beta|_\infty\}\label{convr}
	\end{equation}
	By Assumption 2.3, $\tilde{a}_n=O_p(a_n)=o_p(1)$. We suppose, until stated otherwise, that $\kappa_n+2\tilde{a}_n\leq \epsilon/2$, by Assumptions 2.3 and 2.4 this holds with probability approaching $1$. We also suppose that $n$ is sufficiently large that $C_n\xi_n D_{1,n}\leq 1/2$, this must be true for sufficiently large $n$ by Assumption 2.5.iii. 
	
	Define $\bar{\theta}_{\mathbb{L}}^*$, and $\bar{\theta}_{\mathbb{L}}^\circ$ as follows:
	\begin{align*}
		&\bar{\theta}_{\mathbb{L}}^*  =\sup_{\beta\in\mathbb{R}^{K_n}}\mathbb{L}[\Phi_n']\beta \text{ s.t. } \mathbb{T}[\Phi_n'](x)\beta\leq C(x) ,\forall x\in\mathcal{X}\\
		&\text{and } b_{1,k}(Z_k) \leq E[Y-\Phi_{n,k}(X)'\beta|Z_k]\leq b_{2,k}(Z_k),\forall k\in[K]
	\end{align*}
	\begin{align*}
		&\bar{\theta}_{\mathbb{L}}^\circ  =\sup_{\beta\in\mathbb{R}^{K_n}}\mathbb{L}[\Phi_n']\beta\\
		& \text{s.t. }\mathbb{T}[\Phi_n'](x)\beta\leq C(x) ,\forall x\in\mathcal{X}\\
		&\text{and } b_{1,k}(Z_k) \leq \hat{g}_{n,k}(Z_j)-\hat{\Pi}_{n,k}(Z_k)'\beta\leq b_{2,k}(Z_k),\forall k\in[K]
	\end{align*}
	And let $\underline{\theta}_{\mathbb{L}}^*$ and $\underline{\theta}_{\mathbb{L}}^\circ$ be the respective infima subject to the same constraints. By the triangle inequality:
	\begin{equation}
		|\bar{\theta}_{\mathbb{L}}-\tilde{\bar{\theta}}_{\mathbb{L}}|\leq|\bar{\theta}_{\mathbb{L}}-\bar{\theta}_{\mathbb{L}}^*|+|\bar{\theta}_{\mathbb{L}}^*-\bar{\theta}_{\mathbb{L}}^\circ|+|\bar{\theta}_{\mathbb{L}}^\circ-\tilde{\bar{\theta}}_{\mathbb{L}}|\label{trianglee}
	\end{equation}
	
	We use Lemma 2.2 to bound $|\bar{\theta}_{\mathbb{L}}-\bar{\theta}_{\mathbb{L}}^*|$. To apply Lemma 2.2 let $\mathcal{S}=\mathcal{B}_X$ equipped with the supremum norm and we take $\tilde{\mathcal{S}}$ to be the functions of the form $\Phi_n'\beta$. For $j\in\mathcal{J}=[d]$ let $\mathcal{V}_j=\mathcal{X}$, let $\mathbb{M}_j[h](x)=[\mathbb{T}[h](x)]_{j}$, and $b_j[h](x)=[C(x)]_{j}$, where $[v]_j$ denotes the $j^{th}$ component of a vector $v$. Then Assumption 2.4 implies condition i. of Lemma 2.2 with $\omega=\kappa_n$. To establish condition ii., of Lemma 2.2, note that the linear operators of the form $h\mapsto E[h(X)|Z_{k}=z]$ and $h\mapsto -E[h(X)|Z_{k}=z]$ have operator norm of unity because $|E[h(X)|Z_k]|\leq \sup_{x\in\mathcal{X}}|h(x)|$. So condition ii., of the lemma holds with $c'=1$. Conditions iii. of Lemma 2.2 holds be supposition and iv. with $c''=c_\mathbb{L}$.  The second statement of Lemma 2.2 gives:
	\begin{equation}
		|\bar{\theta}_{\mathbb{L}}-\bar{\theta}_{\mathbb{L}}^*|\leq \kappa_n\big(\frac{1}{\epsilon}(\bar{\theta}_{\mathbb{L}}-\underline{\theta}_{\mathbb{L}})+c_{\mathbb{L}}\big)\label{term1}
	\end{equation}
	Moreover, the first statement in Lemma 2.2 tells us there exists $\tilde{\beta}$ so that for each $j\in[J]$, $\mathbb{M}_j[\Phi_n'\tilde{\beta}](v)\leq b_j(v)-\epsilon+\kappa_n$ for all $v\in\mathcal{V}_j$. We will use this when we employ Lemma 2.1 below.
	
	We now use Lemma 2.1 to bound $|\bar{\theta}_{\mathbb{L}}^*-\bar{\theta}_{\mathbb{L}}^\circ|$. By definition of the supremum we can find a $\beta^*$ and $\beta^\circ$ that respectively satisfy the constraints of the problems for $\bar{\theta}_{\mathbb{L}}^*$ and $\bar{\theta}_{\mathbb{L}}^\circ$ so that $\bar{\theta}_{\mathbb{L}}^*-\mathbb{L}[\Phi_n']\beta^*\leq r$ and $\bar{\theta}_{\mathbb{L}}^\circ-\mathbb{L}[\Phi_n']\beta^\circ\leq r$. Because they satisfy the constraints we have that $\Phi_n'\beta^*,\Phi_n'\beta^\circ\in\mathcal{H}$.  Recall $g_k(Z)=E[Y|Z_k]$ and $\Pi_{n,k}(Z_k)=E[\Phi_n(X)|Z_k]$. Since $\hat{\Pi}_{n,k}(Z_k)'\beta^\circ\leq \hat{g}_{n,k}(Z_k)-b_{1,k}(Z_k)$ we get from (\ref{convr}) that:
	\[E[\Phi_{n,k}(X)|Z_k]'\beta^\circ\leq E[Y|Z_k]-b_{1,k}(Z_k)+\tilde{a}_n\]
	And similarly:
	\[-E[\Phi_{n,k}(X)|Z_k]'\beta^\circ\leq-E[Y|Z_k]+b_{2,k}(Z_k)+\tilde{a}_n\]
	We established earlier that there exists $\tilde{\beta}$ with $\mathbb{M}_j[\Phi_n'\tilde{\beta}](v)\leq b_j(v)-\epsilon+\kappa_n$ for each $v\in\mathcal{V}_j$ and $j\in[J]$ (where $\mathbb{M}$ and $b_j$ were defined above). Because $\kappa_n+2\tilde{a}_n\leq \epsilon/2$ we have $\epsilon>\kappa_n$, so we can apply Lemma 2.1 to get $\mathbb{L}[\Phi_n\beta^\circ]-\bar{\theta}_{\mathbb{L}}^*\leq\frac{\tilde{a}_n}{\epsilon-\kappa_n}(\bar{\theta}_{\mathbb{L}}^*-\underline{\theta}_{\mathbb{L}}^*)$ and so $\bar{\theta}_{\mathbb{L}}^\circ-\bar{\theta}_{\mathbb{L}}^*\leq r+\frac{\tilde{a}_n}{\epsilon-\kappa_n}(\bar{\theta}_{\mathbb{L}}^*-\underline{\theta}_{\mathbb{L}}^*)$. Because this holds for each $r>0$: 
	\begin{equation}
		\bar{\theta}_{\mathbb{L}}^\circ-\bar{\theta}_{\mathbb{L}}^*\leq \frac{\tilde{a}_n}{\epsilon-\kappa_n}(\bar{\theta}_{\mathbb{L}}^*-\underline{\theta}_{\mathbb{L}}^*)\label{oneside}
	\end{equation}
	Now, again using (\ref{convr}) we get see that $\tilde{\beta}$ satisfies the constraints of the problem for $\theta^\circ$ with slack $\epsilon-\kappa_n-\tilde{a}_n$ for each $v\in\mathcal{V}_j$ and $j\in[J]$. Moreover, since $E[\Phi_{k}(X)|Z_k]'\beta^*\leq E[Y|Z_k]-b_{1,k}(Z_k)$ we 
	get from (\ref{convr}) that:
	\[\hat{\Pi}_{n,k}(Z_k)'\beta^*\leq \hat{g}_{n,k}(Z_k)-b_{1,k}(Z_k)+\tilde{a}_n\]
	\[-\hat{\Pi}_{n,k}(Z_k)'\beta^*\leq-\hat{g}_{n,k}(Z_k)+b_{2,k}(Z_k)+\tilde{a}_n\]
	Recall $\kappa_n+2\tilde{a}_n\leq \epsilon/2$ and so $\epsilon>\kappa_n+\tilde{a}_n$, then by Lemma 2.1 $\mathbb{L}[\Phi_n\beta^*]-\bar{\theta}_{\mathbb{L}}^\circ\leq\frac{\tilde{a}_n}{\epsilon-\kappa_n-\tilde{a}_n}(\bar{\theta}_{\mathbb{L}}^\circ-\mathbb{L}[\Phi_n'\tilde{\beta}])$. Using $\mathbb{L}[\Phi_n'\tilde{\beta}]\geq\underline{\theta}_{\mathbb{L}}^*$ and  $\bar{\theta}_{\mathbb{L}}^*-\mathbb{L}[\Phi_n']\beta^*\leq r$ we get: 
	\[\bar{\theta}_{\mathbb{L}}^*-\bar{\theta}_{\mathbb{L}}^\circ\leq r+ \frac{\tilde{a}_n}{\epsilon-\kappa_n-\tilde{a}_n}(\bar{\theta}_{\mathbb{L}}^\circ-\underline{\theta}_{\mathbb{L}}^*)\]
	Again, this holds for each $r>0$, and so $\bar{\theta}_{\mathbb{L}}^*-\bar{\theta}_{\mathbb{L}}^\circ\leq \frac{\tilde{a}_n}{\epsilon-\kappa_n-\tilde{a}_n}(\bar{\theta}_{\mathbb{L}}^\circ-\underline{\theta}_{\mathbb{L}}^*)$. Since $\epsilon>\kappa_n+2\tilde{a}_n$, adding $ \frac{\tilde{a}_n}{\epsilon-\kappa_n-\tilde{a}_n}(\bar{\theta}_{\mathbb{L}}^*-\bar{\theta}_{\mathbb{L}}^\circ)$ to both sides and then dividing both sides by $1- \frac{a_n}{\epsilon-\kappa_n-\tilde{a}_n}$ we get:
	$
	\bar{\theta}_{\mathbb{L}}^*-\bar{\theta}_{\mathbb{L}}^\circ\leq \frac{\tilde{a}_n}{\epsilon-\kappa_n-2\tilde{a}_n}(\bar{\theta}_{\mathbb{L}}^*-\underline{\theta}_{\mathbb{L}}^*)
	$. 
	Combining with (\ref{oneside}) and using that $\kappa_n+2\tilde{a}_n\leq \epsilon/2$ we get:
	\begin{equation}
		|\bar{\theta}_{\mathbb{L}}^*-\bar{\theta}_{\mathbb{L}}^\circ| \leq \frac{2\tilde{a}_n}{\epsilon}(\bar{\theta}_{\mathbb{L}}^*-\underline{\theta}_{\mathbb{L}}^*)\leq \frac{2\tilde{a}_n}{\epsilon}(\bar{\theta}_{\mathbb{L}}-\underline{\theta}_{\mathbb{L}})\label{term2}
	\end{equation}
	Where the final inequality holds because the constraints of the problems for $\bar{\theta}_{\mathbb{L}}^*$ and $\underline{\theta}_{\mathbb{L}}^*$ are stronger than those for $\bar{\theta}_{\mathbb{L}}$ and $\underline{\theta}_{\mathbb{L}}$. 
	
	Next, we apply Lemma 2.3 to bound $|\bar{\theta}_{\mathbb{L}}^\circ-\hat{\bar{\theta}}_{\mathbb{L}}|$. By Assumption 2.5.i, $\Phi_{n}$ is Lipschitz continuous with constant
	$\xi_{n}$, and so the function $\Phi_{n}'\beta$ is Lipschitz continuous
	with constant at most $\|\beta\|_{2}\xi_{n}$. Therefore:
	\[
	|\Phi_{n}'\beta|_{\infty}-\max_{x\in\mathcal{X}_n}|\Phi_{n}(x)'\beta|\leq\|\beta\|_{2}\xi_{n}D_{1,n}
	\]
	It follows that:
	\[
	\max_{x\in\mathcal{X}_n}|\Phi_{n}(x)'\beta|\leq\bar{c}\implies|\Phi_{n}'\beta|_{\infty}\leq\bar{c}+\|\beta\|_{2}\xi_{n}D_{1,n}
	\]
	By the definition of $C_{n}$:
	\[
	|\Phi_{n}'\beta|_{\infty}\leq\bar{c}+\|\beta\|_{2}\xi_{n}D_{1,n}\implies\|\beta\|_{2}\leq C_{n}(\bar{c}+\|\beta\|_{2}\xi_{n}D_{1,n})
	\]
	So if $C_{n}\xi_{n}D_{1,n}<1$  we must have:
	\[
	\max_{x\in\mathcal{X}_n}|\Phi_{n}(x)'\beta|\leq\bar{c}\implies\|\beta\|_{2}\leq\frac{C_{n}}{1-C_{n}\xi_{n}D_{1,n}}\bar{c}\leq 2C_n\bar{c}
	\]
	Where we used that $C_n\xi_n D_{1,n}\leq 1/2$. By Assumption 2.2, for any $\beta$ that satisfies the constraints in the problem for $\hat{\bar{\theta}}_{\mathbb{L}}$ we have $\max_{x\in\mathcal{X}_n}|\Phi_{n}(x)'\beta|\leq\bar{c}$, and so for any such $\beta$, $\|\hat{\beta}\|_{2}\leq 2C_{n}\bar{c}$. By Assumption 2.5.i, each row of $\mathbb{T}[\Phi_{n}']$
	is Lipschitz with constant at most $\xi_{n}$ and so $\mathbb{T}[\Phi_{n}']\beta$ is Lipschitz with constant at most $\xi_n\|\beta\|$. In all, if $\beta$ satisfies the constraints for the problem for $\hat{\bar{\theta}}_{\mathbb{L}}$ then $\mathbb{T}[\Phi_{n}']\beta$ is Lipschitz with constant at most  $2\xi_n C_n\bar{c}$. Further, by Assumption 2.5.i $C$ is Lipschitz continuous with constant at most $\xi_n$. Note that both $2\xi_n C_n\bar{c}$ and $\xi$ are less than $2\xi_n (1+ C_n\bar{c})$
	
	Now, by Assumption 2.5.ii, $b_{1,k}$, $b_{2,k}$, $\hat{g}_{n,k}$, and $\hat{\Pi}_{n,k}$ are Lipschitz with constant at most $G_n$.
	So if $\beta$ satisfies the constraints for the problem for $\hat{\bar{\theta}}_{\mathbb{L}}$ then $\hat{\Pi}_{n,k}'\beta$ and $-\hat{\Pi}_{n,k}'\beta$ are both Lipschitz with constant at most $2G_n C_n\bar{c}$,
	and both $\hat{g}_{n,k}-b_{1,j}$ and $-\hat{g}_{n,k}+b_{2,j}$ are Lipschitz with constant at most $2G_n$. Note that $2G_n C_n\bar{c}$ and $2G_n$ are less than $2G_n (1+ C_n\bar{c})$.
	
	So applying Lemma 2.3 we get:
	\[|\bar{\theta}_{\mathbb{L}}^\circ-\hat{\bar{\theta}}_{\mathbb{L}}|\leq \frac{4(1+ C_n\bar{c})(d\xi_n D_{1,n}+KG_n D_{2,n})}{\epsilon-\kappa_n-\tilde{a}_n}(\bar{\theta}_{\mathbb{L}}^\circ-\underline{\theta}_{\mathbb{L}}^\circ)\]
	Now, recall that $
	|\bar{\theta}_{\mathbb{L}}^*-\bar{\theta}_{\mathbb{L}}^\circ|\leq \frac{\tilde{a}_n}{\epsilon-\kappa_n-2\tilde{a}_n}(\bar{\theta}_{\mathbb{L}}-\underline{\theta}_{\mathbb{L}})
	$. We can apply the same reasoning to get $
	|\underline{\theta}_{\mathbb{L}}^*-\underline{\theta}_{\mathbb{L}}^\circ|\leq \frac{\tilde{a}_n}{\epsilon-\kappa_n-2\tilde{a}_n}(\bar{\theta}_{\mathbb{L}}-\underline{\theta}_{\mathbb{L}})
	$. To see this note that we can just replace $\mathbb{L}$ with $-\mathbb{L}$ in the problems for $\bar{\theta}_{\mathbb{L}}$, $\bar{\theta}_{\mathbb{L}}^*$, and $\bar{\theta}_{\mathbb{L}}^\circ$ and then $-\underline{\theta}_{\mathbb{L}}$, $-\underline{\theta}_{\mathbb{L}}^*$, and $-\underline{\theta}_{\mathbb{L}}^\circ$ are the respective suprema of the new problems. Then by the triangle inequality:
	\[
	\bar{\theta}_{\mathbb{L}}^\circ-\underline{\theta}_{\mathbb{L}}^\circ\leq\bar{\theta}_{\mathbb{L}}^*-\underline{\theta}_{\mathbb{L}}^*+\frac{2\tilde{a}_n}{\epsilon-\kappa_n-2\tilde{a}_n}(\bar{\theta}_{\mathbb{L}}-\underline{\theta}_{\mathbb{L}})\leq\frac{\epsilon-\kappa_n}{\epsilon-\kappa_n-2\tilde{a}_n}(\bar{\theta}_{\mathbb{L}}-\underline{\theta}_{\mathbb{L}})
	\]
	Where the second inequality follows because the problems for $\bar{\theta}_{\mathbb{L}}$ and $\underline{\theta}_{\mathbb{L}}$ have weaker constraints than for $\bar{\theta}_{\mathbb{L}}^*$ and $\underline{\theta}_{\mathbb{L}}^*$ so $\bar{\theta}_{\mathbb{L}}^*-\underline{\theta}_{\mathbb{L}}^*\leq\bar{\theta}_{\mathbb{L}}-\underline{\theta}_{\mathbb{L}}$.
	Combining: 
	\begin{align*}
		|\bar{\theta}_{\mathbb{L}}^\circ-\hat{\bar{\theta}}_{\mathbb{L}}|&\leq \frac{8(\epsilon-\kappa_n)(1+ C_n\bar{c})(d \xi_n D_{1,n}+KG_n D_{2,n})}{(\epsilon-\kappa_n-\tilde{a}_n)(\epsilon-\kappa_n-2\tilde{a}_n)}(\bar{\theta}_{\mathbb{L}}-\underline{\theta}_{\mathbb{L}})\\
		&\leq \frac{32}{\epsilon}(1+ C_n\bar{c})(d \xi_n D_{1,n}+KG_n D_{2,n})(\bar{\theta}_{\mathbb{L}}-\underline{\theta}_{\mathbb{L}})
	\end{align*}
	Where the second inequality uses $\kappa_n+2\tilde{a}_n\leq \epsilon/2$. Combining the above with (\ref{term1}), (\ref{term2}), and (\ref{trianglee}) we get:
	\begin{align*}
		|\bar{\theta}_{\mathbb{L}}-\hat{\bar{\theta}}_{\mathbb{L}}|\leq &\kappa_n\big(\frac{1}{\epsilon}(\bar{\theta}_{\mathbb{L}}-\underline{\theta}_{\mathbb{L}})+c_{\mathbb{L}}\big)+  \frac{2\tilde{a}_n}{\epsilon}(\bar{\theta}_{\mathbb{L}}-\underline{\theta}_{\mathbb{L}})\\
		&+ \frac{32}{\epsilon}(1+ C_n\bar{c})(d \xi_n D_{1,n}+KG_n D_{2,n})(\bar{\theta}_{\mathbb{L}}-\underline{\theta}_{\mathbb{L}})
	\end{align*}
	Note that linearity and boundedness of $\mathbb{L}$ and Assumption 2.2 imply that $(\bar{\theta}_{\mathbb{L}}-\underline{\theta}_{\mathbb{L}})\leq 2c_\mathbb{L}\bar{c}$ so we get:
	\begin{align*}
		|\bar{\theta}_{\mathbb{L}}-\hat{\bar{\theta}}_{\mathbb{L}}|\leq &\kappa_n c_{\mathbb{L}}\big(\frac{2\bar{c}}{\epsilon}+1\big)+ c_{\mathbb{L}} \frac{4\tilde{a}_n\bar{c}}{\epsilon}\\
		&+ c_\mathbb{L}\frac{64 \bar{c}}{\epsilon}(1+ C_n\bar{c})(d \xi_n D_{1,n}+KG_n D_{2,n})
	\end{align*}
	Now, $-\underline{\theta}_{\mathbb{L}}$ and $-\hat{\underline{\theta}}_{\mathbb{L}}$ are the respective suprema of the problems for $\bar{\theta}_{\mathbb{L}}$ and $\hat{\bar{\theta}}_{\mathbb{L}}$ with $\mathbb{L}$ replaced by $-\mathbb{L}$, and so we can repeat the same steps that bound $|\bar{\theta}_{\mathbb{L}}-\hat{\bar{\theta}}_{\mathbb{L}}|$ to get:
	\begin{align*}
		|\underline{\theta}_{\mathbb{L}}-\hat{\underline{\theta}}_{\mathbb{L}}|\leq &\kappa_n c_{\mathbb{L}}\big(\frac{2\bar{c}}{\epsilon}+1\big)+ c_{\mathbb{L}} \frac{4\tilde{a}_n\bar{c}}{\epsilon}\\
		&+ c_\mathbb{L}\frac{64 \bar{c}}{\epsilon}(1+ C_n\bar{c})(d \xi_n D_{1,n}+KG_n D_{2,n})
	\end{align*}
	Finally, recall we derived the above under the assumption that $\kappa_n+2\tilde{a}_n\leq \epsilon/2$. By Assumptions 2.3 and 2.4 this holds with probability approaching $1$. Further, by Assumption 2.3 $\tilde{a}_n=O_p(a_n)$ and by construction $C_n\geq1$. So in all we get:
	\[
	|\bar{\theta}_{\mathbb{L}}-\hat{\bar{\theta}}_{\mathbb{L}}|=O_p(\kappa_n+ a_n+C_n\xi_n D_{1,n}+C_nG_n D_{2,n})
	\]
	And likewise for $|\underline{\theta}_{\mathbb{L}}-\underline{\bar{\theta}_{\mathbb{L}}}|$. Since the above is based on an inequality that depends on $\mathbb{L}$ only through $c_{\mathbb{L}}$, the convergence in probability above must is uniform over all $\mathbb{L}$ that have the same constant $c_{\mathbb{L}}$.
\end{proof}

\begin{proof}[Theorem 2.1]
	In this case $-b_1(Z)=b_2(Z)=b$ so both are Lipschitz continuous with constant $0$. The first statement of the theorem then follows immediately from Theorem 2.0. For the final statement of the theorem note that the evaluation functionals $h\mapsto h(x)$ all have operator norm of unity, because $|h(x)|\leq |h|_\infty$.
\end{proof}

For the proofs below, we refer to a functional $\gamma$ defined as follows. Let $
\hat{Q}_{n}=\frac{1}{n}\sum_{i=1}^{n}\Psi_{n}(Z_{i})\Psi_{n}(Z_{i})'
$, then for each $h\in\mathcal{H}$ let $\hat{\gamma}[h]$ be the least squares
estimator defined by $
\hat{\gamma}[h]=\hat{Q}_{n}^{-1}\frac{1}{n}\sum_{i=1}^{n}\Psi_{n}(Z_{i})h(X_{i})
$. 
If $\hat{Q}_{n}$ is singular we take $\hat{\gamma}[h]$ to be zero
(under Assumption 2.5.i this event happens with probability approaching
zero).
\newtheorem*{L2.3}{Lemma 2.3}
\begin{L2.3}
	
	Suppose Assumptions 2.1.ii, 2.2, 2.6 and 2.7 hold. Then:
	\begin{align*}
		\sup_{h\in\mathcal{H}}|\Psi_{n}'\hat{\gamma}[h]-\mathbb{A}[h]|_{\infty} & =O_{p}\bigg(\frac{\bar{\xi}_{n}}{\sqrt{n}}[1+\sqrt{log(\ell_{n})}+\sqrt{L_{n}}R_{n}(s)]+R_{n}(s)\bigg)\\
		& =o_{p}(1)
	\end{align*}
	
\end{L2.3}
\begin{proof}
	The proof follows some steps of \citet{Belloni2015} Lemma 4.2 with
	alterations to achieve uniformity over $\mathcal{H}$. Recall $Q_{n}=E[\Psi_{n}(Z_{i})\Psi_{n}(Z_{i})']$. By Assumption
	2.6.i we can normalize 
	$
	E[\Psi_{n}(Z)\Psi_{n}(Z)']=I
	$ and the assumptions still hold with sequences
	$\bar{\xi}_{n}$ and $\ell_{n}$ (that satisfy Assumption 2.6) changed
	only by a factor independent of $n$. We maintain
	this normalization throughout. Define $\gamma$ by 
	$
	\gamma[h]=E[\Psi_{n}(Z_{i})h(X_{i})]
	$, $\epsilon_{i}[h]=h(X_{i})-E[h(X_{i})|Z_{i}]$, and $r_{i}[h]=E[h(X_{i})|Z_{i}]-\gamma[h]'\Psi_{n}(Z_{i})$.
	Then:
	\[
	h(X_{i})=\Psi_{n}(Z_{i})'\gamma[h]+\epsilon_{i}[h]+r_{i}[h]
	\]
	By Assumption 2.2, $\mathcal{H}$ contains functions bounded
	by $\bar{c}$, so:
	\begin{align*}
		|\epsilon_{i}[h]|  \leq|h(X_{i})|+|E[h(X_{i})|Z_{i}]| \leq2\bar{c}
	\end{align*}
	Note as well that:
	\begin{align}
		|\epsilon_{i}[h_{1}]-\epsilon_{i}[h_{2}]| & \leq|h_{1}(X_{i})-h_{2}(X_{i})|+|E[h_{1}(X_{i})-h_{2}(X_{i})|Z_{i}]|\nonumber \\
		& \leq2|h_{1}-h_{2}|_{\infty}\label{eq:elip}
	\end{align}
	To bound $|r_{i}[h]|$ we first show that if $|h|_{\infty}<\infty$
	then Assumption 2.6.ii implies $\mathbb{A}[h]$ is smooth. For some $m\leq s$ (where $s$ is the smoothness in Assumption 2.6.ii)
	let $\{q_{j}\}_{j=1}^{\dim(Z)}$ be a sequence of positive integers
	with $\sum_{j=1}^{\dim(Z)}q_{j}=m$. Let $D_{q}$ then be the partial
	derivative operator given by $D_{q}[f](z)=\frac{\partial^{m}}{\partial^{q_{1}}\partial^{q_{2}}...\partial^{q_{\dim(Z)}}}f(z)$
	for any sufficiently differentiable function $f:\,\mathcal{Z}\to\mathbb{R}$.
	Now, from Assumption 2.6.ii, $D_{q}[f_{X|Z}(x,\cdot)](z)$ is bounded
	uniformly over $x$ and $z$ by some constant $\bar{\ell}$. Then
	it follows by the dominated convergence theorem that:
	\begin{align*}
		|D_{q}\mathbb{A}[h](z)| & =|D_{q}\big[\int h(x)f_{X|Z}(x,z)dx\big](z)|\\
		& =|\int h(x)D_{q}\big[f_{X|Z}(x,\cdot)\big](z)dx| \leq|h|_{\infty}\bar{\ell}
	\end{align*}
	And for $m=s$, note that:
	\begin{align*}
		& |D_{q}\mathbb{A}[h](z_{1})-D_{q}\mathbb{A}[h](z_{2})|\\
		= & |\int h(x)D_{q}\big[f_{X|Z}(x,\cdot)\big](z_{1})dx-\int h(x)D_{q}\big[f_{X|Z}(x,\cdot)\big](z_{2})dx|\\
		\leq & |h|_{\infty}\sup_{x\in\mathcal{X}}|D_{q}\big[f_{X|Z}(x,\cdot)\big](z_{1})-D_{q}\big[f_{X|Z}(x,\cdot)\big](z_{2})|\\
		\leq & |h|_{\infty}\bar{\ell}\|z_{1}-z_{2}\|_{2}^{s}
	\end{align*}
	
	Where the last inequality again follows because Assumption 2.6.ii
	states $f_{X|Z}(x,\cdot)$ is of H{\"o}lder smoothness class $s$with
	constant $\bar{\ell}$ for all $x$. From the above we see that $\mathbb{A}[h]$
	is of H{\"o}lder smoothness class $s$with constant at most $|h|_{\infty}\bar{\ell}$.
	So we can apply Assumption 2.7.i and get that $
	|r_{i}[h]|\leq|h|_{\infty}\bar{\ell}R_{n}(s)
	$. 
	Since $|h|_{\infty}\leq\bar{c}$ for all $h\in\mathcal{H}$ (from
	Assumption 2.6.iii), we get for any $h\in\mathcal{H}$, 
	$
	|r_{i}[h]|\leq\bar{c}\bar{\ell}R_{n}(s)
	$. 
	Further, note that by linearity of $r_{i}$:
	\begin{equation}
		|r_{i}[h_{1}]-r_{i}[h_{2}]|=|r_{i}[h_{1}-h_{2}]|\leq|h_{1}-h_{2}|_{\infty}\bar{\ell}R_{n}(s)\label{rlip}
	\end{equation}
	Using Assumption 2.7.iv we can apply Rudelson's matrix LLN (\citet{Rudelson1999},
	\citet{Belloni2015} Lemma 6.2) to get:
	\begin{align*}
		E\|\hat{Q}_{n}-I\|_{op} =O(\sqrt{\frac{\bar{\xi}_{n}^{2}log(L_{n})}{n}}) =o(1)
	\end{align*}
	Which implies $\|\hat{Q}_{n}^{-\frac{1}{2}}\|_{op}=O_{p}(1)$, $\|\hat{Q}_{n}^{\frac {1}{2}}\|_{op}=O_{p}(1)$  and $\|\hat{Q}_{n}^{-1}\|_{op}=O_{p}(1)$. And also by Rudelson's LLN , $E\|\hat{Q}_{n}^{\frac {1}{2}}\|_{op}=O(1)$.
	
	It follows that $\hat{Q}_{n}$ is non-singular with probability approaching
	$1$ and so below we treat it as non-singular.
	Recall we define $\alpha_{n}:\,\mathcal{Z}\to\mathbb{R}^{L_{n}}$
	by $
	\alpha_{n}(z)=\frac{\Psi_{n}(z)}{\|\Psi_{n}(z)\|}
	$.  And note that:
	\begin{align}
		\sqrt{n}\alpha_{n}(z)'(\hat{\gamma}[h]-\gamma[h]) & =\alpha_{n}(z)'\hat{Q}_{n}^{-1}\frac{1}{\sqrt{n}}\sum_{i=1}^{n}\Psi_{n}(Z_{i})\epsilon_{i}[h]\label{splitee}\\
		& +\alpha_{n}(z)'\hat{Q}_{n}^{-1}\frac{1}{\sqrt{n}}\sum_{i=1}^{n}\Psi_{n}(Z_{i})r_{i}[h]\nonumber
	\end{align}
	Let us bound the first term on the RHS above. Note that:
	\[
	E\bigg[\alpha_{n}(z)'\hat{Q}_{n}^{-1}\frac{1}{\sqrt{n}}\sum_{i=1}^{n}\Psi_{n}(Z_{i})\epsilon_{i}[h]\bigg|Z_{1},...,Z_{n}\bigg]=0
	\]
	Let $(\eta_{1},...,\eta_{n})$ be a sample of iid Rademachers, independent
	of the data. By the symmetrization inequality:
	\begin{align}
		& E\bigg[\sup_{z\in\mathcal{Z},h\in\mathcal{H}}|\alpha_{n}(z)'\hat{Q}_{n}^{-1}\frac{1}{\sqrt{n}}\sum_{i=1}^{n}\Psi_{n}(Z_{i})\epsilon_{i}[h]|\,\bigg|Z_{1},...,Z_{n}\bigg]\nonumber \\
		\leq & 2E\bigg[E_{\eta}\big[\sup_{z\in\mathcal{Z},h\in\mathcal{H}}|\alpha_{n}(z)'\hat{Q}_{n}^{-1}\frac{1}{\sqrt{n}}\sum_{i=1}^{n}\Psi_{n}(Z_{i})\eta_{i}\epsilon_{i}[h]|\big]\bigg|Z_{1},...,Z_{n}\bigg]\label{eq:symm}
	\end{align}
	Where the inner expectation on the RHS above is over the Rademachers
	with $\epsilon_{i}$ and $Z_{i}$ for $i=1,...,n$ treated as fixed.
	
	Define the set $\mathcal{T}\subseteq\mathbb{R}^{n}$ by:
	\[
	\mathcal{T}=[t=(t_{1},...,t_{n})\in\mathbb{R}^{n}:\,t_{i}=\alpha_{n}(z)'\hat{Q}_{n}^{-1}\Psi_{n}(Z_{i})\epsilon_{i}[h],z\in\mathcal{Z},h\in\mathcal{H}]
	\]
	
	Define a norm $\|\cdot\|_{n,2}$ on $\mathbb{R}^{n}$ by $\|t\|_{n,2}^{2}=\frac{1}{n}\sum_{i=1}^{n}t_{i}^{2}$.
	By \citet{Dudley1967} there exists a universal constant $D$ so that:
	\begin{align*}
		& E_{\eta}\bigg[\sup_{z\in\mathcal{Z},h\in\mathcal{H}}|\alpha_{n}(z)'\hat{Q}_{n}^{-1}\frac{1}{\sqrt{n}}\sum_{i=1}^{n}\Psi_{n}(Z_{i})(\eta_{i}\epsilon_{i}[h])|\bigg]\\
		\leq & D\int_{0}^{\theta}\sqrt{log\mathcal{N}(\mathcal{T},\|\cdot\|_{n,2},\delta)}d\delta
	\end{align*}
	
	Where $\mathcal{N}(\mathcal{T},\|\cdot\|_{n,2},\delta)$ is the smallest
	number of radius-$\delta$ $\|\cdot\|_{n,2}$ balls needed to cover
	$\mathcal{T}$ and $\theta$ is the smallest upper bound on the $\|\cdot\|_{n,2}$-distance
	between any two points in $\mathcal{T}$. Using our bounds on $|\epsilon_{i}[h]|$:
	\[
	\theta=2\sup_{t\in\mathcal{T}}\|t\|_{n,2}\leq4\bar{c}\|\hat{Q}_{n}^{-\frac{1}{2}}\|_{op}
	\]
	Let $\tilde{A}_{i}[h]=\hat{Q}_{n}^{-1}\Psi_{n}(Z_{i})\epsilon_{i}[h]$.
	Using the Lipschitz constant for $\alpha_{n}$ given in Assumption
	2.7.ii, and using \ref{eq:elip} we get (for $h_{1},h_{2}\in\mathcal{H}$):
	\begin{align}
		& \bigg(\frac{1}{n}\sum_{i=1}^{n}|\alpha_{n}(z_{1})'\tilde{A}_{i}[h_{1}]-\alpha_{n}(z_{2})'\tilde{A}_{i}[h_{2}]|^{2}\bigg)^{\frac{1}{2}}\nonumber \\
		\leq & \bigg((\bar{c}\ell_{n})^{2}\|z_{1}-z_{2}\|_{2}^{2}+|h_{1}-h_{2}|_{\infty}^{2}\bigg)^{\frac{1}{2}}4\sqrt{2}\|\hat{Q}_{n}^{-\frac{1}{2}}\|_{op}\label{lipeqfull}
	\end{align}
	
	Define $b_{n}$ by $b_{n}=4\|\hat{Q}_{n}^{-\frac{1}{2}}\|_{op}$. Note that $\|\hat{Q}_{n}^{-\frac{1}{2}}\|_{op}=O_{p}(1)$ and so $b_{n}=O_{p}(1)$. Then from \ref{lipeqfull}:
	\[
	\mathcal{N}(\mathcal{T},\|\cdot\|_{n,2},\delta)\leq\mathcal{N}(\mathcal{Z},\|\cdot\|_{2},\frac{\delta}{\bar{c}\ell_{n}b_{n}})\mathcal{N}(\mathcal{H},|\cdot|_{\infty},\frac{\delta}{b_{n}})
	\]
	And so (using sub-additivity of the square root):
	\begin{align*}
		\int_{0}^{\theta}\sqrt{log\mathcal{N}(\mathcal{T},\|\cdot\|_{n,2},\delta)}d\delta & \leq\int_{0}^{\bar{c}b_{n}}\sqrt{log\mathcal{N}(\mathcal{Z},\|\cdot\|_{2},\frac{\delta}{\bar{c}\ell_{n}b_{n}})}d\delta\\
		& +\int_{0}^{\bar{c}b_{n}}\sqrt{log\mathcal{N}(\mathcal{H},|\cdot|_{\infty},\frac{\delta}{b_{n}})}d\delta
	\end{align*}
	
	Making the substitution $u=\frac{\delta}{\bar{c}b_{n}}$ into the
	first integral and $u=\frac{\delta}{b_{n}}$ into the second:
	\begin{align*}
		\int_{0}^{\theta}\sqrt{log\mathcal{N}(\mathcal{T},\|\cdot\|_{n,2},\delta)}d\delta & \leq\bar{c}b_{n}\int_{0}^{1}\sqrt{log\mathcal{N}(\mathcal{Z},\|\cdot\|_{2},\frac{u}{\ell_{n}})}du\\
		& +b_{n}\int_{0}^{\bar{c}}\sqrt{log\mathcal{N}(\mathcal{H},|\cdot|_{\infty},u)}du
	\end{align*}
	
	Because the integrand is decreasing in $u$:
	\[
	\int_{0}^{\bar{c}}\sqrt{log\mathcal{N}(\mathcal{H},|\cdot|_{\infty},u)}du\leq\max\{\bar{c},1\}\int_{0}^{1}\sqrt{log\mathcal{N}(\mathcal{H},|\cdot|_{\infty},u)}du
	\]
	
	By Assumption 2.7.iii the integral on the RHS above is finite. So
	let $\omega_{1}$ denote the finite constant on the RHS above.
	
	By Assumption 2.6.ii $\mathcal{Z}$ is bounded and it has dimension
	$\dim(Z)<\infty$ therefore for some constant $\omega_{2}$, $
	\mathcal{N}(\mathcal{Z},\|\cdot\|_{2},\delta)\leq\omega_{2}\bigg(\frac{1}{\delta}\bigg)^{\dim(Z)}
	$. 
	So we have:
	\begin{align*}
		\int_{0}^{\theta}\sqrt{log\mathcal{N}(\mathcal{T},\|\cdot\|_{n,2},\delta)}d\delta & \leq\bar{c}b_{n}\int_{0}^{1}\sqrt{log\omega_{2}+\dim(Z)log\bigg(\frac{\ell_{n}}{u}\bigg)}du\\
		& +b_{n}\omega_{1}
	\end{align*}
	Using sub-additivity of the square root we get from the above:
	\begin{align*}
		&\int_{0}^{\theta}\sqrt{log\mathcal{N}(\mathcal{T},\|\cdot\|_{n,2},\delta)}d\delta \\
		\leq & \bar{c}b_{n}\sqrt{log\omega_{2}}+b_{n}\omega_{1}\\
		+& \bar{c}b_{n}\sqrt{\dim(Z)}\bigg(\sqrt{log(\ell_{n})}+\int_{0}^{1}\sqrt{log\bigg(\frac{1}{u}\bigg)}du\bigg)\\
		=& O_{p}\bigg(1+\sqrt{log(\ell_{n})}\bigg)
	\end{align*}
	And so from \ref{eq:symm} and Markov's inequality:
	\[
	\sup_{z\in\mathcal{Z},h\in\mathcal{H}}|\alpha_{n}(z)'\hat{Q}_{n}^{-1}\frac{1}{\sqrt{n}}\sum_{i=1}^{n}\Psi_{n}(Z_{i})\epsilon_{i}[h]|=O_{p}\bigg(1+\sqrt{log(\ell_{n})}\bigg)
	\]
	
	Now we bound the second term on the RHS in \ref{splitee}. Define $\mathcal{S}^{L_{n}-1}=[\beta\in\mathbb{R}^{L_{n}}:\,\|\beta\|_{2}\leq1]$.
	Note that:
	\begin{align*}
		&|\alpha_{n}(z)'\hat{Q}_{n}^{-1}\frac{1}{\sqrt{n}}\sum_{i=1}^{n}\Psi_{n}(Z_{i})r_{i}[h]|\\
		\leq & \|\hat{Q}_{n}^{-1}\|_{op}\sup_{\beta\in\mathcal{S}^{L_{n}-1}}|\beta'\frac{1}{\sqrt{n}}\sum_{i=1}^{n}\Psi_{n}(Z_{i})r_{i}[h]|
	\end{align*}
	And we already have by the matrix LLN that $\|\hat{Q}_{n}^{-1}\|_{op}=O_{p}(1)$.
	Note as well that:
	$
	E\bigg[\beta'\frac{1}{\sqrt{n}}\sum_{i=1}^{n}\Psi_{n}(Z_{i})r_{i}[h]\bigg]=0
	$. 
	Again, let $(\eta_{1},...,\eta_{n})$ be a sample of iid Rademachers, independent
	of the data. By the symmetrization inequality:
	\begin{align}
		& E\bigg[\sup_{h\in\mathcal{H},\beta\in\mathcal{S}^{L_{n}-1}}|\beta'\frac{1}{\sqrt{n}}\sum_{i=1}^{n}\Psi_{n}(Z_{i})r_{i}[h]|\bigg]\nonumber \\
		\leq & 2E\bigg[E_{\eta}\big[\sup_{h\in\mathcal{H},\beta\in\mathcal{S}^{L_{n}-1}}|\beta'\frac{1}{\sqrt{n}}\sum_{i=1}^{n}\Psi_{n}(Z_{i})r_{i}[h]|\big]\bigg]\label{eq:symm-1}
	\end{align}
	Where the inner expectation on the RHS above is over the Rademachers
	with $r_{i}$ and $Z_{i}$ for $i=1,...,n$ treated as fixed. Define
	a new set $\mathcal{T}\subseteq\mathbb{R}^{n}$ by:
	\[
	\mathcal{T}=[t=(t_{1},...,t_{n})\in\mathbb{R}^{n}:\,t_{i}=\beta'\psi_{n,l}(Z_{i})r_{i}[h],h\in\mathcal{H},\beta\in\mathcal{S}^{L_{n}-1}]
	\]
	Again, by \cite{Dudley1967}:
	\begin{align*}
		& E_{\eta}\bigg[\sup_{h\in\mathcal{H},\beta\in\mathcal{S}^{L_{n}-1}}|\beta'\frac{1}{\sqrt{n}}\sum_{i=1}^{n}\Psi_{n}(Z_{i})(\eta_{i}r_{i}[h])|\bigg]\\
		\leq & D\int_{0}^{\theta}\sqrt{log\mathcal{N}(\mathcal{T},\|\cdot\|_{n,2},\delta)}d\delta
	\end{align*}
	Using our bound on $|r_{i}[h]|$:
	$
	\theta=2\sup_{t\in\mathcal{T}}\|t\|_{n,2}\leq2\bar{c}\bar{\ell}R_{n}(s)\|\hat{Q}_{n}^{\frac{1}{2}}\|_{op}
	$. 
	Using \ref{rlip} we get (for $h_{1},h_{2}\in\mathcal{H}$):
	\begin{align*}
		& \bigg(\frac{1}{n}\sum_{i=1}^{n}|\Psi_{n}(Z_{i})r_{i}[h_{1}]-\Psi_{n}(Z_{i})r_{i}[h_{2}]|^{2}\bigg)^{\frac{1}{2}}\\
		\leq &\sqrt{2}\bar{\ell}R_{n}(s)\|\hat{Q}_{n}^{\frac{1}{2}}\|_{op}\bigg(\bar{c}^{2}\|\beta_{1}-\beta_{2}\|_{2}^{2}+|h_{1}-h_{2}|_{\infty}^{2}\bigg)^{\frac{1}{2}}
	\end{align*}
	Define $c_{n}$ by $c_{n}=2\bar{\ell}R_{n}(s)\|\hat{Q}_{n}^{\frac{1}{2}}\|_{op}$. Then from \ref{lipeqfull}:
	\[
	\mathcal{N}(\mathcal{T},\|\cdot\|_{n,2},\delta)\leq\mathcal{N}(\mathcal{S}^{L_{n}-1},\|\cdot\|_{2},\frac{\delta}{\bar{c}c_{n}})\mathcal{N}(\mathcal{H},|\cdot|_{\infty},\frac{\delta}{b_{n}})
	\]
	And so (using sub-additivity of the square root) and making substitutions, much as before:
	\begin{align*}
		\int_{0}^{\theta}\sqrt{log\mathcal{N}(\mathcal{T},\|\cdot\|_{n,2},\delta)}d\delta & \leq\bar{c}c_{n}\int_{0}^{1}\sqrt{log\mathcal{N}(\mathcal{S}^{L_{n}-1},\|\cdot\|_{2},u)}du\\
		& +c_{n}\int_{0}^{\bar{c}}\sqrt{log\mathcal{N}(\mathcal{H},|\cdot|_{\infty},u)}du
	\end{align*}
	We have already shown that that the second integral on the RHS above
	is bounded by a finite constant $\omega_{1}$. The covering number
	of a unit ball in $\mathbb{R}^{L_{n}}$ satisfies for some universal constant $\omega_{3}>0$:
	\[
	\mathcal{N}(\mathcal{S}^{L_{n}-1},\|\cdot\|_{2},\delta)=\omega_{3}\bigg(\frac{1}{\delta}\bigg)^{L_{n}}
	\]
	Substituting this and using sub-additivity of the square root we get:
	\begin{align*}
		\int_{0}^{\theta}\sqrt{log\mathcal{N}(\mathcal{T},\|\cdot\|_{n,2},\delta)}d\delta & \leq\bar{c}c_{n}\sqrt{log\omega_{3}}+c_{n}\omega_{1}\\
		& +\bar{c}c_{n}\sqrt{L_{n}}\int_{0}^{1}\sqrt{log\bigg(\frac{1}{u}\bigg)}du
	\end{align*}
	And so:
	\begin{align*}
		& E\bigg[\sup_{h\in\mathcal{H},\beta\in\mathcal{S}^{L_{n}-1}}|\beta'\frac{1}{\sqrt{n}}\sum_{i=1}^{n}\Psi_{n}(Z_{i})r_{i}[h]|\bigg]\nonumber \\
		\leq&\bar{c}E[c_{n}]\sqrt{L_{n}}\int_{0}^{1}\sqrt{log\bigg(\frac{1}{u}\bigg)}du+\bar{c}E[c_{n}]\sqrt{L_{n}}\int_{0}^{1}\sqrt{log\bigg(\frac{1}{u}\bigg)}du\\
		=& O\big(R_{n}(s)\sqrt{L_{n}}\big)
	\end{align*}
	And so, by Markov's inequality the second term in the RHS of \ref{splitee} is $ O_p \big(R_{n}(s)\sqrt{L_{n}}\big)$. So in all:
	\begin{align*}
		\sup_{h\in\mathcal{H}}|\Psi_{n}'\hat{\gamma}[h]-\mathbb{A}[h]|_{\infty} & \leq\frac{\bar{\xi}_{n}}{\sqrt{n}}\sup_{z\in\mathcal{Z},h\in\mathcal{H}}|\sqrt{n}\alpha_{n}(z)'(\hat{\gamma}[h]-\gamma[h])|\\
		& +\sup_{h\in\mathcal{H}}|\Psi_{n}'\gamma[h]-\mathbb{A}[h]|_{\infty}\\
		& =O_{p}\bigg(\frac{\bar{\xi}_{n}}{\sqrt{n}}[1+\sqrt{log(\ell_{n})}+\sqrt{L_{n}}R_{n}(s)]+R_{n}(s)\bigg)
	\end{align*}
	
\end{proof}

\begin{proof}[Theorem 2.2]
	
	From Lemma 2.3:
	\begin{align*}
		\sup_{h\in\mathcal{H}}|\Psi_{n}'\hat{\gamma}[h]-\mathbb{A}[h]|_{\infty} & =O_{p}\bigg(\frac{\bar{\xi}_{n}}{\sqrt{n}}[1+\sqrt{log(\ell_{n})}+\sqrt{L_{n}}R_{n}(s)]+R_{n}(s)\bigg)
	\end{align*}
	
	Note that if $\Phi_{n}'\beta\in\mathcal{H}$ then $\hat{\Pi}_{n}'\beta=\Psi_{n}'\hat{\gamma}[\Phi_{n}'\beta]$
	and by definition $\Pi_{n}'\beta=\mathbb{A}[\Phi_{n}'\beta]$, and so:
	\[
	\sup_{\beta\in\mathbb{R}^{K_{n}}:\,\Phi_{n}'\beta\in\mathcal{H}}|(\hat{\Pi}_{n}-\Pi_{n})'\beta|_{\infty}\leq\sup_{h\in\mathcal{H}}|\Psi_{n}'\hat{\gamma}[h]-\mathbb{A}[h]|_{\infty}
	\]
	
	Applying the rates in Assumption 2.8.ii:
	\[
	\frac{\bar{\xi}_{n}}{\sqrt{n}}[1+\sqrt{log(\ell_{n})}+\sqrt{L_{n}}R_{n}(s)]+R_{n}=O\bigg(\frac{\sqrt{L_{n}}}{\sqrt{n}}\sqrt{log(L_{n})}+L_{n}^{-s_{0}(s)/\dim(Z)}]\bigg)
	\]
	
	Note that Assumption 2.6.i is identical to Assumption A.2 in \citet{Belloni2015},
	Assumption 2.7.i implies Assumption A.3 in \citet{Belloni2015} (with
	`$\ell_{k}c_{k}$' in their notation equal to $R_{n}(s)$), Assumption
	2.8.i implies Assumption A.4 in \citet{Belloni2015} and Assumption
	2.8.i and 2.7.ii imply Assumption A.5. Assumption A.1 in \citet{Belloni2015},
	that the data are iid, is assumed in our paper throughout. Therefore
	we can apply \citet{Belloni2015} Theorem 4.3 for $\hat{g}_{n}$,
	and under our other assumptions the rate simplifies to:
	\[
	|\hat{g}_{n}-g_{0}|_{\infty}=O\bigg(\frac{\sqrt{L_{n}}}{\sqrt{n}}\sqrt{log(L_{n})}+L_{n}^{-s_{0}(s)/\dim(Z)}]\bigg)
	\]
	Then the triangle inequality gives the result.
\end{proof}

\begin{proof}[Theorem 2.3]
	We note two facts about B-spline basis functions on an interval. Firstly,
	if $K_{n}\geq2$ then we can apply a linear transformation to $\Phi_{n}$
	so that the first two entries are $1$ and $x$. We will assume without
	loss of generality that $\Phi_{n}$ has been transformed in this way. Secondly, because the basis functions are at least third order, the
	vector of functions $\Phi_{n}$ is at least twice continuously differentiable
	at all but a finite set of points in its domain. For any $z$ at which
	the second derivatives are defined, let $\frac{\partial^{2}}{\partial x^{2}}\Phi_{n}(z)$
	denote the vector of second derivatives of each component of $\Phi_{n}$
	at $z$. We can define this function elsewhere by right-continuity,
	that is if the second derivatives are not defined for some $z$ then
	let $\frac{\partial^{2}}{\partial x^{2}}\Phi_{n}(z)=\lim_{z'\downarrow z}\frac{\partial^{2}}{\partial x^{2}}\Phi_{n}(z')$.
	$\frac{\partial^{2}}{\partial x^{2}}\Phi_{n}$ is an invertible linear
	transformation of the vector of $(s_{0}-2)^{th}$-order B-spline basis
	functions with the same knot points as the original spline basis.
	It then follows from the approximation properties of splines that
	there is a sequence $\tilde{\kappa}_{n}=O(K_{n}^{-1})$ so that for
	any Lipschitz continuous function $h$ defined on $[a,b]$ with Lipschitz
	constant $L$, there exists $\beta\in\mathbb{R}^{K_{n}}$ with $
	|\frac{\partial^{2}}{\partial x^{2}}\Phi_{n}'\beta-h|_{\infty}\leq L\tilde{\kappa}_{n}
	$. 
	See for example \citet{DeVore1993}.
	
	First we show that any $h\in\mathcal{H}$ is Lipschitz continuous
	with constant at most $C=\frac{2}{(b-a)}+c(b-a)$.
	Consider some $h\in\mathcal{H}$, since $h$ is twice differentiable
	its first derivative also exists and so for any $x_{1},x_{2}\in[a,b]$:
	\[
	h(x_{2})-h(x_{1})=\int_{0}^{1}(x_{2}-x_{1})\frac{\partial}{\partial x}h(x_{1}+t(x_{2}-x_{1}))dt
	\]
	Because $|\frac{\partial^{2}}{\partial x^{2}}h|_{\infty}\leq c$ the above implies:
	\[
	|h(x_{2})-h(x_{1})|+c|x_{2}-x_{1}|^{2}\geq |(x_{2}-x_{1})\frac{\partial}{\partial x}h(x_{1})|
	\]
	Substituting $x_{1}=\frac{1}{2}(a+b)$ and $x_{2}=b$ into the above we get:
	\[
	|h(x_{2})-h(x_{1})|+c\frac{1}{4}(b-a)^{2}\geq\frac{1}{2}(b-a)|\frac{\partial}{\partial x}h(x_{1})|
	\]
	Since $h(x)\in[0,1]$, $|h(x_{1})-h(x_{2})|\leq1$ and so we get:
	\[
	|\frac{\partial}{\partial x}h(x_{1})|\leq\frac{2}{(b-a)}+c\frac{1}{2}(b-a)
	\]
	And again, since $|\frac{\partial^{2}}{\partial x^{2}}h|_{\infty}\leq c$
	we then have that for any $x\in[a,b]$:
	\[
	|\frac{\partial}{\partial x}h(x)|\leq\frac{2}{(b-a)}+c(b-a)
	\]
	Denote the RHS by $C$. We thus have that any $h\in\mathcal{H}$ is
	Lipschitz continuous with constant at most $C$.
	
	Let the functional $P$ extend a function $h\in\mathcal{H}$ to a function defined on $\mathbb{R}$ as follows: $
	P[h](x)=
	h(x)  \text{ if }x\in[a,b]$, $P[h](x)=
	h(a) \text{ if }x<a $, $P[h](x)=
	h(b)  \text{ if }x>b$.
	For $r\in[0,b-a]$ define the linear operator $M_{r}:\,\mathcal{B}_{X}\to\mathcal{B}_{X}$
	by (for all $x\in[a,b]$) $
	M_{r}[h](x)=\frac{\int_{x-r}^{x+r}P[h](y)dy}{2r}
	$. 
	If $|h|_{\infty}<\infty$ then it is easy to see that $
	|M_{r}[h]|_{\infty}\leq|h|_{\infty}
	$. 
	With a substitution we get $
	M_{r}[h](x)=\frac{\int_{-r}^{r}P[h](y+x)dy}{2r}
	$. 
	For any $h\in\mathcal{H}$, $\frac{\partial^{2}}{\partial x^{2}}h$
	is uniformly bounded by $c$. So we can use the dominated convergence theorem to get that for any $h\in\mathcal{H}$:
	\begin{align*}
		\frac{\partial^{2}}{\partial x^{2}}M_{r}[h](x)  =&\frac{1}{2r}\int_{-r}^{r}\frac{\partial^{2}}{\partial x^{2}}P[h](y+x)dy=\frac{1}{2r}\int_{\max\{a,x-r\}}^{\min\{b,x+r\}}\frac{\partial^{2}}{\partial y^{2}}h(y)dy
	\end{align*}
	
	Where we have used that the second derivative of $P[h]$ is zero outside of $[a,b]$. So for any $h\in\mathcal{H}$, $
	|\frac{\partial^{2}}{\partial x^{2}}M_{r}[h](x)|\leq c
	$.
	
	Assuming without loss of generality that $x_{1}\geq x_{2}$, we see
	for any $h\in\mathcal{H}$:
	\begin{align*}
		&|\frac{\partial^{2}}{\partial x^{2}}M_{r}[h](x_{1})-\frac{\partial^{2}}{\partial x^{2}}M_{r}[h](x_{2}) |\\
		=& \frac{1}{2r}\bigg|\int_{\max\{x_{1}-r,x_{2}+r\}}^{x_{1}+r}\frac{\partial^{2}}{\partial y^{2}}P[h](y)dy-\int_{x_{2}-r}^{\min\{x_{2}+r,x_{1}-r\}}\frac{\partial^{2}}{\partial y^{2}}P[h](y)dy\bigg|\\
		\leq& \frac{c}{r}|x_{1}-x_{2}|
	\end{align*}
	Where the last inequality follows because $\frac{\partial^{2}}{\partial x^{2}}P[h]$
	is uniformly bounded by $c$ for $h\in\mathcal{H}$. So we have established that for any $h\in\mathcal{H}$ the function
	$M_{r}[h]$ has second derivatives that are Lipschitz continuous with
	Lipschitz constant at most $\frac{c}{r}$.
	
	And now note that, because $h\in\mathcal{H}$ is Lipschitz continuous
	with constant $C$ for any $x\in[a,b]$, $
	|M_{r}[h](x)-h(x)|\leq rC
	$. 
	
	Now, recall the properties of B-splines discussed at the beginning of this proof. Because $\frac{\partial^{2}}{\partial x^{2}}M_{r}[h]$ is Lipschitz
	continuous with constant $\frac{c}{r}$, there is some $\beta\in\mathbb{R}^{K_{n}}$
	with $
	|\frac{\partial^{2}}{\partial x^{2}}\Phi_{n}'\beta-\frac{\partial^{2}}{\partial x^{2}}M_{r}[h]|_{\infty}\leq\frac{c}{r}\tilde{\kappa}_{n}
	$. 
	In which case $|\frac{\partial^{2}}{\partial x^{2}}\Phi_{n}'\beta|_{\infty}\leq c(1+\frac{\tilde{\kappa}_n}{r})$. And so letting $\tilde{\beta}=(1+\frac{\tilde{\kappa}_{n}}{r})^{-1}\beta$ we get $|\frac{\partial^{2}}{\partial x^{2}}\Phi_{n}'\tilde{\beta}|_{\infty}\leq c$, and from the triangle inequality:
	\begin{align*}
		|\frac{\partial^{2}}{\partial x^{2}}\Phi_{n}'\tilde{\beta}-\frac{\partial^{2}}{\partial x^{2}}M_{r}[h]|_{\infty}\leq&\frac{c}{r}\tilde{\kappa}_{n}+\bigg(\frac{\tilde{\kappa}_n/r}{1+(\tilde{\kappa}_n/r)}\bigg)|\frac{\partial^{2}}{\partial x^{2}}\Phi_{n}'\beta|_{\infty}
	\end{align*}
	Which is less than $2\frac{c}{r}\tilde{\kappa}_{n}$, and so:
	\begin{align*}
		|\frac{\partial}{\partial x}\Phi_{n}'\tilde{\beta}-\frac{\partial}{\partial x}M_{r}[h]|_{\infty} & \leq\bigg|\frac{\partial}{\partial x}\Phi_{n}(\frac{a+b}{2})'\tilde{\beta}-\frac{\partial}{\partial x}M_{r}[h](\frac{a+b}{2})\bigg|\\
		& +(b-a)\frac{c}{r}\tilde{\kappa}_{n}
	\end{align*}
	Where again $\frac{\partial}{\partial x}\Phi_{n}$ is defined at points
	at which the derivative is undefined by right continuity. The inequality above then implies:
	\begin{align}
		|\Phi_{n}'\tilde{\beta}-M_{r}[h]|_{\infty} & \leq|\Phi_{n}(\frac{a+b}{2})'\tilde{\beta}-M_{r}[h](\frac{a+b}{2})|\label{eqqqu}\\
		& +\frac{b-a}{2}\bigg|\frac{\partial}{\partial x}\Phi_{n}(\frac{a+b}{2})'\tilde{\beta}-\frac{\partial}{\partial x}M_{r}[h](\frac{a+b}{2})\bigg|\nonumber\\
		& +\bigg(\frac{b-a}{2}\bigg)^{2}2\frac{c}{r}\tilde{\kappa}_{n}\nonumber
	\end{align}
	But the first two entries of $\Phi_{n}(x)$ are $1$ and $x$, so
	let $\beta^{*}$ be identical to $\tilde{\beta}$ aside from its first
	two entries $\beta_{1}^{*}$ and $\beta_{2}^{*}$ which are given
	by:
	\[
	\beta_{2}^{*}=\frac{\partial}{\partial x}M_{r}[h](\frac{a+b}{2})-\frac{\partial}{\partial x}\Phi_{n}(\frac{a+b}{2})'\tilde{\beta}
	\]
	And:
	\[
	\beta_{1}^{*}=M_{r}[h](\frac{a+b}{2})-\Phi_{n}(\frac{a+b}{2})'\tilde{\beta}-(\frac{a+b}{2})\beta_{2}^{*}
	\]
	Then $\frac{\partial^{2}}{\partial x^{2}}\Phi_{n}'\beta^{*}=\frac{\partial^{2}}{\partial x^{2}}\Phi_{n}'\tilde{\beta}$, and so $|\frac{\partial^{2}}{\partial x^{2}}\Phi_{n}'\beta^{*}|_{\infty}\leq c$. Moreover, repeating the same steps used to get (\ref{eqqqu}) above we see: 
	\[
	|\Phi_{n}'\beta^{*}-M_{r}[h]|_{\infty}\leq\bigg(\frac{b-a}{2}\bigg)^{2}2\frac{c}{r}\tilde{\kappa}_{n}
	\]
	We already showed $|M_{r}[h](x)-h(x)|\leq rC$ and so
	\[
	|\Phi_{n}'\beta^{*}-h|_{\infty}\leq\bigg(\frac{b-a}{2}\bigg)^{2}2\frac{c}{r}\tilde{\kappa}_{n}+rC
	\]
	And so setting $r=\sqrt{\tilde{\kappa}_{n}}$:
	\[
	|\Phi_{n}'\beta^{*}-h|_{\infty}\leq\bigg[\bigg(\frac{b-a}{2}\bigg)^{2}2c+C\bigg]\sqrt{\tilde{\kappa}_{n}}
	\]
	Since we found such a $\beta^{*}$ for any $h$ the result follows. \end{proof}

\end{document}